\documentclass[envcountsect,envcountsame,runningheads]{llncs}

\usepackage{xspace}
\usepackage[all]{xy}
\usepackage{wrapfig}

\usepackage{amsmath}
\usepackage{amssymb}
\usepackage{mathrsfs}
\usepackage{xcolor}
\usepackage{mathtools}

\newdimen\proofrulebreadth \proofrulebreadth=.05em
\newdimen\proofdotseparation \proofdotseparation=1.25ex
\newdimen\proofrulebaseline \proofrulebaseline=2ex
\newcount\proofdotnumber \proofdotnumber=3
\let\then\relax
\def\hfi{\hskip0pt plus.0001fil}
\mathchardef\squigto="3A3B
\newif\ifinsideprooftree\insideprooftreefalse
\newif\ifonleftofproofrule\onleftofproofrulefalse
\newif\ifproofdots\proofdotsfalse
\newif\ifdoubleproof\doubleprooffalse
\let\wereinproofbit\relax
\newdimen\shortenproofleft
\newdimen\shortenproofright
\newdimen\proofbelowshift
\newbox\proofabove
\newbox\proofbelow
\newbox\proofrulename
\def\shiftproofbelow{\let\next\relax\afterassignment\setshiftproofbelow\dimen0 }
\def\shiftproofbelowneg{\def\next{\multiply\dimen0 by-1 }%
\afterassignment\setshiftproofbelow\dimen0 }
\def\setshiftproofbelow{\next\proofbelowshift=\dimen0 }
\def\setproofrulebreadth{\proofrulebreadth}

\def\prooftree{
\ifnum  \lastpenalty=1
\then   \unpenalty
\else   \onleftofproofrulefalse
\fi
\ifonleftofproofrule
\else   \ifinsideprooftree
        \then   \hskip.5em plus1fil
        \fi
\fi
\bgroup
\setbox\proofbelow=\hbox{}\setbox\proofrulename=\hbox{}%
\let\justifies\proofover\let\leadsto\proofoverdots\let\Justifies\proofoverdbl
\let\using\proofusing\let\[\prooftree
\ifinsideprooftree\let\]\endprooftree\fi
\proofdotsfalse\doubleprooffalse
\let\thickness\setproofrulebreadth
\let\shiftright\shiftproofbelow \let\shift\shiftproofbelow
\let\shiftleft\shiftproofbelowneg
\let\ifwasinsideprooftree\ifinsideprooftree
\insideprooftreetrue
\setbox\proofabove=\hbox\bgroup$\displaystyle 
\let\wereinproofbit\prooftree
\shortenproofleft=0pt \shortenproofright=0pt \proofbelowshift=0pt
\onleftofproofruletrue\penalty1
}

\def\eproofbit{
\ifx    \wereinproofbit\prooftree
\then   \ifcase \lastpenalty
        \then   \shortenproofright=0pt  
        \or     \unpenalty\hfil         
        \or     \unpenalty\unskip       
        \else   \shortenproofright=0pt  
        \fi
\fi
\global\dimen0=\shortenproofleft
\global\dimen1=\shortenproofright
\global\dimen2=\proofrulebreadth
\global\dimen3=\proofbelowshift
\global\dimen4=\proofdotseparation
\global\count255=\proofdotnumber
$\egroup  
\shortenproofleft=\dimen0
\shortenproofright=\dimen1
\proofrulebreadth=\dimen2
\proofbelowshift=\dimen3
\proofdotseparation=\dimen4
\proofdotnumber=\count255
}

\def\proofover{
\eproofbit 
\setbox\proofbelow=\hbox\bgroup 
\let\wereinproofbit\proofover
$\displaystyle
}%
\def\proofoverdbl{
\eproofbit 
\doubleprooftrue
\setbox\proofbelow=\hbox\bgroup 
\let\wereinproofbit\proofoverdbl
$\displaystyle
}%
\def\proofoverdots{
\eproofbit 
\proofdotstrue
\setbox\proofbelow=\hbox\bgroup 
\let\wereinproofbit\proofoverdots
$\displaystyle
}%
\def\proofusing{
\eproofbit 
\setbox\proofrulename=\hbox\bgroup 
\let\wereinproofbit\proofusing
\kern0.3em$
}

\def\endprooftree{
\eproofbit 
  \dimen5 =0pt
\dimen0=\wd\proofabove \advance\dimen0-\shortenproofleft
\advance\dimen0-\shortenproofright
\dimen1=.5\dimen0 \advance\dimen1-.5\wd\proofbelow
\dimen4=\dimen1
\advance\dimen1\proofbelowshift \advance\dimen4-\proofbelowshift
\ifdim  \dimen1<0pt
\then   \advance\shortenproofleft\dimen1
        \advance\dimen0-\dimen1
        \dimen1=0pt
        \ifdim  \shortenproofleft<0pt
        \then   \setbox\proofabove=\hbox{%
                        \kern-\shortenproofleft\unhbox\proofabove}%
                \shortenproofleft=0pt
        \fi
\fi
\ifdim  \dimen4<0pt
\then   \advance\shortenproofright\dimen4
        \advance\dimen0-\dimen4
        \dimen4=0pt
\fi
\ifdim  \shortenproofright<\wd\proofrulename
\then   \shortenproofright=\wd\proofrulename
\fi
\dimen2=\shortenproofleft \advance\dimen2 by\dimen1
\dimen3=\shortenproofright\advance\dimen3 by\dimen4
\ifproofdots
\then
        \dimen6=\shortenproofleft \advance\dimen6 .5\dimen0
        \setbox1=\vbox to\proofdotseparation{\vss\hbox{$\cdot$}\vss}%
        \setbox0=\hbox{%
                \advance\dimen6-.5\wd1
                \kern\dimen6
                $\vcenter to\proofdotnumber\proofdotseparation
                        {\leaders\box1\vfill}$%
                \unhbox\proofrulename}%
\else   \dimen6=\fontdimen22\the\textfont2 
        \dimen7=\dimen6
        \advance\dimen6by.5\proofrulebreadth
        \advance\dimen7by-.5\proofrulebreadth
        \setbox0=\hbox{%
                \kern\shortenproofleft
                \ifdoubleproof
                \then   \hbox to\dimen0{%
                        $\mathsurround0pt\mathord=\mkern-6mu%
                        \cleaders\hbox{$\mkern-2mu=\mkern-2mu$}\hfill
                        \mkern-6mu\mathord=$}%
                \else   \vrule height\dimen6 depth-\dimen7 width\dimen0
                \fi
                \unhbox\proofrulename}%
        \ht0=\dimen6 \dp0=-\dimen7
\fi
\let\doll\relax
\ifwasinsideprooftree
\then   \let\VBOX\vbox
\else   \ifmmode\else$\let\doll=$\fi
        \let\VBOX\vcenter
\fi
\VBOX   {\baselineskip\proofrulebaseline \lineskip.2ex
        \expandafter\lineskiplimit\ifproofdots0ex\else-0.6ex\fi
        \hbox   spread\dimen5   {\hfi\unhbox\proofabove\hfi}%
        \hbox{\box0}%
        \hbox   {\kern\dimen2 \box\proofbelow}}\doll%
\global\dimen2=\dimen2
\global\dimen3=\dimen3
\egroup 
\ifonleftofproofrule
\then   \shortenproofleft=\dimen2
\fi
\shortenproofright=\dimen3
\onleftofproofrulefalse
\ifinsideprooftree
\then   \hskip.5em plus 1fil \penalty2
\fi
}

\newcommand{\indrulename}[1]{\texttt{\textup{#1}}}
\newcommand{\indrule}[3]{
\ensuremath{\begin{array}{c}
  \prooftree #2
    \justifies #3
    \thickness=0.05em
    \using \indrulename{\scriptsize{#1}}
  \endprooftree
\end{array}}}

\renewcommand{\theenumi}{\arabic{enumi}}
\renewcommand{\theenumii}{\arabic{enumii}}
\renewcommand{\theenumiii}{\arabic{enumiii}}

\makeatletter
\renewcommand\p@enumii{\theenumi.}
\renewcommand\p@enumiii{\theenumi.\theenumii.}
\renewcommand\p@enumiv{\theenumi.\theenumii.\theenumiii.}
\makeatother

\newcommand{\llem}[1]{\label{lemma:#1}}
\newcommand{\rlem}[1]{Lem.~\ref{lemma:#1}}
\newcommand{\ldef}[1]{\label{def:#1}}
\newcommand{\rdef}[1]{Def.~\ref{def:#1}}
\newcommand{\lprop}[1]{\label{prop:#1}}
\newcommand{\rprop}[1]{Prop.~\ref{prop:#1}}
\newcommand{\lthm}[1]{\label{thm:#1}}
\newcommand{\rthm}[1]{Thm.~\ref{thm:#1}}

\newcommand{\lsec}[1]{\label{section:#1}}
\newcommand{\rsec}[1]{Sec.~\ref{section:#1}}
\newcommand{\lexample}[1]{\label{example:#1}}
\newcommand{\rexample}[1]{Ex.~\ref{example:#1}}

\newcommand{\eg}{{\em e.g.}\xspace}
\newcommand{\ie}{{\em i.e.}\xspace}
\newcommand{\ih}{{\em i.h.}\xspace}
\newcommand{\etal}{et al.\xspace}
\newcommand{\cf}{{\em cf.}\xspace}
\newcommand{\ST}{\ |\ }
\newcommand{\HS}{\hspace{.5cm}}

\renewcommand{\emptyset}{\varnothing}
\newcommand{\set}[1]{\{#1\}}
\newcommand{\Nat}{\mathtt{Nat}}
\newcommand{\NN}{\mathbb{N}}
\newcommand{\eqdef}{\overset{\mathrm{def}}{=}}

\newcommand{\FIT}[1]{\resizebox{\hsize}{!}{#1}}

\newcommand{\Var}{\mathsf{Var}}
\newcommand{\Cons}{\mathsf{Con}}
\newcommand{\Val}{\mathsf{Val}}
\newcommand{\Type}{\mathsf{Type}}

\newcommand{\Loc}{\mathsf{Loc}}
\newcommand{\Asg}{\mathsf{Env}}

\newcommand{\var}{x}
\newcommand{\vartwo}{y}
\newcommand{\varthree}{z}

\newcommand{\cons}{{\bf c}}
\newcommand{\constwo}{{\bf d}}
\newcommand{\consthree}{{\bf e}}

\newcommand{\tm}{t}
\newcommand{\tmtwo}{s}
\newcommand{\tmthree}{u}
\newcommand{\tmfour}{r}

\newcommand{\unit}{{\bf ok}}
\newcommand{\fail}{\mathtt{fail}}
\newcommand{\FAIL}{\bot}
\newcommand{\unif}{\overset{\bullet}{=}}

\newcommand{\seq}{;}
\newcommand{\alt}{\oplus}

\newcommand{\lam}[2]{\lambda #1.\,#2}
\newcommand{\laml}[3]{\lambda^{#1} #2.\,#3}
\newcommand{\fresh}[2]{\nu #1.\,#2}
\newcommand{\sub}[2]{\{#1:=#2\}} 

\DeclareMathOperator{\supp}{supp}

\newcommand{\val}{\mathtt{v}}
\newcommand{\valtwo}{\mathtt{w}}

\newcommand{\goals}{\mathsf{G}}
\newcommand{\goalstwo}{\mathsf{H}}

\newcommand{\fv}[1]{\mathsf{fv}(#1)}
\newcommand{\locs}[1]{\mathsf{locs}(#1)}

\newcommand{\ctxof}[1]{\langle#1\rangle}
\newcommand{\ctxhole}{\Box}
\newcommand{\SEP}{\mid\!\mid}
\newcommand{\gctx}{\mathsf{C}}
\newcommand{\gctxof}[1]{\gctx\ctxof{#1}}
\newcommand{\wctx}{\mathsf{W}}
\newcommand{\wctxof}[1]{\wctx\ctxof{#1}}

\newcommand{\structeq}{\equiv}
\newcommand{\permeq}{\sim}

\newcommand{\rulename}[1]{\texttt{\textup{#1}}}

\newcommand{\tounifa}[1]{\mathrel{\rightsquigarrow_{\rulename{#1}}}}
\newcommand{\toca}[1]{
  \def\tmp{#1}\ifx\tmp\empty
    \rightarrow%
  \else
    \xrightarrow{\rulename{#1}}%
  \fi}
\newcommand{\rtoca}[1]{
  \def\tmp{#1}\ifx\tmp\empty
    \twoheadrightarrow%
  \else
    \mathrel{\xrightarrow{\rulename{#1}}\!\!\!\!\!\xrightarrow{}}%
  \fi}
\newcommand{\topar}[1]{
  \def\tmp{#1}\ifx\tmp\empty
    \Rightarrow
  \else
    \xRightarrow{#1\,}%
  \fi
}

\newcommand{\prog}{P}
\newcommand{\progtwo}{Q}
\newcommand{\progthree}{R}

\newcommand{\loc}{\ell}
\newcommand{\loctwo}{\loc'}
\newcommand{\locthree}{\loc''}

\newcommand{\SUB}[1]{{}^{#1}}
\newcommand{\subst}{\sigma}
\newcommand{\substtwo}{\rho}
\newcommand{\substthree}{\tau}
\newcommand{\substa}{\alpha}
\newcommand{\scomp}{\cdot}
\newcommand{\sleq}{\lesssim}
\newcommand{\mgu}[1]{\mathsf{mgu}(#1)}
\newcommand{\bigalt}{\bigoplus}

\newcommand{\nprog}{\prog^\star}
\newcommand{\nprogtwo}{\progtwo^\star}
\newcommand{\ntm}{\tm^\star}
\newcommand{\ntmtwo}{\tmtwo^\star}
\newcommand{\stm}{S}
\newcommand{\stmtwo}{S'}
\newcommand{\stuck}{\bigtriangledown}

\newcommand{\btyp}{\alpha}
\newcommand{\btyptwo}{\beta}
\newcommand{\btypthree}{\gamma}

\newcommand{\constyp}[1]{\mathcal{T}_{#1}}

\newcommand{\typ}{A}
\newcommand{\typtwo}{B}
\newcommand{\typthree}{C}

\newcommand{\emptyctx}{\emptyset}
\newcommand{\tctx}{\Gamma}
\newcommand{\tctxtwo}{\Delta}

\newcommand{\fctx}{\Phi}
\newcommand{\fctxtwo}{\Psi}

\newcommand{\semtyp}[1]{[\![#1]\!]}
\newcommand{\seme}[2]{[\![#1]\!]_{#2}}
\newcommand{\semf}[3][\fctx]{[\![#2]\!]^{#1}_{#3}}

\newcommand{\powerset}[1]{\mathcal{P}(#1)}
\newcommand{\binterp}[1]{\mathsf{S}_{#1}}
\newcommand{\cinterp}[1]{\underline{#1}}
\newcommand{\asg}{\rho}
\newcommand{\asgextend}[2]{[#1\mapsto#2]}

\newcommand{\obj}{a}
\newcommand{\objtwo}{b}
\newcommand{\objthree}{c}
\newcommand{\objfour}{d}
\newcommand{\objfun}{f}
\newcommand{\objfuntwo}{g}

\newcommand{\altt}{\boxplus}

\newcommand{\lambdaunif}{\lambda^{\mathtt{U}}}

\newcommand{\codesym}{\mathtt{C}}
\newcommand{\allocsym}{\mathtt{A}}

\colorlet{darkgreen}{green!60!black}
\newcommand{\SeeAppendix}{\textcolor{darkgreen}{$\clubsuit$}}
\newcommand{\SeeAppendixRef}[1]{\textcolor{darkgreen}{$\clubsuit$\,#1}}
\newcommand{\WithProofs}[1]{\textcolor{darkgreen}{#1}}

\newcommand{\pairing}{\mathbf{t}}
\newcommand{\walg}[1]{\mathbb{W}[#1]}

\usepackage{hyperref}
\usepackage{chngcntr}
\usepackage{apptools}
\AtAppendix{\counterwithin{theorem}{subsection}}

\begin{document}
\title{Semantics of a Relational $\lambda$-Calculus\thanks{Work partially supported by
project grants ECOS~Sud~A17C01, PUNQ~1346/17, and UBACyT~20020170100086BA.}\\(Extended Version)}
%
%
\author{
Pablo Barenbaum\inst{1}
  \and
Federico Lochbaum\inst{2} 
  \and
Mariana Milicich\inst{3}
}
\authorrunning{P.~Barenbaum, F.~Lochbaum, M.~Milicich}
%

\institute{
Universidad de Buenos Aires and Universidad Nacional de Quilmes, Argentina \\
\email{pbarenbaum@dc.uba.ar}
\and
Universidad Nacional de Quilmes, Argentina \\
\email{federico.lochbaum@gmail.com}
\and
Universidad de Buenos Aires, Argentina \\
\email{milicichmariana@gmail.com}
}
\maketitle              
\begin{abstract}
We extend the $\lambda$-calculus with constructs suitable
for relational and functional--logic programming: non-deterministic
choice, fresh variable introduction, and unification of expressions. 
In order to be able to unify $\lambda$-expressions and still obtain a
confluent theory, we depart from related approaches, such as
$\lambda$Prolog, in that we do not attempt to solve
higher-order unification.
Instead, abstractions are decorated with a {\em location},
which intuitively may be understood as its memory address,
and we impose a simple {\em coherence} invariant: abstractions in the same location
must be equal. This allows us to formulate a {\em confluent} small-step operational
semantics which only performs first-order unification
and does not require {\em strong} evaluation (below lambdas).
We study a simply typed version of the system.
Moreover, a denotational semantics for the calculus is proposed
and reduction is shown to be sound with respect to the denotational
semantics.

\end{abstract}
\section{Introduction}
  \lsec{introduction}
  
Declarative programming is defined by the ideal that programs should resemble
abstract specifications rather than concrete implementations.
One of the most significant declarative paradigms is
{\em functional programming},
represented by languages such as Haskell.
Some of its salient features are the presence of
first-class functions and inductive datatypes manipulated through pattern matching.
The fact that the underlying model of computation---the $\lambda$-calculus---is
{\em confluent} allows one to reason equationally about the behavior of
functional programs.

Another declarative paradigm is {\em logic programming}, represented by
languages such as Prolog. Some of its salient features are the ability to define
relations rather than functions, and
the presence of existentially quantified {\em symbolic variables} that become
instantiated upon querying.
This sometimes allows to use $n$-ary relations with various patterns of
instantiation, \eg \texttt{add(3, 2, X)} computes $\texttt{X := 3 + 2}$ whereas
\texttt{add(X, 2, 5)} computes $\texttt{X := 5 - 2}$.
The underlying model of computation is based on {\em unification} and
refutation search with {\em backtracking}.

The idea to marry functional and logic programming has been around for a long time,
and there have been many attempts to combine their features gracefully.
For example, $\lambda$Prolog (Miller and Nadathur~\cite{nadathur1984higher,miller2012programming})
takes Prolog as a starting point,
generalizing first-order terms to {\em $\lambda$-terms} and
the mechanism of first-order unification to that of
{\em higher-order unification}.
Another example is Curry (Hanus et al.~\cite{DBLP:conf/popl/Hanus97,Hanus13})
in which programs are defined by equations, quite like in functional languages,
but evaluation is non-deterministic and evaluation is based on {\em narrowing},
\ie variables become instantiated in such a way as to fulfill the
constraints imposed by equations.

One of the interests of combining functional and logic programming is the fact
that the increased expressivity aids declarative programming. For instance,
if one writes a parser as a function
\texttt{parser : String $\longrightarrow$ AST},
it should be possible, under the right conditions,
to {\em invert} this function to obtain a
pretty-printer \texttt{pprint : AST $\longrightarrow$ String}:
\[
{\small
\texttt{pprint ast = $\nu$ source . ((ast $\unif$ parse source) ; source)}
}
\]
In this hypothetical functional--logic language, intuitively speaking,
the expression $(\fresh{\var}{\tm})$
creates a fresh symbolic variable $\var$ and proceeds to evaluate~$\tm$;
the expression $(\tm \unif \tmtwo)$ unifies $\tm$ with $\tmtwo$;
and the expression $(\tm \seq \tmtwo)$ returns the result of evaluating
$\tmtwo$ whenever the evaluation of $\tm$ succeeds.

Given that unification is a generalization of pattern matching,
a functional language with explicit unification should in some sense generalize
{\em $\lambda$-calculi with patterns},
such as the Pure Pattern Calculus~\cite{jay2006pure}.
For example, by relying on unification one may build
{\em dynamic} or {\em functional patterns}, \ie patterns that include operations
other than constructors. A typical instance is the following function
\texttt{last : [a] $\longrightarrow$ a},
which returns the last element of a non-empty cons-list:
\[
{\small
\texttt{last (xs ++ [x]) = x}
}
\]
Note that \verb|++| is not a constructor.
This definition may be desugared similarly as for the \texttt{pprint}
example above:
\[
{\small
\texttt{last lst = $\nu$ xs . $\nu$ x. (lst $\unif$ (xs ++ [x])); x}
}
\]

Still another interest comes from the point of view of the
{\em proposition-as-types correspondence}.
Terms of a $\lambda$-calculus with types can be understood as encoding
proofs, so for instance the identity function $(\lam{\var:\typ}{\var})$
may be understood as a proof of the implication $\typ \to \typ$.
From this point of view, a functional--logic program may be understood as a
{\em tactic},
as can be found in proof assistants such as
Isabelle or Coq~(see \eg~\cite{Coq}).
A term of type $\typ$ should then be understood as a non-deterministic
procedure which attempts to find a proof of $\typ$
and it may leave holes in the proof or even fail.
For instance if $P$ is a property on natural numbers,
$p$ is a proof of $P(0)$ and $q$ is a proof of $P(1)$,
then $\lam{n}{((n \unif 0) \seq p) \altt ((n \unif 1)\seq q)}$
is a tactic that given a natural number $n$
produces a proof of $P(n)$ whenever $n \in \set{0,1}$,
and otherwise it fails. Here $(\tm \altt \tmtwo)$ denotes
the {\em non-deterministic alternative} between $\tm$ and $\tmtwo$.

\medskip
The goal of this paper is to {\bf provide a foundation for functional--logic
programming} by {\bf extending the $\lambda$-calculus with relational constructs}.
Recall that the syntactic elements of the $\lambda$-calculus are
$\lambda$-terms ($\tm,\tmtwo,\hdots$),
which inductively may be {\em variables}~($\var,\vartwo,\hdots)$,
{\em abstractions}~($\lam{\var}{\tm}$), and {\em applications}~($\tm\,\tmtwo$).
Relational programming may be understood as the purest form of logic
programming, chiefly represented by the family of miniKanren
languages~(Byrd et al.~\cite{friedman_reasoned_2005,byrd2010relational}).
The core syntactic elements of miniKanren,
following for instance Rozplokhas et~al.~\cite{rozplokhas2019certified}
are goals ($G,G',\hdots$)
which are inductively given by:
{\em relation symbol invocations}, of the form $R(T_1,\hdots,T_n)$,
where $R$ is a relation symbol and $T_1,\hdots,T_n$ are terms of
a first-order language,
{\em unification} of first-order terms ($T_1 \unif T_2$),
{\em conjunction} of goals ($G \seq G'$),
{\em disjunction} of goals ($G \altt G'$),
and {\em fresh variable introduction} ($\fresh{\var}{G}$).

Our starting point is a ``chimeric creature''---a functional--logic
language resulting from cross breeding the $\lambda$-calculus and miniKanren,
given by the following abstract syntax:
  \vspace{-.25cm}
\[
  \begin{array}{rr@{\hspace{.25cm}}l@{\HS}l@{}rr@{\hspace{.25cm}}l@{\HS}l}
  \tm,\tmtwo
      & ::=  & \var              & \text{variable} &
      & \mid & \cons             & \text{constructor} \\
      & \mid & \lam{\var}{\tm}   & \text{abstraction} &
      & \mid & \tm\,\tmtwo       & \text{application} \\
      & \mid & \fresh{\var}{\tm} & \text{fresh variable introduction} &
      & \mid & \tm \altt \tmtwo  & \text{non-deterministic choice} \\
      & \mid & \tm \seq \tmtwo   & \text{guarded expression} &
      & \mid & \tm \unif \tmtwo  & \text{unification} \\
  \end{array}
\]
Its informal semantics has been described above.
Variables ($\var,\vartwo,\hdots$) may be instantiated by unification,
while constructors ($\cons,\constwo,\hdots$) are constants.
For example, if $\texttt{coin} \eqdef ({\bf true} \altt {\bf false})$
is a non-deterministic boolean with two possible values and
$\texttt{not} \eqdef \lam{\var}{((\var \unif {\bf true}) \seq {\bf false})\altt((\var \unif {\bf false}) \seq {\bf true})}$
is the usual boolean negation, the following non-deterministic computation:
\[
  (\lam{\var}{\lam{\vartwo}{(\var \unif \texttt{not}\,\vartwo)\seq {\bf pair}\ \var\,\vartwo}})\,\texttt{coin}\,\texttt{coin}
\]
should have two results,
namely ${\bf pair}\,{\bf true}\,{\bf false}$
and ${\bf pair}\,{\bf false}\,{\bf true}$.

\medskip
{\bf Structure of this paper.}
In Section~2, we discuss some technical difficulties that arise as one
intends to provide a formal operational semantics for the informal
functional--logic calculus sketched above.
In Section~3, we refine this rough proposal into a calculus we call
the {\em $\lambdaunif$-calculus}, with a
formal small-step operational semantics~(\rdef{operational_semantics_reduction_rules}).
To do so, we distinguish {\em terms},
which represent a single choice, from {\em programs}, which represent
a non-deterministic alternative between zero or more terms.
Moreover, we adapt the standard first-order unification algorithm to our
setting by imposing a {\em coherence} invariant on programs.
In Section~4, we study the operational properties of the
$\lambdaunif$-calculus: we provide an
inductive characterization of the set of normal forms~(\rprop{characterization_normal_forms}),
and we prove that it is confluent~(\rthm{confluence}) (up to a notion of structural equivalence).
In Section~5, we propose a straightforward system of simple types
and we show that it enjoys subject reduction~(\rprop{subject_reduction}).
In Section~6, we define a (naive) denotational semantics,
and we show that the operational semantics is sound (although it is not complete)
with respect to this denotational semantics~(\rthm{soundness}).
In Section~7, we conclude and we lay out avenues of further research.

\medskip
\noindent \WithProofs{{\bf Note.}
Most proofs have been left out from the body of the paper.
Detailed proofs that can be found in the technical appendix have been
marked with \SeeAppendix.}

\section{Technical Challenges}
  \lsec{technical_challenges}
  
This section is devoted to discussing technical stumbling blocks that
we encountered as we attempted to define an operational semantics
for the functional--logic
calculus incorporating all the constructs mentioned in the introduction.
These technical issues motivate the design decisions behind the
actual $\lambdaunif$-calculus defined in \rsec{lambdaunif_calculus}.
The discussion in this section is thus {\bf informal}.
Examples are carried out with their hypothetical or intended semantics.
\medskip

\noindent {\bf Locality of symbolic variables.}
The following program introduces a fresh variable $\var$ and then there
are two alternatives:
either $\var$ unifies with $\cons$ and the result is $\var$,
or $\var$ unifies with $\constwo$ and the result is $\var$.
The expected reduction semantics is the following.
The constant $\unit$ is the result obtained after a successful unification:
\[
  \begin{array}{rcll}
  \fresh{\var}{\left({((\var \unif \cons) \seq \var) \altt ((\var \unif \constwo) \seq \var)}\right)}
  & \to &
  ((\var \unif \cons) \seq \var) \altt ((\var \unif \constwo) \seq \var)
  & \text{with $\var$ fresh}
  \\
  & \to &
  (\unit \seq \cons) \altt ((\var \unif \constwo) \seq \var)
  & (\bigstar)
  \\
  & \to &
  (\unit \seq \cons) \altt (\unit \seq \constwo)
  \\
  & \twoheadrightarrow &
  \cons \altt \constwo
  \end{array}
\]
Note that in the step marked with $(\bigstar)$,
the variable $\var$ becomes instantiated to $\cons$,
but {\em only to the left of the choice operator} ($\altt$).
This suggests that programs should consist of different {\em threads}
fenced by choice operators.
Symbolic variables should be local to each thread.
\medskip

\noindent {\bf Need of commutative conversions.}
Redexes may be {\em blocked} by the choice operator---for example
in the application $((\tm \altt \lam{\var}{\tmtwo})\,\tmthree)$,
there is a potential $\beta$-redex $((\lam{\var}{\tmtwo})\,\tmthree)$
which is blocked. This suggests that {\em commutative conversions}
that distribute the choice operator should be incorporated,
allowing for instance a reduction step
$(\tm \altt \lam{\var}{\tmtwo})\,\tmthree \to \tm\,\tmthree \altt (\lam{\var}{\tmtwo})\,\tmthree$.
In our proposal, we force in the syntax that a program is always written,
canonically, in the form $\tm_1 \altt \hdots \altt \tm_n$,
where each $\tm_i$ is a deterministic program (\ie choice
operators may only appear inside lambdas).
This avoids the need to introduce commutative rules.
\medskip

\noindent {\bf Confluence only holds up to associativity and commutativity.}
There are two ways to distribute the choice operators in the following example:
\[\FIT{$
  \xymatrix@C=.1cm@R=.5cm{
    \tm_1(\tmtwo_1 \altt \tmtwo_2) \altt \tm_2(\tmtwo_1 \altt \tmtwo_2)
    \ar[d]
  &
    \ar[l]
    (\tm_1 \altt \tm_2)\,(\tmtwo_1 \altt \tmtwo_2)
    \ar[r]
  &
    (\tm_1 \altt \tm_2)\,\tmtwo_1 \altt (\tm_1 \altt \tm_2)\,\tmtwo_2
    \ar[d]
  \\
    (\tm_1\,\tmtwo_1 \altt \tm_1\,\tmtwo_2) \altt (\tm_2\,\tmtwo_1 \altt \tm_2\,\tmtwo_2)
    \ar@{.}[r]
  &
    \equiv
    \ar@{.}[r]
  &
    (\tm_1\,\tmtwo_1 \altt \tm_2\,\tmtwo_1) \altt (\tm_1\,\tmtwo_2 \altt \tm_2\,\tmtwo_2)
  }
$}\]
The resulting programs cannot be equated unless one works up to an equivalence
relation that takes into account the associativity and commutativity of the choice operator.
As we mentioned, the $\lambdaunif$-calculus works with programs in canonical
form $\tm_1 \altt \hdots \altt \tm_n$, so there is no need to work modulo associativity.
However, we do need commutativity. As a matter of fact, we shall define a notion of
{\em structural equivalence} ($\structeq$) between programs, allowing
the arbitrary reordering of threads. This relation will be shown to be well-behaved,
namely, a {\em strong bisimulation} with respect to the reduction relation,
\cf~\rlem{reduction_modulo_structeq}.
\medskip

\noindent{\bf Non-deterministic choice is an effect.}
Consider the program $(\lam{\var}{\var\,\var})(\cons\altt\constwo)$,
which chooses between $\cons$ and $\constwo$ and then it produces two
copies of the chosen value. Its expected reduction semantics is:
\[
  (\lam{\var}{\var\,\var})(\cons\altt\constwo)
  \to
  (\lam{\var}{\var\,\var})\cons \altt (\lam{\var}{\var\,\var})\constwo
  \twoheadrightarrow
  \cons\,\cons \altt \constwo\,\constwo
\]
This means that the first step in the following reduction, which
produces two copies of $(\cons\altt\constwo)$ cannot be
allowed, as it would break confluence:
\[
  (\lam{\var}{\var\,\var})(\cons\altt\constwo) \not\to (\cons\altt\constwo)\,(\cons\altt\constwo)
  \twoheadrightarrow
  \cons\,\cons \altt \cons\,\constwo \altt \constwo\,\cons \altt \constwo\,\constwo
\]
The deeper reason is that non-deterministic choice is a side-effect rather
than a value. Our design decision, consistent with this remark,
is to follow a {\bf call-by-value} discipline.
Another consequence of this remark is that the choice operator should
not commute with abstraction, given that $\lam{\var}{(\tm\altt\tmtwo)}$
and $(\lam{\var}{\tm})\altt(\lam{\var}{\tmtwo})$ are not observationally
equivalent. In particular, $\lam{\var}{(\tm\altt\tmtwo)}$ is a value, which
may be {\em copied},
while $(\lam{\var}{\tm})\altt(\lam{\var}{\tmtwo})$ is not a value.
On the other hand, if 
$\wctx$ is any {\em weak} context, \ie a term with a hole which does
{\em not} lie below a binder, and we write $\wctxof{\tm}$ for the result of plugging a
term $\tm$ into the hole of $\wctx$, then
$\wctxof{\tm \altt \tmtwo} = \wctxof{\tm} \altt \wctxof{\tmtwo}$
should hold.
\medskip

\noindent {\bf Evaluation should be weak.}
Consider the term $F \eqdef \lam{\vartwo}{((\vartwo \unif \var)\seq \var})$.
Intuitively, it unifies its argument with a (global)
symbolic variable $\var$ and then returns $\var$.
This poses two problems. First, when $\var$ becomes instantiated to $\vartwo$,
it may be outside the scope of the abstraction binding $\vartwo$,
for instance, the step
$F\,x = 
(\lam{\vartwo}{((\vartwo \unif \var)\seq \var}))\,\var
\to (\lam{\vartwo}{(\unit\seq\vartwo)})\,\vartwo$
produces a meaningless free occurrence of $\vartwo$.
Second, consider the following example in which two copies of $F$
are used with different arguments. If we do {\bf not} allow evaluation
under lambdas, this example fails due to a unification clash,
\ie it produces no outputs:
\[
  \begin{array}{rcll}
  (\lam{f}{(f\,\cons)\,(f\,\constwo)})\,F
  & \to &
  (F\,\cons)\,(F\,\constwo)
  \\
  & \to &
  ((\cons\unif\var)\seq\var)\,((\constwo\unif\var)\seq\var)
  \\
  & \to &
  (\unit\seq\cons)\,((\constwo\unif\cons)\seq\cons)
  & (\bigstar)
  \\
  & \to &
  \fail
  \end{array}
\]
Note that in the step marked with $(\bigstar)$, the symbolic variable $\var$
has become instantiated to $\cons$, leaving us with the unification goal
$\constwo \unif \cons$ which fails.
On the other hand, if we were to allow reduction under lambdas,
given that there are no other occurrences of $\var$ anywhere in the
term, in one step $F$ becomes $\lam{\vartwo}{(\unit\seq\vartwo)}$,
which then behaves as the identity:
\[
  \begin{array}{rcl}
    (\lam{f}{(f\,\cons)\,(f\,\constwo)})\,F
  & \not\to &
    (\lam{f}{(f\,\cons)\,(f\,\constwo)})\,(\lam{\vartwo}{\unit\seq\vartwo})
  \\
  & \to &
    ((\lam{\vartwo}{\unit\seq\vartwo})\,\cons)\,((\lam{\vartwo}{\unit\seq\vartwo})\,\constwo)
  \\
  & \twoheadrightarrow &
    \cons\,\constwo
  \end{array}
\]
Thus allowing reduction below abstractions in this example
would break confluence.
This suggests that evaluation should be {\bf weak},
\ie it should not proceed below binders.
\medskip

\noindent {\bf Avoiding higher-order unification.}
The calculus proposed in this paper rests on the design choice to
{\em avoid} attempting to solve higher-order unification problems.
Higher-order unification problems can be expressed in the syntax:
for example in $(f \cons \unif \cons)$ the variable $f$
represents an unknown value which should fulfill the given constraint.
From our point of view, however, this program is stuck and its evaluation
cannot proceed---it is a normal form.
However, note that we {\bf do} want to allow {\em pattern matching} against
functions; for example the following should succeed, instantiating
$f$ to the identity:
\[
  (\cons\,f \unif \cons(\lam{\var}{\var}))\seq(f \unif f)
  \to
  (\lam{\var}{\var}) \unif (\lam{\var}{\var})
  \to
  \unit
\]
The decision to sidestep higher-order unification is a debatable one,
as it severely restricts the expressivity of the language.
But there are various reasons to explore alternatives.
First, higher-order unification is undecidable~\cite{huet1973undecidability},
and even second order unification is known to be undecidable~\cite{levy2000undecidability}.
Huet's semi-decision procedure~\cite{huet1975unification} does find a
solution should it exist, but even then higher-order unification problems do
not necessarily possess {\em most general unifiers}~\cite{gould66},
which turns confluence hopeless\footnote{Key in our proof of confluence
is the fact that if $\subst$ and $\subst'$
are most general unifiers for unification problems $\goals$ and $\goals'$
respectively, then the most general unifier for $(\goals \cup \goals')$
is an instance of both $\subst$ and $\subst'$. See~\rexample{confluence_unif_unif}.}.
Second, there are decidable restrictions of higher-order unification which
do have most general unifiers, such as
{\em higher-order pattern unification}~\cite{miller1991unification}
used in $\lambda$Prolog, and {\em nominal unification}~\cite{urban2004nominal}
used in $\alpha$Prolog. But these mechanisms require {\em strong} evaluation,
\ie evaluation below abstractions, departing from the traditional execution
model of eager applicative languages such as in the Lisp and ML families,
in which closures are opaque values whose bodies cannot be examined.
Moreover, they are formulated in a necessarily typed setting.

The calculus studied in this paper relies on a standard first-order unification
algorithm, with the only exception that abstractions are deemed to be equal
if and only if they have the same ``identity''. Intuitively speaking,
this means that they are stored in the same memory location, \ie they are
represented by the same pointer.
This is compatible with the usual implementation techniques
of eager applicative languages, so it should allow to use standard compilation
techniques for $\lambda$-abstractions. Also note that
the operational semantics does not require to work with typed
terms---in fact the system presented in \rsec{lambdaunif_calculus} is
untyped, even though we study a typed system in~\rsec{type_system}.

\section{The $\lambdaunif$-Calculus --- Operational Semantics}
  \lsec{lambdaunif_calculus}
  
In this section we describe the operational semantics
of our proposed calculus, including its syntax,
reduction rules~(\rdef{operational_semantics_reduction_rules}),
an invariant ({\em coherence}) which is preserved by reduction~(\rlem{location_coherence_invariant}),
and a notion of structural equivalence which is a
{\em strong bisimulation} with respect to reduction~(\rlem{reduction_modulo_structeq}).
\medskip

{\bf Syntax of terms and programs.}
Suppose given denumerably infinite sets
of {\em variables} $\Var = \set{\var, \vartwo, \varthree, \hdots}$,
{\em constructors} $\Cons = \set{\cons, \constwo, \consthree, \hdots}$,
and
{\em locations} $\Loc = \set{\loc,\loctwo,\locthree,\hdots}$.
We assume that there is a distinguished constructor $\unit$.
The sets of
{\em terms} $\tm, \tmtwo, \hdots$ and
{\em programs} $\prog,\progtwo,\hdots$
are defined mutually inductively as follows:
\[
  \begin{array}[t]{lrrl@{\HS}l@{}lrrl@{\HS}l}
    & \tm & ::=  & \var                       & \textup{variable}       &
    &     & \mid & \cons                      & \textup{constructor}    \\
    &     & \mid & \lam{\var}{\prog}          & \textup{abstraction} &
    &     & \mid & \laml{\loc}{\var}{\prog}   & \textup{allocated abstraction} \\
    &     & \mid & \tm\,\tm                   & \textup{application} &
    &     & \mid & \fresh{\var}{\tm}          & \textup{fresh variable introduction} \\
    &     & \mid & \tm \seq \tm               & \textup{guarded expression}       &
    &     & \mid & \tm \unif \tm              & \textup{unification}    \\
  \\
    & \prog & ::=  & \fail                 & \textup{empty program} \\
    &       & \mid & \tm \alt \prog        & \textup{non-deterministic choice} \\
  \end{array}
\]
The set of {\em values} $\Val = \set{\val,\valtwo,\hdots}$
is a subset of the set of terms, given by the grammar
$\val ::= \var \mid \laml{\loc}{\var}{\prog} \mid \cons\,\val_1\hdots\val_n$.
Values of the form $\cons\,\val_1\hdots\val_n$ are called {\em structures}.

Intuitively,
an (unallocated) abstraction $\lam{\var}{\prog}$ represents the static code
to create a closure, while $\laml{\loc}{\var}{\prog}$ represents the closure
created in runtime, stored in the memory cell $\loc$.
When the abstraction is evaluated, it becomes decorated with a location
({\em allocated}). We will have a rewriting rule like
$\lam{\var}{\prog} \to \laml{\loc}{\var}{\prog}$ where $\loc$ is fresh.

{\bf Notational conventions.}
We write $\gctx,\gctx',\hdots$ for {\em arbitrary contexts}, \ie terms with
a single free occurrence of a hole $\ctxhole$.
We write $\wctx,\wctx',\hdots$ for {\em weak contexts}, which 
do not enter below abstractions nor fresh variable declarations,
\ie
$\wctx ::= \ctxhole
      \mid \wctx\,\tm
      \mid \tm\,\wctx
      \mid \wctx \seq \tm
      \mid \tm \seq \wctx
      \mid \wctx \unif \tm
      \mid \tm \unif \wctx$.
We write $\alt_{i=1}^{n} \tm_i$
or also $\tm_1 \alt \tm_2 \hdots \alt \tm_n$ to stand for the program
$\tm_1 \alt (\tm_2 \alt \hdots (\tm_n \alt \fail))$.
In particular, if $\tm$ is a term, sometimes we write $\tm$ for the {\em singleton} program $\tm \alt \fail$.
The set of free variables $\fv{\tm}$ (resp. $\fv{\prog}$) of a term (resp. program) is defined as expected,
noting that fresh variable declarations $\fresh{\var}{\tm}$ and both kinds of abstractions
$\lam{\var}{\prog}$ and $\laml{\loc}{\var}{\prog}$ bind the free occurrences of $\var$ in the body.
Expressions are considered up to $\alpha$-equivalence, \ie renaming of all bound variables.
Given a context or weak context $\gctx$ and a term $\tm$,
we write $\gctxof{\tm}$ for the (capturing) substitution of $\ctxhole$ by $\tm$ in $\gctx$.
The set of locations $\locs{\tm}$ (resp. $\fv{\prog}$) of a term (resp. program) is defined as the
set of all locations $\loc$ decorating any abstraction on $\tm$.
We write $\tm\sub{\loc}{\loc'}$ for the term that results from replacing all occurrences of
the location $\loc$ in $\tm$ by $\loc'$.
The program being evaluated is called
the {\em toplevel program}. The toplevel program is always of the form
$\tm_1 \alt \tm_2 \hdots \alt \tm_n$, and each of the $\tm_i$ is called a {\em thread}.

{\bf Operations with programs.}
We define the operations $\prog \alt \progtwo$ and $\wctxof{\prog}$ by
induction on the structure of $\prog$ as follows;
note that the notation ``$\alt$'' is overloaded
both for consing a term onto a program and for concatenating programs:
\[
  \begin{array}{rcll}
    \fail            \alt \progtwo & \eqdef & \progtwo \\
    (\tm \alt \prog) \alt \progtwo & \eqdef & \tm \alt (\prog \alt \progtwo) \\
  \end{array}
  \hspace{1cm}
  \begin{array}{rcll}
    \wctxof{\fail}          & \eqdef & \fail \\
    \wctxof{\tm \alt \prog} & \eqdef & \wctxof{\tm} \alt \wctxof{\prog} \\
  \end{array}
\]

{\bf Substitutions.}
A {\em substitution} is a function
$\subst : \Var \to \Val$ with {\em finite support},
\ie such that the set $\supp(\subst) \eqdef \set{\var \ST \subst(\var) \neq \var}$
is finite.
We write $\set{\var_1 \mapsto \val_1, \hdots, \var_n \mapsto \val_n}$
for the substitution $\subst$ such that
$\supp(\subst) = \set{\var_1, \hdots, \var_n}$
and $\subst(\var_i) = \val_i$ for all $i \in 1..n$.
A {\em renaming} is a bijective substitution mapping each variable to a variable, \ie
a substitution of the form $\set{\var_1 \mapsto \vartwo_1, \hdots, \var_n \mapsto \vartwo_n}$.

If $\subst : \Var \to \Val$ is a substitution and
$\tm$ is a term, $\tm\SUB{\subst}$ denotes the capture-avoiding substitution
of each occurrence of a free variable $\var$ in $\tm$
by $\subst(\var)$.
Capture-avoiding substitution of a single variable $\var$
by a value $\val$ in a term $\tm$ is written $\tm\sub{\var}{\val}$
and defined by $\tm\SUB{\set{\var\mapsto\val}}$.
Subsitutions $\substtwo,\subst$ may be {\em composed} as follows:
$(\substtwo\scomp\subst)(\var) \eqdef \substtwo(\var)\SUB{\subst}$.
Substitutions can also be applied to weak contexts,
taking $\ctxhole\SUB{\subst} \eqdef \ctxhole$.
A substitution $\subst$ is {\em idempotent} if $\subst \cdot \subst = \subst$.
A substitution $\subst$ is {\em more general} than a substitution $\substtwo$,
written $\subst \sleq \substtwo$ if there is a substitution $\substthree$
such that $\substtwo = \subst \scomp \substthree$.


{\bf Unification.}
We describe how to adapt the standard first-order unification algorithm
to our setting, in order to deal with unification of $\lambda$-abstractions.
As mentioned before, our aim is to solve only first-order unification problems.
This means that the unification algorithm should only deal with equations
involving terms which are already {\em values}.
Note that unallocated abstractions ($\lam{\var}{\prog}$) are {\em not} considered values;
abstractions are only values when they are allocated ($\laml{\loc}{\var}{\prog}$).
Allocated abstractions are to be considered equal
if and only if they are decorated with the same location.
Note that terms of the form $\var\,\tm_1\hdots\tm_n$ are {\em not} considered
values if $n > 0$, as this would pose a higher-order unification problem,
possibly requiring to instantiate $\var$ as a function of its arguments.

We expand briefly on why a naive approach to first-order unification would not work.
Suppose that we did not have locations and we declared
that two abstractions $\lam{\var}{\prog}$ and $\lam{\vartwo}{\progtwo}$ are equal
whenever their bodies are equal, up to $\alpha$-renaming
(\ie $\prog\sub{\var}{\vartwo} = \progtwo$).
The problem is that this notion of equality is not preserved by substitution,
for example, the unification problem given by the equation
$\laml{}{\var}{\vartwo} \unif \laml{}{\var}{\varthree}$ would {\em fail},
as $\vartwo \neq \varthree$.
However, the variable $\vartwo$ may become instantiated into $\varthree$, and
the equation would become $\laml{}{\var}{\varthree} \unif \laml{}{\var}{\varthree}$,
which succeeds. This corresponds to the following critical pair in the calculus,
which cannot be closed:
\[
  \fail
  \leftarrow
  (\laml{}{\var}{\vartwo} \unif \laml{}{\var}{\varthree}) \seq (\vartwo \unif \varthree)
  \to
  (\laml{}{\var}{\varthree} \unif \laml{}{\var}{\varthree}) \seq \unit
  \to
  \unit \seq \unit
\]
This is where the notion of {\em allocated abstraction} plays an important role.
We will work with the invariant that if $\laml{\loc}{\var}{\prog}$ and $\laml{\loctwo}{\vartwo}{\progtwo}$
are two allocated abstractions in the same location ($\loc = \loctwo$) then their bodies will be equal,
up to $\alpha$-renaming.
This ensures that different allocated abstractions are still different after
substitution, as they must be decorated with different locations.

{\bf Unification goals and unifiers.}
A {\em goal} is a term of the form $\val \unif \valtwo$.
A {\em unification problem} is a finite set of goals $\goals = \set{\val_1 \unif \valtwo_1, \hdots, \val_n \unif \valtwo_n}$.
If $\subst$ is a substitution we write $\goals\SUB{\subst}$ for
$\set{\val_1\SUB{\subst} \unif \valtwo_1\SUB{\subst}, \hdots, \val_n\SUB{\subst} \unif \valtwo_n\SUB{\subst}}$.
A {\em unifier} for $\goals = \set{\val_1 \unif \valtwo_1, \hdots, \val_n \unif \valtwo_n}$ is
a substitution $\subst$ such that $\val_i\SUB{\subst} = \valtwo_i\SUB{\subst}$
for all $1 \leq i \leq n$.
A unifier $\subst$ for $\goals$ is {\em most general}
if for any other unifier $\substtwo$ one has $\subst \sleq \substtwo$.

{\bf Coherence invariant.}
As mentioned before, we impose an invariant on programs forcing that
allocated abstractions decorated with the same location must be syntactically
equal. Moreover, we require that allocated abstractions do not refer to
variables bound outside of their scope, \ie that they are in fact closures.
Note that the source program trivially satisfies this invariant,
as it is expected that allocated abstractions are not written by the user
but generated at runtime.

More precisely, a set $X$ of terms is {\bf coherent} if the two following conditions hold.
{\bf (1)}
  Consider any allocated abstraction under a context $\gctx$,
  \ie let $\tm \in X$ such that $\tm = \gctxof{\laml{\loc}{\var}{\prog}}$.
  Then the context $\gctx$ does not bind any of the free variables of $\laml{\loc}{\var}{\prog}$.
{\bf (2)}
  Consider any two allocated abstractions in $\tm$ and $\tmtwo$ with the same location,
  \ie let $\tm, \tmtwo \in X$ be such that $\tm = \gctxof{\laml{\loc}{\var}{\prog}}$
  and $\tmtwo = \gctx'\ctxof{\laml{\loc}{\vartwo}{\progtwo}}$,
  Then $\prog\sub{\var}{\vartwo} = \progtwo$.

We extend the notion of coherence to other syntactic categories as follows.
A term $\tm$ is coherent if $\set{\tm}$ is coherent.
A program $\prog = \tm_1 \alt \hdots \alt \tm_n$ is {\em coherent}
if each thread $\tm_i$ is coherent.
A unification problem $\goals$ is {\em coherent} if it is coherent seen as a set.
Note that a program may be coherent even if different abstractions
in different threads have the same location.
For example,
$(\laml{\loc}{\var}{\var\,\var} \unif \laml{\loc}{\vartwo}{\cons}) \alt (\laml{\loctwo}{\vartwo}{\vartwo})$
is not coherent,
whereas
$(\laml{\loc}{\var}{\var\,\var} \unif \laml{\loc}{\vartwo}{\vartwo\,\vartwo}) \alt (\laml{\loc}{\vartwo}{\cons})$ is coherent.

{\bf Unification algorithm.}
The standard Martelli--Montanari~\cite{martelli1982efficient} unification algorithm
can be adapted to our setting.
In particular, there is a computable function $\mgu{-}$
such that if $\goals$ is a coherent unification problem
then either
$\mgu{\goals} = \subst$,
\ie $\mgu{\goals}$ returns a substitution $\subst$ which is an idempotent most general unifier for $\goals$,
or $\mgu{\goals} = \FAIL$, \ie $\mgu{\goals}$ fails and $\goals$ has no unifier.
Moreover, it can be shown that if the algorithm succeeds,
the set $\goals\SUB{\subst} \cup \set{\subst(\var) \ST \var \in \Var}$ is coherent.
\WithProofs{The algorithm, formal statement and proofs are detailed in the
appendix~\SeeAppendixRef{\rsec{appendix_unification_algorithm}}.}
\medskip


{\bf Operational semantics.}
The {\bf $\lambdaunif$-calculus} is the rewriting system whose
objects are programs, and whose reduction relation is given by
the union of the following six rules:

\begin{definition}[Reduction rules]
\ldef{operational_semantics_reduction_rules}
\[
{\small
  \begin{array}{rll@{\HS}l}
    \prog_1 \alt \wctxof{\lam{\var}{\prog}} \alt \prog_2
    & \toca{alloc} &
    \prog_1 \alt \wctxof{\laml{\loc}{\var}{\prog}} \alt \prog_2
    & \text{if $\loc \not\in \locs{\wctxof{\lam{\var}{\prog}}}$}
  \\
    \prog_1 \alt \wctxof{(\laml{\loc}{\var}{\prog})\,\val} \alt \prog_2
    & \toca{beta} &
    \prog_1 \alt \wctxof{\prog\sub{\var}{\val}} \alt \prog_2
  \\
    \prog_1 \alt \wctxof{\val\seq\tm} \alt \prog_2
    & \toca{guard} &
    \prog_1 \alt \wctxof{\tm} \alt \prog_2
  \\
    \prog_1 \alt \wctxof{\fresh{\var}{\tm}} \alt \prog_2
    & \toca{fresh} &
    \prog_1 \alt \wctxof{\tm\sub{\var}{\vartwo}} \alt \prog_2
    & \text{if $\vartwo$ is fresh w.r.t. $\wctxof{\fresh{\var}{\tm}}$}
  \\
    \prog_1 \alt \wctxof{\val \unif \valtwo} \alt \prog_2
    & \toca{unif} &
    \prog_1 \alt \wctxof{\unit}\SUB{\subst} \alt \prog_2
    & \text{if $\mgu{\set{\val\unif\valtwo}} = \subst$}
  \\
    \prog_1 \alt \wctxof{\val \unif \valtwo} \alt \prog_2
    & \toca{fail} &
    \prog_1 \alt \prog_2
    & \text{if $\mgu{\set{\val\unif\valtwo}}$ fails}
  \end{array}
}
\]
\end{definition}
Note that all rules operate on a single thread and they are {\bf not} closed
under any kind of evaluation contexts.
The \indrulename{alloc} rule allocates a closure,
\ie whenever a $\lambda$-abstraction is found below an evaluation
context, it may be assigned a fresh location $\loc$.
The \indrulename{beta} rule applies a function to a value.
The \indrulename{guard} rule proceeds with the evaluation
of the right part of a guarded expression when the left part is already a value.
The \indrulename{fresh} rule introduces a fresh symbolic variable.
The requirement that $\vartwo$ be fresh can be stated more precisely
as the condition that
$\vartwo$ does not occur free in $\wctx$ nor in $\tm$,
\ie $\vartwo \notin \fv{\wctx} \cup \fv{\tm}$,
and that $\vartwo$ is not bound by $\wctx$,
\ie $\vartwo \in \fv{\wctxof{\vartwo}}$.
The \indrulename{unif} and \indrulename{fail} rules solve a unification
problem, corresponding to the success and failure cases respectively.
If there is a unifier, the substitution is applied to the affected thread.
For example:
\[
{\small
  \begin{array}{rrl}
  (\lam{\var}{\var\alt(\fresh{\vartwo}{((\var \unif \cons\,\vartwo)\seq\vartwo)})})\,(\cons\,\constwo)
  & \toca{alloc} &
  (\laml{\loc}{\var}{\var\alt(\fresh{\vartwo}{((\var \unif \cons\,\vartwo)\seq\vartwo)})})\,(\cons\,\constwo)
  \\
  & \toca{beta} &
  \cons\,\constwo \alt \fresh{\vartwo}{((\cons\,\constwo \unif \cons\,\vartwo)\seq\vartwo)}
  \\
  & \toca{fresh} &
  \cons\,\constwo \alt ((\cons\,\constwo \unif \cons\,\varthree)\seq\varthree)
  \\
  & \toca{unif} &
  \cons\,\constwo \alt (\unit\seq\constwo)
  \\
  & \toca{guard} &
  \cons\,\constwo \alt \constwo
  \end{array}
}
\]

{\bf Structural equivalence.}
As already remarked in~\rsec{technical_challenges}, we will not be able to
prove that confluence holds strictly speaking,
but only {\em up to reordering of threads}
in the toplevel program.
Moreover the \indrulename{alloc} and \indrulename{fresh} rules introduce
fresh names, and, as usual the most general unifier is unique only
{\em up to renaming}.
These conditions are expressed formally by means of the following relation
of structural equivalence.

Formally, {\bf structural equivalence} between programs
is written $\prog \structeq \progtwo$ and defined as the reflexive, symmetric,
and transitive closure of the three following axioms:
\begin{enumerate}
\item \rulename{$\structeq$-swap}:
  $\prog \alt \tm \alt \tmtwo \alt \progtwo \structeq \prog \alt \tmtwo \alt \tm \alt \progtwo$.
\item \rulename{$\structeq$-var}:
  If $\vartwo \not\in \fv{\tm}$ then
  $\prog \alt \tm \alt \progtwo \structeq \prog \alt \tm\sub{\var}{\vartwo} \alt \progtwo$.
\item \rulename{$\structeq$-loc}:
  If $\loc' \not\in \locs{\tm}$, then
  $\prog \alt \tm \alt \progtwo \structeq \prog \alt \tm\sub{\loc}{\loc'} \alt \progtwo$.
\end{enumerate}
In short, \indrulename{$\structeq$-swap} means that threads may be reordered arbitrarily,
\indrulename{$\structeq$-var} means that symbolic variables are local to each thread,
and \indrulename{$\structeq$-loc} means that locations are local to each thread.

\medskip
The following lemma establishes that the coherence invariant
is closed by reduction and structural equivalence,
which means that the $\lambdaunif$-calculus is well-defined
if restricted to coherent programs. In the rest of this paper, we
always assume that {\bf all programs enjoy the coherence invariant}.

\begin{lemma}
\llem{location_coherence_invariant}
Let $\prog$ be a coherent program.
If $\prog \structeq \progtwo$ or $\prog \toca{} \progtwo$,
then $\progtwo$ is also coherent.
\SeeAppendixRef{\rsec{appendix_location_coherence_invariant}}
\end{lemma}

The following lemma establishes that reduction is well-defined
modulo structural equivalence (\ie it lifts to $\structeq$-equivalence classes):

\begin{lemma}
\llem{reduction_modulo_structeq}
Structural equivalence is a strong bisimulation with respect to $\to$.
Precisely, let $\prog \equiv \prog' \toca{x} \progtwo$
with $\rulename{x} \in \set{
  \rulename{alloc}, \rulename{beta}, \rulename{guard},
  \rulename{fresh}, \rulename{unif}, \rulename{fail}
}$.
Then there exists a program $\progtwo'$ such that $\prog \toca{x} \progtwo' \equiv \progtwo$.
\SeeAppendixRef{\rsec{appendix_reduction_modulo_structeq}}
\end{lemma}

\begin{example}[Type inference algorithm]
As an illustrative example, the following translation $\walg{-}$
converts an untyped $\lambda$-term $\tm$ into a $\lambdaunif$-term that
calculates the principal type of $\tm$ according to the
usual Hindley--Milner~\cite{milner1978theory}
type inference algorithm, or fails if it has no type.
Note that an arrow type $(\typ \to \typtwo)$ is encoded as $({\bf f}\,\typ\,\typtwo)$:
\[\FIT{$
  \walg{\var} \eqdef a_{\var}
  \HS
  \walg{\lam{\var}{\tm}} \eqdef \fresh{a_{\var}}{\mathbf{f}\,a_{\var}\,\walg{\tm}}
  \HS
  \walg{\tm\,\tmtwo}     \eqdef \fresh{a}((\walg{\tm} \unif \mathbf{f}\,\walg{\tmtwo}\,a) \seq a)
$}\]
For instance,
$\walg{\lam{\var}{\lam{\vartwo}{\vartwo\,\var}}}
= \fresh{a}{\mathbf{f}\,a\,(\fresh{b}{\mathbf{f}\,b\,(\fresh{c}{(b \unif \mathbf{f}\,a\,c)\seq c})})}
\twoheadrightarrow
\mathbf{f}\,a\,(\mathbf{f}\,(\mathbf{f}\,a\,c)\,c)$.
\end{example}

\section{Operational Properties}
  
In this section we study some properties of the operational semantics.
First, we characterize the set of {\em normal forms} of the $\lambdaunif$-calculus
syntactically, by means of an inductive
definition~(\rprop{characterization_normal_forms}).
Then we turn to the main result of this section, proving that it enjoys
confluence up to structural equivalence~(\rthm{confluence}).
\smallskip

{\bf Characterization of normal forms.} 
The set of {\bf normal terms} $\ntm,\ntmtwo,\hdots$
and {\bf stuck terms} $\stm,\stmtwo,\hdots$
are defined mutually inductively as follows.
A normal term is either a value or a stuck term, \ie $\ntm ::= \val \mid \stm$.
A term is stuck if the judgment $\tm\stuck$ is derivable with the following rules:
\[
{\small
  \indrule{stuck-var}{
    n > 0
  }{
    \var\,\ntm_1\hdots\ntm_n \stuck
  }
  \HS
  \indrule{stuck-cons}{
    \ntm_i \stuck \text{ for some $i \in \set{1,2,\hdots,n}$}
  }{
    \cons\,\ntm_1\hdots\ntm_n \stuck
  }
}
\]
\[
{\small
  \indrule{stuck-guard}{
    \ntm_1 \stuck
    \HS
    n \geq 0
  }{
    (\ntm_1\seq\ntm_2)\,\ntmtwo_1\hdots\ntmtwo_n \stuck
  }
  \HS
  \indrule{stuck-unif}{
    \ntm_i \stuck \text{ for some $i \in \set{1,2}$}
    \HS
    n \geq 0
  }{
    (\ntm_1\unif\ntm_2)\,\ntmtwo_1\hdots\ntmtwo_n \stuck
  }
}
\]
\[
{\small
  \indrule{stuck-lam}{
    \ntm \stuck
    \HS
    n \geq 0
  }{
    (\laml{\loc}{\var}{\prog})\,\ntm\,\ntmtwo_1\hdots\ntmtwo_n \stuck
  }
}
\]
The set of {\bf normal programs} $\nprog,\nprogtwo,\hdots$ is given by the
following grammar: $\nprog ::= \fail \mid \ntm \alt \nprog$.
For example, the program $(\laml{\loc}{\var}{\var\unif\var}) \alt ((\vartwo\,\cons\unif\constwo)\seq\consthree) \alt \varthree\,(\varthree\,\cons)$
is normal, being the non-deterministic alternative of a value and two stuck terms.
Normal programs capture the notion of normal form:
\begin{proposition}
\lprop{characterization_normal_forms}
The set of normal programs is exactly the set of $\toca{}$-normal forms.
\SeeAppendixRef{\rsec{appendix_characterization_normal_forms}}
\end{proposition}
\bigskip

{\bf Confluence.}
In order to prove that the $\lambdaunif$-calculus has the Church--Rosser
property, we adapt the method due to Tait and Martin-L\"of~\cite[Sec.~3.2]{Barendregt:1984}
by defining a {\em simultaneous reduction relation} $\topar{}$,
and showing that it verifies the diamond property
(\ie $\Leftarrow\Rightarrow \,\subseteq\, \Rightarrow\Leftarrow$)
and
the inclusions $\to \,\subseteq\, \topar{} \,\subseteq\, \twoheadrightarrow$,
where $\twoheadrightarrow$ denotes the reflexive--transitive closure
of $\to$.
Actually, these properties only hold up to structural equivalence,
so our confluence result,
rather than the usual inclusion $\twoheadleftarrow\twoheadrightarrow \,\subseteq\, \twoheadrightarrow\twoheadleftarrow$,
expresses the weakened inclusion $\twoheadleftarrow\twoheadrightarrow \,\subseteq\, \twoheadrightarrow\structeq\twoheadleftarrow$.

To define the relation of simultaneous reduction,
we use the following notation,
to lift the binary operations of unification ($\tm\unif\tmtwo$),
guarded expression ($\tm\seq\tmtwo$), and application ($\tm\,\tmtwo$)
from the sort of terms to the sort of programs.
Let $\star$ denote a binary term constructor (\eg unification, guarded expression, or application).
Then we write
$
  (\bigalt_{i=1}^n \tm_i) \star (\bigalt_{j=1}^m \tmtwo_j)
  \eqdef
  \bigalt_{i=1}^n \bigalt_{j=1}^m (\tm_i\star\tmtwo_j)
$.

First, we define a judgment $\tm \topar{\goals} \prog$ of simultaneous
reduction, relating a term and a program,
parameterized by a set $\goals$ of unification goals representing
pending constraints:
\[
{\small
  \indrule{Var}{
  }{
    \var \topar{\emptyset} \var
  }
  \HS
  \indrule{Cons}{
  }{
    \cons \topar{\emptyset} \cons
  }
  \HS
  \indrule{Fresh$_1$}{
  }{
    \fresh{\var}{\tm} \topar{\emptyset} \fresh{\var}{\tm}
  }
  \HS
  \indrule{Fresh$_2$}{
    \tm \topar{\goals} \prog \HS \text{$\var$ fresh}
  }{
    \fresh{\var}{\tm} \topar{\goals} \prog
  }
}
\]
\[
{\small
  \indrule{Abs$^\codesym_1$}{
  }{
    \lam{\var}{\prog} \topar{\emptyset} \lam{\var}{\prog} 
  }
  \HS
  \indrule{Abs$^\codesym_2$}{
    \text{$\loc$ fresh}
  }{
    \lam{\var}{\prog} \topar{\emptyset} \laml{\loc}{\var}{\prog} 
  }
  \HS
  \indrule{Abs$^\allocsym$}{
  }{
    \laml{\loc}{\var}{\prog} \topar{\emptyset} \laml{\loc}{\var}{\prog}
  }
}
\]
\[
{\small
  \indrule{App$_1$}{
    \tm \topar{\goals} \prog
    \HS
    \tmtwo \topar{\goalstwo} \progtwo
  }{
    \tm\,\tmtwo \topar{\goals \cup \goalstwo} \prog\,\progtwo
  }
  \HS
  \indrule{App$_2$}{
  }{
    (\laml{\loc}{\var}{\prog})\,\val \topar{\emptyset} \prog\sub{\var}{\val}
  }
  \HS
  \indrule{Guard$_1$}{
    \tm \topar{\goals} \prog
    \HS
    \tmtwo \topar{\goalstwo} \progtwo
  }{
    \tm\seq\tmtwo \topar{\goals \cup \goalstwo} \prog\seq\progtwo
  }
}
\]
\[
{\small
  \indrule{Guard$_2$}{
    \tm \topar{\goals} \prog
  }{
    \val\seq\tm \topar{\goals} \prog
  }
  \HS
  \indrule{Unif$_1$}{
    \tm \topar{\goals} \prog
    \HS
    \tmtwo \topar{\goalstwo} \progtwo
  }{
    \tm\unif\tmtwo \topar{\goals \cup \goalstwo} \prog\unif\progtwo
  }
  \HS
  \indrule{Unif$_2$}{
  }{
    \val\unif\valtwo \topar{\set{\val \unif \valtwo}} \unit
  }
}
\]
As usual, most term constructors have two rules,
the rule decorated with ``$1$'' is a congruence rule which chooses not
to perform any evaluation on the root of the term,
while the rule decorated with ``$2$'' requires that there is a redex at the
root of the term, and contracts it.
Note that rule \indrulename{Unif$_2$} does {\em not} perform
the unification of $\val$ and $\valtwo$ immediately;
it merely has the effect of propagating the unification constraint.

Using the relation defined above, we are now able to define the relation
of {\bf simultaneous reduction} between programs:
\[
{\small
  \indrule{Fail}{
  }{
    \fail \topar{} \fail
  }
  \HS
  \indrule{Alt}{
    \tm \topar{\goals} \prog
    \HS
    \progtwo \topar{} \progtwo'
    \HS
    \prog' =
      \begin{cases}
      \prog\SUB{\subst} & \text{if $\subst = \mgu{\goals}$} \\
      \fail             & \text{if $\mgu{\goals}$ fails} \\
      \end{cases}
  }{
    \tm\alt\progtwo \topar{} \prog' \alt \progtwo'
  }
}
\]

The following lemma summarizes some of the key properties of simultaneous
reduction. Most are straightforward proofs by induction, except for item~3.:
\begin{lemma}[Properties of simultaneous reduction]
\llem{properties_of_simultaneous_reduction}
\llem{parallel_context_closure}
\llem{parallel_modulo_structeq}
\llem{parallel_substitution}
\quad
\begin{enumerate}
\item
  {\bf Reflexivity.} 
  $\tm \topar{\emptyset} \tm$ and $\prog \topar{} \prog$.
\item
  {\bf Context closure.}
    If $\tm \topar{\goals} \prog$ then $\wctxof{\tm} \topar{\goals} \wctxof{\prog}$.
\item
  {\bf Strong bisimulation.}
  \label{parallel_modulo_structeq_item}
    Structural equivalence is a strong bisimulation with respect to $\topar{}$,
    \ie if $\prog \structeq \prog' \topar{} \progtwo$ then there is a program $\progtwo'$
    such that $\prog \topar{} \progtwo' \structeq \progtwo$.
  \SeeAppendixRef{\rsec{appendix_parallel_modulo_structeq}}
\item
  {\bf Substitution.}
    If $\tm \topar{\goals} \prog$
    then $\tm\SUB{\subst} \topar{\goals\SUB{\subst}} \prog\SUB{\subst}$.
\end{enumerate}
\end{lemma}

The core argument is the following adaptation of Tait--Martin-L\"of's technique,
from which confluence comes out as an easy corollary.
\WithProofs{See \SeeAppendixRef{\rsec{appendix_tait_martin_lof_technique}}
in the appendix for details.}
\begin{proposition}[Tait--Martin-L\"of's technique, up to $\structeq$]
\lprop{tait_martin_lof_technique}
\quad\\
1. $\toca{}\ \subseteq\ \topar{}\structeq$ \\
2. $\topar{}\, \subseteq\ \rtoca{}\structeq$ \\
3. $\topar{}$ has the diamond property, up to $\structeq$, that is:\\
\hspace*{.7em}
   If $\prog_1 \topar{} \prog_2$ and $\prog_1 \topar{} \prog_3$
   then $\prog_2 \topar{}\equiv \prog_4$ and $\prog_3 \topar{}\equiv \prog_4$
   for some $\prog_4$.
\end{proposition}

\begin{theorem}[Confluence]
\lthm{confluence}
The reduction relation $\toca{}$ is confluent, up to $\structeq$.
More precisely, if
$\prog_1 \rtoca{} \prog_2$
and
$\prog_1 \rtoca{} \prog_3$
then there is a program $\prog_4$ such that
$\prog_2 \rtoca{}\structeq \prog_4$
and
$\prog_3 \rtoca{}\structeq \prog_4$.
\end{theorem}

\begin{example}
\lexample{confluence_unif_unif}
Suppose that $\sigma = \mgu{\val_1\unif\val_2}$ and $\tau = \mgu{\valtwo_1\unif\valtwo_2}$.
Consider:
\[
  (\val_1\SUB{\tau}\unif\val_2\SUB{\tau})\,\unit\,\tm\SUB{\tau}
  \leftarrow
  (\val_1\unif\val_2)\,(\valtwo_1\unif\valtwo_2)\,\tm
  \rightarrow
  \unit\,(\valtwo_1\SUB{\sigma}\unif\valtwo_2\SUB{\sigma})\,\tm\SUB{\sigma}
\]
Then both $\sigma' = \mgu{\val_1\SUB{\tau}\unif\val_2\SUB{\tau}}$
and $\tau' = \mgu{\valtwo_1\SUB{\sigma}\unif\valtwo_2\SUB{\sigma}}$
must exist, and the peak may be closed as follows:
\[
  (\val_1\SUB{\tau}\unif\val_2\SUB{\tau})\,\unit\,\tm\SUB{\tau}
  \to
  \unit\,\unit\,(\tm\SUB{\tau})\SUB{\sigma'}
  \equiv
  \unit\,\unit\,(\tm\SUB{\sigma})\SUB{\tau'}
  \leftarrow
  \unit\,(\valtwo_1\SUB{\sigma}\unif\valtwo_2\SUB{\sigma})\,\tm\SUB{\sigma}
\]
the equivalence relies on the fact that
$\tau' \circ \sigma$ and $\sigma' \circ \tau$
are both most general unifiers of
$\set{\val_1 \unif \val_2, \valtwo_1 \unif \valtwo_2}$,
hence
$(\tm\SUB{\tau})\SUB{\sigma'} \equiv (\tm\SUB{\sigma})\SUB{\tau'}$,
up to renaming.
\end{example}

\section{Simple Types for $\lambdaunif$}
  \lsec{type_system}
  
In this section we discuss a simply typed system for the $\lambdaunif$-calculus.
The system does not present any essential difficulty, but it is a necessary
prerequisite to be able to define the denotational semantics
of~\rsec{denotational_semantics}. The main result in this section is
subject reduction~(\rprop{subject_reduction}).

Note that, unlike in the simply typed $\lambda$-calculus,
reduction may create free variables, due to fresh variable introduction.
For instance, in the reduction step $\cons(\fresh{\var}{\var}) \to \cons\,\var$,
a new variable $\var$ appears free on the right-hand side.
Therefore the subject reduction lemma has to {\em extend} the typing context
in order to account for freshly created variables.
This may be understood only as a matter of notation, \eg in a different
presentation of the $\lambdaunif$-calculus the step above could be
written as $\cons(\fresh{\var}{\var}) \to \fresh{\var}{(\cons\,\var)}$,
using a {\em scope extrusion} rule reminiscent of the rule to create
new channels in process calculi (\eg $\pi$-calculus),
avoiding the creation of free variables.

{\bf Types and typing contexts.}
Suppose given a denumerable set of {\em base types}
$\btyp, \btyptwo, \btypthree, \hdots$.
The sets of {\em types} $\Type = \set{\typ,\typtwo,\hdots}$
and {\em typing contexts} $\tctx,\tctxtwo,\hdots$ are given by:
\[
  \typ,\typtwo,\hdots ::= \btyp \mid \typ \to \typtwo
  \HS
  \tctx ::= \emptyctx \mid \tctx,\var:\typ
\]
we assume that no variable occurs twice in a typing context.
Typing contexts are to be regarded as finite sets of assumptions of the
form $(\var:\typ)$, \ie we work implicitly modulo contraction and
exchange.
We assume that each constructor $\cons$ has an associated type
$\constyp{\cons}$.

{\bf Typing rules.}
Judgments are of the form.
``$\tctx \vdash X : \typ$'' where $X$ may be a term or a program,
meaning that $X$ has type $\typ$ under $\tctx$.
The typing rules are the following:
\[
  \indrule{t-var}{
    (\var : \typ) \in \tctx
  }{
    \tctx \vdash \var : \typ
  }
  \HS
  \indrule{t-cons}{
  }{
    \tctx \vdash \cons : \constyp{\cons}
  }
\]
\[
  \indrule{t-app}{
    \tctx \vdash \tm : \typ \to \typtwo
    \HS
    \tctx \vdash \tmtwo : \typ
  }{
    \tctx \vdash \tm\,\tmtwo : \typtwo
  }
  \HS
  \indrule{t-lam(l)}{
    \tctx,\var:\typ \vdash \prog : \typtwo
  }{
    \tctx \vdash \laml{(\loc)}{\var}{\prog} : \typ \to \typtwo
  }
\]
\[
  \indrule{t-unif}{
    \tctx \vdash \tm : \typ
    \HS
    \tctx \vdash \tmtwo : \typ
  }{
    \tctx \vdash \tm \unif \tmtwo : \constyp{\unit}
  }
  \HS
  \indrule{t-guard}{
    \tctx \vdash \tm : \constyp{\unit}
    \HS
    \tctx \vdash \tmtwo : \typ
  }{
    \tctx \vdash \tm\seq\tmtwo : \typ
  }
\]
\[
  \indrule{t-fresh}{
    \tctx,\var:\typ \vdash \tm : \typtwo
  }{
    \tctx \vdash \fresh{\var}{\tm} : \typtwo
  }
  \HS
  \indrule{t-fail}{
  }{
    \tctx \vdash \fail : \typ
  }
  \HS
  \indrule{t-alt}{
    \tctx \vdash \tm : \typ
    \HS
    \tctx \vdash \prog : \typ
  }{
    \tctx \vdash \tm \alt \prog : \typ
  }
\]
Note that all abstractions are typed in the same way, regardless of whether
they are allocated or not.
A unification has the same type as the constructor $\unit$,
as does $\tm$ in the guarded expression $(\tm\seq\tmtwo)$.
A freshly introduced variable of type $\typ$ represents,
from the logical point of view, an unjustified assumption of $\typ$.
The empty program $\fail$ can also be given any type.
All the threads in a program must have the same type.
The following properties of the type system are routine:
\begin{lemma}
\llem{type_system_basic_properties}
\llem{weakening}
\llem{strengthening}
\llem{substitution}
\llem{contextual_substitution}
\llem{program_composition}
Let $X$ stand for either a term or a program. Then:
\begin{enumerate}
\item
  {\bf Weakening.}
  If $\tctx \vdash X : \typ$
  then $\tctx, \var : \typtwo \vdash X : \typ$.
\item
  {\bf Strengthening.}
  If $\tctx, \var : \typ \vdash X : \typtwo$
  and $\var \not\in \fv{X}$, then $\tctx \vdash X : \typtwo$.
\item
  {\bf Substitution.}
  If $\tctx, \var:\typ \vdash X : \typtwo$
  and $\tctx \vdash \tmtwo : \typ$
  then $\tctx \vdash X\sub{\var}{\tmtwo} : \typtwo$.
\item
  {\bf Contextual substitution.}
  $\tctx \vdash \wctxof{\tm} : \typ$ holds
  if and only if
  there is a type $\typtwo$ such that
  $\tctx, \ctxhole:\typtwo \vdash \wctx : \typ$
  and $\tctx \vdash \tm : \typtwo$ hold.
\item
  {\bf Program composition/decomposition.}
  $\tctx \vdash \prog \alt \progtwo : \typ$ holds
  if and only if
  $\tctx \vdash \prog : \typ$ and $\tctx \vdash \progtwo : \typ$ hold.
\end{enumerate}
\end{lemma}

\begin{proposition}[Subject reduction]
\lprop{subject_reduction}
Let $\tctx \vdash \prog : \typ$
and $\prog \toca{} \progtwo$.
Then $\tctx' \vdash \progtwo : \typ$,
where $\tctx' = \tctx$ if the step is derived using any reduction rule other than \indrulename{fresh},
and $\tctx' = (\tctx,\var:\typtwo)$
if the step introduces a fresh variable $(\var:\typtwo)$.
\end{proposition}
\begin{proof}
By case analysis on the transition $\prog \toca{} \progtwo$,
using~\rlem{type_system_basic_properties}.
The interesting case is the \indrulename{unif} case,
which requires proving that the substitution $\subst$ returned
by $\mgu{\goals}$ preserves the types of the instantiated variables.
\SeeAppendixRef{\rsec{appendix_subject_reduction}}
\end{proof}

\section{Denotational Semantics}
  \lsec{denotational_semantics}
  
In this section we propose a {\em naive} denotational semantics for the
$\lambdaunif$-calculus. The semantics is naive in at least three senses:
first, types are interpreted merely as sets,
rather than as richer structures (\eg complete partial orders)
or in a~more abstract (\eg categorical) framework.
Second, since types are interpreted as sets, the {\em multiplicities}
of results are not taken into account, so for example
$\seme{\var \alt \var}{} = \seme{\var}{} \cup \seme{\var}{} = \seme{\var}{}$.
Third, and most importantly, the denotation of abstractions ($\lam{\var}{\prog}$)
is conflated with the denotation
of allocated abstractions ($\laml{\loc}{\var}{\prog}$). This means that
the operational semantics cannot be complete with respect to the denotational
one, given that for example $\laml{\loc}{\var}{\var}$ and $\laml{\loc'}{\var}{\var}$
have the same denotation but they are not observationally equivalent\footnote{{\em E.g.}
$\laml{\loc}{\var}{\var} \unif \laml{\loc}{\var}{\var}$ succeeds but $\laml{\loc}{\var}{\var} \unif \laml{\loc'}{\var}{\var}$ fails.}.
Nevertheless, studying this simple denotational semantics already presents
some technical challenges, and we regard it as a first necessary
step towards formulating a better behaved semantics\footnote{\label{footnote_denotational_memory}We expect that
a less naive semantics should be stateful, involving a {\em memory},
in such a way that abstractions ($\lam{\var}{\prog}$) allocate a memory
cell and store a closure,
whereas allocated abstractions ($\laml{\loc}{\var}{\prog}$) denote
a memory location in which a closure is already stored.}.

Roughly speaking,
the idea is that a type $\typ$ shall be interpreted as a set $\semtyp{\typ}$,
while a program $\prog$ of type $\typ$ shall be interpreted as a subset
$\seme{\prog}{} \subseteq \semtyp{\typ}$.
For example, if $\semtyp{\Nat} = \NN$,
then given constructors $\mathbf{1} : \Nat$, $\mathbf{2} : \Nat$
with their obvious interpretations,
and if $\mathit{add} : \Nat\to\Nat\to\Nat$ denotes addition,
we expect that:
\[
  \seme{(\lam{f:\Nat\to\Nat}{\fresh{\vartwo}{((\vartwo \unif \mathbf{1})\seq\mathit{add}\,\vartwo\,(f\,\vartwo)}}))(\lam{\var}{x\alt\mathbf{2}})}{} =
  \set{1+1,1+2} = \set{2,3}
\]
The soundness result that we shall prove
states that if $\prog \rtoca{} \progtwo$ then $\seme{\prog}{} \supseteq \seme{\progtwo}{}$.
Intuitively, the possible behaviors of $\progtwo$
are among the possible behaviors of $\prog$.

To formulate the denotational semantics, for ease of notation,
we work with an {\em \`a la Church} variant of the type system\footnote{Transitioning
between Church vs. Curry style variants of this system is a straightforward
exercise, following for instance \cite[Prop.~1.2.19]{barendregt2013lambda}.}.
That is,
we suppose that the set of variables is partitioned in such a way
that each variable has an intrinsic type.
More precisely, for each type $\typ$ there is a denumerably infinite
set of variables $\var^{\typ}, \vartwo^{\typ}, \varthree^{\typ}, \hdots$
of that type.
We also decorate each occurrence of $\fail$ with its type,
\ie we write $\fail^\typ$ for the empty program of type $\typ$.
Sometimes we omit the type decoration if it is clear from the context.
Under this assumption, it is easy to show that the system enjoys
a strong form of {\em unique typing},
\ie that if $X$ is a typable term or program then
there is a unique derivation $\tctx \vdash X : \typ$,
up to weakening of $\tctx$ with variables not in $\fv{X}$.
This justifies that we may write $\vdash X : \typ$
omitting the context.
\medskip

{\bf Domain of interpretation.}
We suppose given a {\bf non-empty} set $\binterp{\btyp}$ for each base type $\btyp$.
The {\em interpretation} of a type $\typ$ is a set written $\semtyp{\typ}$
and defined recursively as follows, where $\powerset{X}$ is the
usual set-theoretic power set, and $Y^X$ is the set of functions
with domain $X$ and codomain $Y$:
\[
  \semtyp{\btyp} \eqdef \binterp{\btyp}
  \hspace{1cm}
  \semtyp{\typ \to \typtwo} \eqdef \powerset{\semtyp{\typtwo}}^{\semtyp{\typ}}
\]
Note that, for every type $\typ$, the set $\semtyp{\typ}$ is non-empty,
given that we require that $\binterp{\btyp}$ be non-empty.
This decision is not arbitrary; rather it is necessary for soundness to hold.
For instance,
operationally we have that ${\var^\typ\seq\vartwo^\typtwo} \toca{guard} \vartwo^\typtwo$,
so denotationally we would expect $\seme{\var^\typ\seq\vartwo^\typtwo}{} \supseteq \seme{\vartwo^\typtwo}{}$.
This would not hold if $\semtyp{\typ} = \emptyset$ and $\semtyp{\typtwo} \neq \emptyset$,
as then $\seme{\var^\typ\seq\vartwo^\typtwo}{} = \emptyset$
whereas $\seme{\vartwo^\typtwo}{}$ would be a non-empty set.
\medskip

Another technical constraint that we must impose
is that {\em the interpretation of a value should always be a singleton}.
For example, operationally we have that
$(\lam{\var:\Nat}{\var+\var})\,\val \,\rtoca{}\, \val + \val$,
so denotationally, by soundness, we would expect that
$\seme{(\lam{\var:\Nat}{\var+\var})\,\val}{} \supseteq \seme{\val + \val}{}$.
If we had that $\seme{\val}{} = \set{1,2}$ is not a singleton, then
we would have that $\seme{(\lam{\var}{\var + \var})\,\val}{} = \set{1+1,2+2}$
whereas $\seme{\val+\val}{} = \set{1+1,1+2,2+1,2+2}$.

Following this principle,
given that terms of the form $\cons\,\val_1\hdots\val_n$ are values,
their denotation $\seme{\cons\,\val_1\hdots\val_n}{}$ must always be a singleton.
This means that constructors must be interpreted as singletons,
and constructors of function type should always return singletons
(which in turn should return singletons if they are functions, and so on, recursively).
Formally,
any element $\obj \in \semtyp{\btyp}$ is declared to be {\bf $\btyp$-unitary},
and a function $\objfun \in \semtyp{\typ \to \typtwo}$
is {\bf $(\typ \to \typtwo)$-unitary}
if for each $\obj \in \semtyp{\typ}$
the set $\objfun(\obj) = \set{\objtwo} \subseteq \semtyp{\typtwo}$
is a singleton and $\objtwo$ is $\typtwo$-unitary.
Sometimes we say that an element $\obj$
is {\em unitary} if the type is clear from
the context.
If $\objfun$ is $(\typ \to \typtwo)$-unitary,
and $\obj \in \semtyp{\typ}$
sometimes, by abuse of notation, we may write
$\objfun(\obj)$ for the unique element
$\objtwo \in \objfun(\obj)$.

{\bf Interpretation of terms.}
For each constructor $\cons$,
we suppose given a $\constyp{\cons}$-unitary
element $\cinterp{\cons} \in \semtyp{\constyp{\cons}}$.
Moreover, we suppose that the interpretation of constructors
is {\em injective}, \ie
that
$\cinterp{\cons}(\obj_1)\hdots(\obj_n) = \cinterp{\cons}(\objtwo_1)\hdots(\objtwo_n)$
implies $\obj_i = \objtwo_i$ for all $i=1..n$.

An {\em environment} is a function
$\asg : \Var \to \bigcup_{\typ \in \Type} \semtyp{\typ}$
such that $\asg(\var^\typ) \in \semtyp{\typ}$
for each variable $\var^\typ$ of each type $\typ$.
If $\asg$ is an environment and $\obj \in \semtyp{\typ}$,
we write $\asg\asgextend{\var^\typ}{\obj}$
for the environment that maps $\var^\typ$ to $\obj$
and agrees with $\asg$ on every other variable.
We write $\Asg$ for the set of all environments.

Let $\vdash \tm : \typ$ (resp. $\vdash \prog : \typ$)
be a typable term (resp. program)
and let $\asg$ be an environment.
If $\vdash X : \typ$ is a typable term or program,
we define its {\em denotation under the environment $\asg$},
written $\seme{X}{\asg}$ as a subset of $\semtyp{\typ}$ as follows:
\[
{\small
  \begin{array}{rcl}
    \seme{\var^\typ}{\asg}
  & \eqdef &
    \set{\asg(\var^\typ)}
  \\
    \seme{\cons}{\asg}
  & \eqdef &
    \set{\cinterp{\cons}}
  \\
    \seme{\lam{\var^\typ}{\prog}}{\asg}
  & \eqdef &
    \set{\objfun}
    \text{\HS
      where $\objfun : \semtyp{\typ} \to \powerset{\semtyp{\typtwo}}$
      is given by
      $\objfun(\obj) = \seme{\prog}{\asg\asgextend{\var^\typ}{\obj}}$}
  \\
    \seme{\laml{\loc}{\var^\typ}{\prog}}{\asg}
  & \eqdef &
    \set{\objfun}
    \text{\HS
      where $\objfun : \semtyp{\typ} \to \powerset{\semtyp{\typtwo}}$
      is given by
      $\objfun(\obj) = \seme{\prog}{\asg\asgextend{\var^\typ}{\obj}}$}
  \\
    \seme{\tm\,\tmtwo}{\asg}
  & \eqdef &
   \set{\objtwo \ST
          \exists\objfun \in \seme{\tm}{\asg},\ %
          \exists\obj \in \seme{\tmtwo}{\asg},\ %
          \objtwo \in \objfun(\obj)}
  \\
    \seme{\tm \unif \tmtwo}{\asg}
  & \eqdef &
    \set{\cinterp{\unit} \ST
           \exists\obj \in \seme{\tm}{\asg},\ %
           \exists\objtwo \in \seme{\tmtwo}{\asg},\ %
           \obj = \objtwo}
  \\
    \seme{\tm \seq \tmtwo}{\asg}
  & \eqdef &
    \set{\obj \ST
           \exists\objtwo \in \seme{\tm}{\asg},\ %
           \obj \in \seme{\tmtwo}{\asg}}
  \\
    \seme{\fresh{\var^\typ}{\tm}}{\asg}
  & \eqdef &
    \set{\objtwo \ST
      \exists\obj \in \semtyp{\typ},\ %
      \objtwo \in \seme{\tm}{\asg\asgextend{\var^\typ}{\obj}}
    }
  \\
    \seme{\fail^\typ}{\asg}
  & \eqdef &
    \emptyset
  \\
    \seme{\tm \alt \prog}{\asg}
  & \eqdef &
    \seme{\tm}{\asg} \cup \seme{\prog}{\asg}
  \\
  \end{array}
}
\]
The denotation of a toplevel program is written $\seme{\prog}{}$
and defined as the union of its denotations under all possible environments,
\ie
$
  \seme{\prog}{} \eqdef \bigcup_{\asg \in \Asg} \seme{\prog}{\asg}
$.

\begin{proposition}[Properties of the denotational semantics]
\lprop{properties_of_the_denotation}
\llem{irrelevance}
\llem{compositionality}
\llem{interpretation_of_values}
\llem{interpretation_of_substitution}
\begin{enumerate}
\item {\bf Irrelevance.}
  If $\asg$ and $\asg'$ agree on $\fv{X}$,
  then $\seme{X}{\asg} = \seme{X}{\asg'}$.
  Here $X$ stands for either a program or a term.
  \SeeAppendixRef{\rlem{appendix:irrelevance}}
\item {\bf Compositionality.} \SeeAppendixRef{\rlem{appendix:compositionality}}
  \begin{enumerate}
  \item
    $\seme{\prog \alt \progtwo}{\asg} =
     \seme{\prog}{\asg} \cup \seme{\progtwo}{\asg}$.
  \item
    If $\wctx$ is a context whose hole is
    of type $\typ$, then
    $\seme{\wctxof{\tm}}{\asg} =
     \set{\objtwo \ST
          \obj \in \seme{\tm}{\asg},
          \objtwo \in \seme{\wctx}{\asg\asgextend{\ctxhole^\typ}{\obj}}}$.
  \end{enumerate}
\item {\bf Interpretation of values.}
  If $\val$ is a value then $\seme{\val}{\asg}$ is a singleton.
  \SeeAppendixRef{\rlem{appendix:interpretation_of_values}}
\item {\bf Interpretation of substitution.}
  \SeeAppendixRef{\rlem{appendix:interpretation_of_substitution}}\\
  Let $\subst = \set{\var^{\typ_1}_1\mapsto\val_1,\hdots,\var^{\typ_n}_n\mapsto\val_n}$
  be a substitution
  such that $\var_i \notin \fv{\val_j}$ for all $i,j$.
  Let $\seme{\val_i}{\asg} = \set{\obj_i}$ for each $i=1..n$
  (noting that values are singletons, by the previous item of this lemma).
  Then for any program or term $X$
  we have that
  $
   \seme{X\SUB{\subst}}{\asg} =
   \seme{X}{\asg\asgextend{\var_1}{\obj_1}\hdots\asgextend{\var_n}{\obj_n}}
  $.
\end{enumerate}
\end{proposition}
\bigskip

To conclude this section,
the following theorem shows that the operational semantics is
sound with respect to the denotational semantics.
\begin{theorem}[Soundness]
\lthm{soundness}
Let $\tctx \vdash \prog : \typ$
and $\prog \toca{} \progtwo$.
Then $\seme{\prog}{} \supseteq \seme{\progtwo}{}$.
The inclusion is an equality for all reduction rules other than the
\indrulename{fail} rule.
\end{theorem}
\begin{proof}
The proof (\SeeAppendixRef{\rthm{appendix:soundness}})
is technical by exhaustive case analysis of all possible reduction
steps, using \rprop{properties_of_the_denotation} throughout.
The \indrulename{unif} rule is non-trivial, as it requires to formulate
an invariant for the unification algorithm.
The core of the argument is an auxiliary lemma essentially stating that if
$\goals \tounifa{} \goalstwo$ is a step of the unification algorithm
that does not fail,
then the set of environments that fulfill the equality constraints imposed by $\goals$
are the same environments that fulfill the equality constraints imposed by $\goalstwo$.
\end{proof}

\begin{example}
Consider the reduction
$
  \fresh{\var}{\left((\lam{\varthree}{\fresh{\vartwo}{((\varthree \unif \pairing\,{\bf 1}\,\vartwo)\seq(\pairing\,\vartwo\,\var))}})\,(\pairing\,x\,{\bf 2})\right)}
  \rtoca{}
  \pairing\,{\bf 2}\,{\bf 1}
$.
If $\semtyp{\mathtt{Tuple}} = \semtyp{\Nat} \times \semtyp{\Nat} = \NN \times \NN$,
the constructors ${\bf 1} : \Nat$, ${\bf 2} : \Nat$ are given their
obvious interpretations and $\pairing : \Nat \to \Nat \to \mathtt{Tuple}$
is the pairing function\footnote{Precisely, $\cinterp{\pairing}(n) = \set{f_n}$ with $f_n(m) = \set{(n,m)}$.}, then for any environment $\asg$,
if we abbreviate
$\asg' := \asg\asgextend{\var}{n}\asgextend{\varthree}{p}\asgextend{\vartwo}{m}$,
we have:
\[
{\small
  \begin{array}{rcl}
  && \seme{\fresh{\var}{\left((\lam{\varthree}{\fresh{\vartwo}{((\varthree \unif \pairing\,{\bf 1}\,\vartwo)\seq(\pairing\,\vartwo\,\var))}})\,(\pairing\,x\,{\bf 2})\right)}}{\asg}
  \\
  & = &
    \set{
      \seme{(\lam{\varthree}{\fresh{\vartwo}{((\varthree \unif \pairing\,{\bf 1}\,\vartwo)\seq(\pairing\,\vartwo\,\var))}})\,(\pairing\,x\,{\bf 2})
      }{
        \asg\asgextend{\var}{n}
      }
    \ST
      n\in\NN
    }
  \\
  & = &
    \set{
      r
    \ST
      n \in\NN,
      f \in \seme{\lam{\varthree}{\fresh{\vartwo}{((\varthree \unif \pairing\,{\bf 1}\,\vartwo)\seq(\pairing\,\vartwo\,\var))}}}{\asg\asgextend{\var}{n}},
      p \in \seme{\pairing\,\var\,{\bf 2}}{\asg\asgextend{\var}{n}},
      r \in f(p)
    }
  \\
  & = &
    \set{
      r
    \ST
      n, m \in\NN,
      p \in \seme{\pairing\,\var\,{\bf 2}}{\asg\asgextend{\var}{n}},
      r \in \seme{(\varthree \unif \pairing\,{\bf 1}\,\vartwo)\seq(\pairing\,\vartwo\,\var)}{\asg'}
    }
  \\
  & = &
    \set{
      r
    \ST
      n, m \in \NN,
      p \in \set{(n, 2)},
      r \in \seme{(\varthree \unif \pairing\,{\bf 1}\,\vartwo)\seq(\pairing\,\vartwo\,\var)}{\asg'}
    }
  \\
  & = &
    \set{
      r
    \ST
      n, m \in \NN,
      p \in \set{(n, 2)},
      b \in \seme{\varthree \unif \pairing\,{\bf 1}\,\vartwo}{\asg'},
      r \in \seme{\pairing\,\vartwo\,\var}{\asg'}
    }
  \\
  & = &
    \set{
      r
    \ST
      n, m \in \NN,
      p \in \set{(n, 2)},
      p = (1,m),
      r \in \seme{\pairing\,\vartwo\,\var}{\asg'}
    }
  \\
  & = &
    \set{
      r
    \ST
      n \in \set{1}, m \in \set{2}, p \in \set{(1, 2)},
      r \in \seme{\pairing\,\vartwo\,\var}{\asg'}
    }
  \\
  & = &
    \set{(2,1)}
  \\
  & = &
    \seme{\pairing\,{\bf 2}\,{\bf 1}}{\asg}
  \end{array}
}
\]
\end{example}
An example in which the inclusion is proper
is the reduction step
$\laml{\loc}{\var}{\var} \unif \laml{\loc'}{\var}{\var} \,\toca{fail}\, \fail$.
Note that
$\seme{\laml{\loc}{\var}{\var} \unif \laml{\loc'}{\var}{\var}}{} = \set{\cinterp{\unit}}
\supsetneq \emptyset = \seme{\fail}{}$,
given that our naive semantics equates the denotations of the abstractions,
\ie $\seme{\laml{\loc}{\var}{\var}} = \seme{\laml{\loc'}{\var}{\var}}{}$,
in spite of the fact that their locations differ.

\section{Conclusion}
  
In this work, we have proposed the {\bf $\lambdaunif$-calculus}~(\rdef{operational_semantics_reduction_rules})
an extension of the $\lambda$-calculus with
relational features, including non-deterministic choice
and first-order unification.
We have studied some of its operational properties,
providing an inductive {\bf characterization of normal forms}~(\rprop{characterization_normal_forms}),
and proving that
it is {\bf confluent}~(\rthm{confluence}) up to structural equivalence,
by adapting the technique by Tait and Martin-L\"of.
We have proposed a system of simple types enjoying
{\bf subject reduction}~(\rprop{subject_reduction}).
We have also proposed a naive denotational semantics, in which a program of
type $\typ$ is interpreted as a set of elements of a set $\semtyp{\typ}$,
for which we have proven {\bf soundness}~(\rthm{soundness}).
The denotational semantics is not complete.

As of the writing of this paper,
we are attempting to formulate a refined denotational semantics
involving a notion of {\em memory}, following the ideas mentioned in
footnote~\ref{footnote_denotational_memory}. One difficulty is that
in a term like
$
 ((\var \unif \lam{\varthree}{\varthree}) \seq \vartwo)
 ((\vartwo \unif \lam{\varthree}{\varthree}) \seq \var)
$,
there seems to be a cyclic dependency between
the denotation of the subterm on the left and denotation of the subterm on the right, so
it is not clear how to formulate the semantics compositionally.

We have attempted to prove normalization results for the simply typed
system, until now unsuccessfully.
Given a constructor $\cons : (\typ \to \typ) \to \typ$,
a self-looping term $\omega(\cons\,\omega)$ with
$\omega \eqdef
  \lam{\var^\typ}{\fresh{\vartwo^{\typ\to\typ}}{
    ((\cons\vartwo \unif \var) \seq \vartwo\,\var)
  }}$
can be built,
so some form of {\em positivity condition} should be imposed.
Other possible lines for future work include studying the relationship
between calculi with patterns and $\lambdaunif$ by means of translations,
and formulating richer type systems. For instance,
one would like to be able to express {\em instantiation restrictions},
in such a way that a fresh variable representing a natural number
is of type $\texttt{Nat}^-$ while a term of type $\texttt{Nat}^+$ represents
a fully instantiated natural number.

{\bf Related Work.}
On {\bf functional--logic} programming,
we have mentioned $\lambda$Prolog~\cite{nadathur1984higher,miller2012programming}
and Curry~\cite{DBLP:conf/popl/Hanus97,Hanus13}.
Other languages combining functional and logic features are
Mercury~\cite{somogyi1996execution} and Mozart/Oz~\cite{van2005multiparadigm}.
There is a vast amount of literature on functional--logic programming.
We mention a few works which most resemble our own.
Miller~\cite{miller1991logic} proposes a language with lambda-abstraction
and a decidable extension of first-order unification which admits most general unifiers.
Chakravarty \etal~\cite{chakravarty1998goffin} and Smolka~\cite{smolka1997foundation}
propose languages in which the functional--logic paradigm is modeled as a concurrent
process with communication.
Albert \etal~\cite{albert2002operational} formulate a {\em big-step} semantics
for a functional--logic calculus with narrowing.
On pure {\bf relational} programming (without $\lambda$-abstractions),
recently Rozplokhas \etal~\cite{rozplokhas2019certified}
have studied the operational and denotational semantics of miniKanren.
On {\bf $\lambda$-calculi with patterns} (without full unification),
there have been many different approaches to their
formulation~\cite{jay2006pure,arbiser2006lambda,klop2008lambda,DBLP:journals/corr/abs-1009-3429,ayala2019typed}.
On {\bf $\lambda$-calculi with non-deterministic choice} (without unification),
we should mention works on the $\lambda$-calculus
extended with {\em erratic}~\cite{schmidt2000lambda}
as well as with {\em probabilistic} choice~\cite{ramsey2002stochastic,DBLP:conf/lics/FaggianR19}.

\noindent{\bf Acknowledgements.}
To Alejandro Díaz-Caro for supporting our interactions.
To Eduardo Bonelli, Delia Kesner, and the anonymous reviewers 
for their feedback and suggestions.

\newpage
\appendix
\renewcommand*\theequation{\thechapter.\thetheorem.\arabic{equation}}
\section{Technical Appendix}

The following lemma summarizes some expected properties of substitution that
we use throughout the appendix. We omit the proofs, which are routine:
\begin{lemma}[Properties of substitution]
\llem{properties_of_substitution}
Let $\subst$ be an arbitrary substitution. Then:
\begin{enumerate}
\item
  $\wctxof{\tm}\SUB{\subst} = \wctx\SUB{\subst}\ctxof{\tm\SUB{\subst}}$.
  Note that there cannot be capture, given that $\wctx$ is a weak context,
  and it does not bind variables.
\item $(\tm\SUB{\subst})\SUB{\substtwo} = \tm^{\subst \scomp \substtwo}$
\item $\tm\sub{\var}{\val}\SUB{\subst} = \tm\SUB{\subst}\sub{\var}{\val\SUB{\subst}}$
      as long as there is no capture, \ie $\var \not\in \supp{\subst}$ and
      for all $\vartwo \in \fv{\tm}$ we have that $\var \not\in \fv{\subst(\vartwo)}$.
\item If $\val$ is a value then $\val\SUB{\subst}$ is a value.
\item The relation $\sleq$ is a preorder, \ie reflexive and transitive.
\end{enumerate}
\end{lemma}

\subsection{Unification Algorithm}
\lsec{appendix_unification_algorithm}

We define the free variables ($\fv{\goals}$), locations ($\locs{\goals}$),
and capture-avoiding substitution ($\goals\sub{\var}{\tm}$)
for goals as follows:
\[
  \begin{array}{rcl}
    \fv{\set{\val_1\unif\valtwo_1,\hdots,\val_n\unif\valtwo_n}}
    & \eqdef &
    \fv{\val_1\unif\valtwo_1} \cup \hdots \cup \fv{\val_n\unif\valtwo_n}
  \\
    \locs{\set{\val_1\unif\valtwo_1,\hdots,\val_n\unif\valtwo_n}}
    & \eqdef &
    \locs{\val_1\unif\valtwo_1} \cup \hdots \cup \locs{\val_n\unif\valtwo_n}
   \\
    \set{\val_1\unif\valtwo_1,\hdots,\val_n\unif\valtwo_n}\sub{\var}{\tm}
    & \eqdef &
    \set{(\val_1\unif\valtwo_1)\sub{\var}{\tm},\hdots,(\val_n\unif\valtwo_n)\sub{\var}{\tm}}
  \end{array}
\]

\begin{definition}[Unification algorithm]
The following is a variant of Martelli--Montanari's unification algorithm.
We say that two values $\val,\valtwo$ {\em clash} if any of the following conditions
holds:
\begin{enumerate}
\item Constructor clash:
  $\val = \cons\,\val_1\hdots\val_n$ and $\valtwo = \constwo\,\valtwo_1\hdots\valtwo_m$ 
  with $\cons \neq \constwo$.
\item Arity clash:
  $\val = \cons\,\val_1\hdots\val_n$ and $\valtwo = \cons\,\valtwo_1\hdots\valtwo_m$ 
  with $n \neq m$.
\item Type clash:
  $\val = \cons\,\val_1\hdots\val_n$ and $\valtwo = \laml{\loc}{\var}{\prog}$
  or vice-versa.
\item Location clash:
  $\val = \laml{\loc}{\var}{\prog}$ and $\valtwo = \laml{\loctwo}{\vartwo}{\progtwo}$ 
  with $\loc \neq \loctwo$.
\end{enumerate}
We define a rewriting system whose objects are unification problems $\goals$, and the symbol $\FAIL$.
The binary rewriting relation $\tounifa{}$ is given by the union of the following rules.
Note that ``$\uplus$'' stands for the disjoint union of sets:
\[\FIT{$
  \begin{array}{rlll}
    \set{\var \unif \var} \uplus \goals
    & \tounifa{u-delete} &
    \goals
  \\
    \set{\val \unif \var} \uplus \goals
    & \tounifa{u-orient} &
    \set{\var \unif \val} \uplus \goals
    & \text{if $\val \notin \Var$}
  \\
    \set{\laml{\loc}{\var}{\prog} \unif \laml{\loc}{\var}{\prog}} \uplus \goals
    & \tounifa{u-match-lam} &
    \goals
  \\
    \set{\cons\,\val_1\hdots\val_n \unif \cons\,\valtwo_1\hdots\valtwo_n} \uplus \goals
    & \tounifa{u-match-cons} &
    \set{\val_1 \unif \valtwo_1, \hdots, \val_n \unif \valtwo_n} \uplus \goals
  \\
    \set{\val \unif \valtwo} \uplus \goals
    & \tounifa{u-clash} &
    \FAIL
    & \text{if $\val$ and $\valtwo$ clash}
  \\
    \set{\var \unif \val} \uplus \goals
    & \tounifa{u-eliminate} &
    \set{\var \unif \val} \uplus \goals\sub{\var}{\val}
    & \text{if $\var \in \fv{\goals} \setminus \fv{\val}$}
  \\
    \set{\var \unif \val} \uplus \goals
    & \tounifa{u-occurs-check} &
    \FAIL
    & \text{if $\var \neq \val$ and $\var \in \fv{\val}$}
  \\
  \end{array}
$}\]
\end{definition}

\begin{lemma}[Coherence is invariant by unification]
\llem{location_coherence_unification}
If $\goals$ is a coherent unification problem
and $\goals \tounifa{} \goalstwo$ then $\goalstwo$ is
coherent.
\end{lemma}
\begin{proof}
By inspection of the unification rules.
The only interesting case is the \rulename{u-eliminate} rule:
\[
    \set{\var \unif \val} \cup \goals
    \HS\tounifa{u-eliminate}\HS
    \set{\var \unif \val} \cup \goals\sub{\var}{\val}
\]
Consider two abstractions
$\laml{\loc}{\vartwo}{\tm}$ and $\laml{\loctwo}{\varthree}{\tmtwo}$
in $\goals$  such that, after performing the substitution
$(\laml{\loc}{\vartwo}{\tm})\sub{\var}{\val} = (\laml{\loctwo}{\varthree}{\tmtwo})\sub{\var}{\val}$
they have the same location, \ie $\loc = \loctwo$.
Then since $\goals$ is coherent we have that
$\laml{\loc}{\vartwo}{\tm} = \laml{\loctwo}{\varthree}{\tmtwo}$,
and this means that
$(\laml{\loc}{\vartwo}{\tm})\sub{\var}{\val} = (\laml{\loctwo}{\varthree}{\tmtwo})\sub{\var}{\val}$,
as required.
\end{proof}

\begin{theorem}
\lthm{most_general_unifier}
\lthm{appendix_mgu}
Consider the relation $\tounifa{}$ restricted to {\em coherent} unification problems~(\rlem{location_coherence_unification}).
Then:
\begin{enumerate}
\item
  The relation $\tounifa{}$ is strongly normalizing.
\item
  The normal forms of $\tounifa{}$ are $\FAIL$ and sets of goals of the form
  $\set{\var_1 \unif \val_1, \hdots, \var_n \unif \val_n}$
  where $\var_i \neq \var_j$ and $\var_i \notin \fv{\val_j}$ for every $i,j \in 1..n$.

  If the normal form of $\goals$ is $\set{\var_1 \unif \val_1, \hdots, \var_n \unif \val_n}$,
  we say that $\mgu{\goals}$ exists, and
  $\mgu{\goals} = \set{\var_1 \mapsto \val_1, \hdots, \var_n \mapsto \val_n}$.
  If the normal form is $\FAIL$, we say that $\mgu{\goals}$ fails.  
\item
  The substitution $\subst = \mgu{\goals}$ exists if and only if there exists a unifier for $\goals$.
  When it exists, $\mgu{\goals}$ is an {\em idempotent most general unifier}.
  Moreover:
  \begin{enumerate}
  \item The set $\goals\SUB{\subst} \cup \set{\subst(\var) \ST \var \in \Var}$ is coherent.
  \item For any $\var \in \Var$ and any allocated abstraction $\laml{\loc}{\vartwo}{\prog}$ in $\subst(\var)$,
        the location $\loc$ decorates an allocated abstraction in $\goals$.
  \end{enumerate}
\end{enumerate}
\end{theorem}
\begin{proof}
A straightforward adaptation of standard results, see for example
\cite[Section~4.6]{baader1999term}.
We only focus in the interesting differences, namely the two subitems of item 3.:
\begin{enumerate}
\item
  Let us write $\goals \tounifa{}^*_\substtwo \goals'$ if there is a sequence
  of unification steps from $\goals$ to $\goals'$ such that
  $\substtwo$ is the composition of all the substitutions
  performed in the \rulename{u-eliminate} steps.

  We claim that if $\goals \tounifa{}^*_\substtwo \goals'$ then
  $\goals\SUB{\substtwo} \cup \goals'$ is coherent.
  By induction on the length of the sequence.
  The empty case is immediate, so let us suppose that
  $\goals \tounifa{}^*_\substthree \goals'' \tounifa{} \goals'$.
  By \ih, $\goals\SUB{\substthree} \cup \goals''$ is coherent.
  Consider two cases, depending on whether the step
  $\goals'' \tounifa{} \goals'$ is an \rulename{u-eliminate} step or not.
  \begin{enumerate}
  \item
    If it is an \rulename{u-eliminate} step, substituting a variable $\var$ for a value $\val$,
    then we also have a step
    $\goals\SUB{\substthree} \cup \goals'' \tounifa{u-eliminate} \goals\SUB{\substthree\scomp(\var \mapsto \val)} \cup \goals'$
    and by \rlem{location_coherence_unification}
    we have that $\goals\SUB{\substthree\scomp(\var \mapsto \val)} \cup \goals'$ is coherent,
    as required.
  \item
    If it is not an \rulename{u-eliminate} step, then
    we also have a step
    $\goals\SUB{\substthree} \cup \goals'' \tounifa{} \goals\SUB{\substthree} \cup \goals'$
    and by \rlem{location_coherence_unification}
    we have that $\goals\SUB{\substthree} \cup \goals'$ is coherent,
    as required.
  \end{enumerate}
  From this claim we have that if
  $\goals \tounifa{}^* \set{\var_1 \unif \val_1, \hdots, \var_n \unif \val_n}$
  and $\subst := \set{\var_1 \mapsto \val_1, \hdots, \var_n \mapsto \val_n}$
  then $\goals^\subst \cup \set{\var_1 \unif \val_1, \hdots, \var_n \unif \val_n}$ is coherent,
  which entails the required property.
\item
  We claim that if $\goals \tounifa{}^* \goals'$ then
  for any allocated abstraction $\laml{\loc}{\var}{\prog}$ in $\goals'$,
  the location $\loc$ decorates an allocated abstraction in $\goals$.
  This is straightforward to prove by induction on the length of the
  reduction sequence, and it entails the required property.
\end{enumerate}
\end{proof}

\subsection{Properties of Most General Unifiers}

\begin{lemma}[Properties of most general unifiers]
\llem{properties_of_mgus}
\quad
\begin{enumerate}
\item If $\subst, \subst'$ are idempotent most general unifiers of $\goals$, there is a {\em renaming},
      \ie a substitution of the form $\substtwo = \set{\var_1 \mapsto \vartwo_1, \hdots, \var_n \mapsto \vartwo_n}$,
      such that $\subst' = \subst \scomp \substtwo$.
\item If $\subst$ is an idempotent most general unifier of $\goals$
      and $\vartwo \not\in \fv{\goals}$,
      then $\subst' := (\vartwo \mapsto \var) \scomp \subst$ is an idempotent most general unifier of $\goals\sub{\var}{\vartwo}$.
\item If $\subst$ is an idempotent most general unifier of $\goals$
      and $\loc' \not\in \locs{\goals}$
      then the substitution $\subst'$ given by $\subst'(\var) = \subst(\var)\sub{\loc}{\loc'}$
      is an idempotent most general unifier of $\goals\sub{\loc}{\loc'}$.
\end{enumerate}
\end{lemma}
\begin{proof}
We prove each item:
\begin{enumerate}
\item A standard result, see for example \cite[Section~4.6]{baader1999term}.
\item Indeed:
  \begin{enumerate}
  \item {\em Unifier:} 
    For each goal $(\val \unif \valtwo) \in \goals$, we have that
    $\val\sub{\var}{\vartwo}\SUB{\subst'} = \val\SUB{\subst} = \valtwo\SUB{\subst} = \valtwo\sub{\var}{\vartwo}\SUB{\subst'}$
    since $\vartwo \not\in \fv{\goals}$ and $\subst$ is a unifier of $\goals$.
  \item {\em Most general:}
    Let $\substtwo$ be a unifier of $\goals\sub{\var}{\vartwo}$,
    \ie such that $\val\sub{\var}{\vartwo}\SUB{\substtwo} = \valtwo\sub{\var}{\vartwo}\SUB{\substtwo}$
    for every goal $(\val \unif \valtwo) \in \goals$.
    Then it is easily checked $(\var \mapsto \vartwo) \scomp \substtwo$ is a unifier of $\goals$.
    Since $\subst$ is a most general unifier of $\goals$, we have that $(\var \mapsto \vartwo) \scomp \substtwo = \subst \scomp \substthree$
    for some $\substthree$.
    Hence $\substtwo
          = (\vartwo \mapsto \var) \scomp (\var \mapsto \vartwo) \scomp \substtwo
          = (\vartwo \mapsto \var) \scomp \subst \scomp \substthree
          = \subst' \scomp \substthree
          $ as required.
  \end{enumerate}
\item
  It suffices to observe that if
  $\goals \tounifa{} \goals'$
  then
  $\goals\sub{\loc}{\loc'} \tounifa{} \goals'\sub{\loc}{\loc'}$.
  This is easy to check for each rule. The only noteworthy remark is that
  in the \rulename{u-clash} we have that if $\val$ and $\valtwo$ have a location clash, then
  $\val\sub{\loc}{\loc'}$ and $\valtwo\sub{\loc}{\loc'}$ also have a location clash,
  because $\loc' \not\in \locs{\goals}$.

  Then by induction on the number of $\tounifa{}$ steps,
  we have that if the normal form of $\goals$ is $\set{\var_1 \unif \val_1, \hdots, \var_n \unif \val_n}$,
  then the normal form of $\goals\sub{\loc}{\loc'}$ is $\set{\var_1 \unif \val_1\sub{\loc}{\loc'}, \hdots, \var_n \unif \val_n\sub{\loc}{\loc'}}$.
\end{enumerate}
\end{proof}

\begin{lemma}[Compositionality of most general unifiers]
\llem{mgu_compositional}
The following are equivalent:
\begin{enumerate}
\item $\subst = \mgu{\goals \cup \goalstwo}$ exists.
\item $\subst_1 = \mgu{\goals}$ and $\subst_2 = \mgu{\goalstwo\SUB{\subst_1}}$ both exist.
\end{enumerate}
Moreover, if $\subst,\subst_1,\subst_2$ exist, then $\subst = \subst_1\scomp\subst_2\scomp\substtwo$
for some renaming $\substtwo$.
\end{lemma}
\begin{proof}
\quad
\begin{itemize}
\item[] ($1 \implies 2$)
  Let $\subst = \mgu{\goals \cup \goalstwo}$.
  Note in particular that $\subst$ is a unifier for $\goals$, so $\subst_1 = \mgu{\goals}$ exists by \rthm{appendix_mgu}.
  On the other hand, note that $\subst_1$ is more general than $\subst$, so $\subst = \subst_1 \cdot \substthree$
  for some substitution $\substthree$.
  Since $\subst$ is a unifier for $\goalstwo$, we have that $\substthree$ is a unifier for $\goalstwo\SUB{\subst_1}$.
  This means that $\subst_2 = \mgu{\goalstwo\SUB{\subst_1}}$ exists by \rthm{appendix_mgu}.
\item[] ($2 \implies 1$)
  We claim that $\subst_1\scomp\subst_2$ is a unifier of $\goals \cup \goalstwo$.
  Indeed, note if $\val \unif \valtwo$ is a goal in $\goals$
  we have that $\subst_1$ is a unifier for $\goals$,
  so $\val\SUB{\subst_1} = \valtwo\SUB{\subst_1}$ and
  $\val\SUB{\subst_1\scomp\subst_2} = \valtwo\SUB{\subst_1\scomp\subst_2}$.
  Moreover, if $\val \unif \valtwo$ is a goal in $\goalstwo$,
  then $\val\SUB{\subst_1} \unif \valtwo\SUB{\subst_1}$ is a goal in $\goalstwo\SUB{\subst_1}$,
  and since $\subst_2$ is a unifier for $\goalstwo\SUB{\subst_1}$
  we conclude that $\val\SUB{\subst_1\scomp\subst_2} \unif \valtwo\SUB{\subst_1\scomp\subst_2}$,
  as required.
\end{itemize}
For the final property in the statement,
by \rlem{properties_of_mgus}, it suffices to show that $\subst_1\scomp\subst_2$
is more general than $\subst$.
Indeed, since $\subst$ is a unifier of $\goals$,
we have that $\subst = \subst_1\scomp\substthree$ for some substitution $\substthree$,
and since $\substthree$ is a unifier of $\goalstwo\SUB{\subst_1}$,
we have that $\substthree = \subst_2\scomp\substthree'$ for some substitution $\substthree'$,
then $\subst = \subst_1\scomp\subst_2\scomp\substthree$, which means that
$\subst_1\scomp\subst_2$ is more general than $\subst$.
\end{proof}

\subsection{Proof of \rlem{location_coherence_invariant} --- Coherence Invariant}
\lsec{appendix_location_coherence_invariant}

\begin{proof}
Item 1. is immediate by inspection of all the possible rules defining $\structeq$.
For item 2., rules \rulename{guard}, \rulename{fresh}, and \rulename{fail} are immediate.
Let us analyze the remaining cases:
\begin{enumerate}
\item \rulename{alloc}:
  $\wctxof{\lam{\var}{\tm}} \to \wctxof{\laml{\loc}{\var}{\tm}}$.
  Immediate,
  as evaluation is under a weak context $\wctx$, so the newly allocated abstraction has no
  variables bound by $\wctx$.
  Moreover the new location is fresh so there are no other abstractions in the same location,
  and the rest of the program remains unmodified.
\item \rulename{beta}:
  $\wctxof{(\laml{\loc}{\var}{\prog})\val} \to \wctxof{\prog\sub{\var}{\val}}$.
  First consider an allocated abstraction $\laml{\loc'}{\vartwo}{\progtwo}$
  in $\wctxof{\prog\sub{\var}{\val}}$ and let us show that it has no variables
  bound by the context.
  If it is disjoint from the contracted redex, it is immediate.
  If it is in $\prog$, \ie $\prog = \gctxof{\laml{\loc'}{\vartwo}{\progtwo'}}$
  then $\laml{\loc'}{\vartwo}{\progtwo'}$ has no variables bound by $\gctx$,
  so $\laml{\loc'}{\vartwo}{\progtwo} = \laml{\loc'}{\vartwo}{\progtwo'\sub{\var}{\val}}$
  also has no variables bound by $\gctx$.
  If it is inside one of the copies of $\val$, then it 
  also has no variables bound by $\gctx$, as substitution is capture-avoiding.

  Consider any two abstractions $\laml{\loc'}{\vartwo}{\progtwo}$ and $\laml{\loc'}{\vartwo}{\progthree}$
  in $\wctxof{\prog\sub{\var}{\val}}$ such that they have the same location,
  and consider three cases, depending on the positions of the lambdas:
  \begin{enumerate}
  \item
    If each lambda lies inside $\wctx$ or inside one of the copies of $\val$,
    then they can be traced back to abstractions in the term on the left-hand side,
    so $\progtwo = \progthree$ by hypothesis.
  \item
    If the lambdas are both in $\prog$, \ie $\prog = \gctxof{\laml{\loc'}{\vartwo}{\progtwo} \SEP \laml{\loc'}{\vartwo}{\progthree}}$
    then $\progtwo = \progthree$ by hypothesis.
    Moreover, note that by the invariant $\var \not\in \fv{\progtwo}\cup\fv{\progthree}$,
    so the lambdas in the reduct are equal.
  \item
    If one lambda is in $\prog$, \ie $\prog = \gctxof{\laml{\loc'}{\vartwo}{\progtwo}}$,
    and the other one in $\wctx$ or in a copy of $\val$,
    note that by the invariant $\var \not\in \fv{\progtwo}$,
    so $(\laml{\loc'}{\vartwo}{\progtwo})\sub{\var}{\val} = \laml{\loc'}{\vartwo}{\progtwo}$,
    so the lambdas in the reduct are equal.
  \end{enumerate}
\item \rulename{unif}:
  $\wctxof{\val \unif \valtwo} \to \wctxof{\unit}\SUB{\subst}$.
  First consider an allocated abstraction $\laml{\loc}{\var}{\prog}$ in $\wctxof{\unit}\SUB{\subst}$.
  Then $\wctxof{\unit} = \gctxof{\laml{\loc}{\var}{\prog'}}$
  such that $\prog'\SUB{\subst} = \prog$.
  Note that $\prog'$ has no variables bound by $\gctx$,
  so $\prog'\SUB{\subst}$ also has no variables bound by $\gctx$,
  given that substitution is capture-avoiding, and $\subst$ is coherent.

  Consider moreover any two allocated abstractions $\laml{\loc}{\var}{\prog}$ and $\laml{\loc}{\var}{\progtwo}$,
  in $\wctxof{\unit}\SUB{\subst}$ such that they have the same location, and consider three cases
  depending on the positions of the lambdas:
  \begin{enumerate}
  \item
    If the lambdas are both in $\wctx$, then their bodies trace back to the term on the left-hand
    side, $\laml{\loc}{\var}{\prog_0}$ and $\laml{\loc}{\var}{\progtwo_0}$, so $\prog_0 = \progtwo_0$
    are equal by hypothesis, and moreover $\prog = \prog_0\SUB{\subst} = \progtwo_0\SUB{\subst} = \progtwo$,
    as required.
  \item
    If one lambda is in $\wctx$ and the other one in $\subst(\vartwo)$ for some variable $\vartwo \in \fv{\wctxof{\unit}}$,
    suppose without loss of generality that the position of the lambda of $\laml{\loc}{\var}{\prog}$
    is inside $\wctx$. Then there is an abstraction $\laml{\loc}{\var}{\prog_0}$ in
    the term of the left-hand side of the rule such that $\prog = \prog_0\SUB{\subst}$.
    Moreover, $\laml{\loc}{\var}{\progtwo}$ is an abstraction of the term $\subst(\vartwo)$.
    By \rthm{most_general_unifier}, there must be an abstraction $\laml{\loc}{\var}{\progtwo_0}$
    of $\val \unif \valtwo$ such that, moreover, $\progtwo_0\SUB{\subst} = \progtwo$.
    Then since $\laml{\loc}{\var}{\prog_0}$ and $\laml{\loc}{\var}{\progtwo_0}$ are abstractions
    on the left-hand side, we have by hypothesis that $\prog_0 = \progtwo_0$, hence
    $\prog = \prog_0\SUB{\subst} = \progtwo_0\SUB{\subst} = \progtwo$, as required.
  \item
    If the lambdas are in the terms $\subst(\vartwo)$ and $\subst(\varthree)$,
    for certain variables $\vartwo, \varthree \in \fv{\wctxof{\unit}}$,
    then by \rthm{most_general_unifier}, there must terms $\laml{\loc}{\var}{\prog_0}$
    and $\laml{\loc}{\var}{\progtwo_0}$ each of which is an abstraction of
    $\val \unif \valtwo$, and such that moreover $\prog_0\SUB{\subst} = \prog$
    and $\progtwo_0\SUB{\subst} = \progtwo$.
    Then since $\laml{\loc}{\var}{\prog_0}$ and $\laml{\loc}{\var}{\progtwo_0}$ are abstractions
    on the left-hand side, we have by hypothesis that $\prog_0 = \progtwo_0$, hence
    $\prog = \prog_0\SUB{\subst} = \progtwo_0\SUB{\subst} = \progtwo$, as required.
  \end{enumerate}
\end{enumerate}
\end{proof}

\subsection{Proof of \rlem{reduction_modulo_structeq} --- Reduction modulo structural equivalence}
\lsec{appendix_reduction_modulo_structeq}

\begin{lemma}{Basic properties of structural equivalence}
\llem{basic_properties_structeq}
The following properties hold:
\begin{enumerate}
\item $\prog \alt \progtwo \structeq \progtwo \alt \prog$
\item If $\prog \structeq \prog'$ then $\progtwo_1 \alt \prog \alt \progtwo_2 \structeq \progtwo_1 \alt \prog' \alt \progtwo_2$
\end{enumerate}
\end{lemma}
\begin{proof}
Straightforward, by induction on the derivation of the corresponding equivalences.
\end{proof}

We turn to the proof of \rlem{reduction_modulo_structeq}:

\begin{proof}
By induction on the derivation of $\prog \equiv \prog'$.
The reflexivity and transitivity cases are immediate.
Moreover, it is easy to check that the axioms are symmetric.
So it suffices to show that the
property holds when $\prog \equiv \prog'$ is derived using one of the axioms:
\begin{enumerate}
\item \rulename{$\structeq$-swap}:
  Let $\tm \toca{} \progtwo$. The situation is:
  \[
  \xymatrix@C=.2cm@R=.5cm{
    \prog_1 \alt \tm \alt \tmtwo \alt \prog_2
    \ar[d]
    & \structeq & 
    \prog_1 \alt \tmtwo \alt \tm \alt \prog_2
    \ar[d]
  \\
    \prog_1 \alt \progtwo \alt \tmtwo \alt \prog_2
    & \structeq & 
    \prog_1 \alt \tmtwo \alt \progtwo \alt \prog_2
  }
  \]
  The equivalence at the bottom is justified by \rlem{basic_properties_structeq}.
\item \rulename{$\structeq$-var}:
  Let $\tm \to \progtwo$, $\varthree \not\in \fv{\tm}$.
  Then we argue that $\tm\sub{\vartwo}{\varthree} \to \progtwo\sub{\vartwo}{\varthree} \structeq \progtwo$.
  By case analysis on the reduction rule applied.
  \begin{enumerate}
  \item \rulename{alloc}:
    The situation is:
    \[
    \xymatrix@C=.2cm@R=.5cm{
      \prog_1 \alt \wctxof{\lam{\var}{\prog}} \alt \prog_2
      \ar[d]
      & \structeq & 
      \prog_1 \alt \wctx\sub{\vartwo}{\varthree}\ctxof{\lam{\var}{\prog\sub{\vartwo}{\varthree}}} \alt \prog_2
      \ar[d]
    \\
      \prog_1 \alt \wctxof{\laml{\loc}{\var}{\prog}} \alt \prog_2
      & \structeq & 
      \prog_1 \alt \wctx\sub{\vartwo}{\varthree}\ctxof{\laml{\loc'}{\var}{\prog\sub{\vartwo}{\varthree}}} \alt \prog_2
    }
    \]
    For the equivalence at the bottom is justified using
    \rulename{$\structeq$-var} to rename $\vartwo$ to $\varthree$, and
    \rulename{$\structeq$-loc} if necessary to rename $\loc$ to $\loc'$.
  \item \rulename{beta}:
    The situation is:
    \[
    {\small
    \xymatrix@C=.2cm@R=.5cm{
      \prog_1 \alt \wctxof{(\laml{\loc}{\var}{\prog})\,\val} \alt \prog_2
      \ar[d]
      & \structeq & 
      \prog_1 \alt \wctx\sub{\vartwo}{\varthree}\ctxof{(\laml{\loc}{\var}{\prog\sub{\vartwo}{\varthree}})\,\val\sub{\vartwo}{\varthree}} \alt \prog_2
      \ar[d]
    \\
      \prog_1 \alt \wctxof{\prog\sub{\var}{\val}} \alt \prog_2
      & \structeq & 
      \prog_1 \alt \wctx\sub{\vartwo}{\varthree}\ctxof{\prog\sub{\var}{\val}\sub{\vartwo}{\varthree}} \alt \prog_2
    }
    }
    \]
    For the equivalence at the bottom, note that by \rlem{properties_of_substitution},
    $\prog\sub{\var}{\val}\sub{\vartwo}{\varthree} = \prog\sub{\vartwo}{\varthree}\sub{\var}{\val\sub{\vartwo}{\varthree}}$.
  \item \rulename{guard}:
    This case is straightforward.
  \item \rulename{fresh}:
    The situation is:
    \[
    {\small
    \xymatrix@C=.2cm@R=.5cm{
      \prog_1 \alt \wctxof{\fresh{\var}{\tm}} \alt \prog_2
      \ar[d]
      & \structeq & 
      \ar[d]
      \prog_1 \alt \wctx\sub{\vartwo}{\varthree}\ctxof{\fresh{\var}{\tm\sub{\vartwo}{\varthree}}} \alt \prog_2
    \\
      \prog_1 \alt \wctxof{\tm} \alt \prog_2
     & \structeq & 
      \prog_1 \alt \wctx\sub{\vartwo}{\varthree}\ctxof{\tm\sub{\vartwo}{\varthree}} \alt \prog_2
    }
    }
    \]
    Note that assume $\var \neq \vartwo$ by Barendregt's variable convention.
  \item \rulename{unif}:
    Let $\mgu{\val \unif \valtwo} = \subst$.
    Then $\subst' := (\varthree \mapsto \vartwo) \scomp \subst$
    is an idempotent most general unifier of the single goal $\val\sub{\vartwo}{\varthree} \unif \valtwo\sub{\vartwo}{\varthree}$
    by \rlem{properties_of_mgus}.
    So $\subst'' = \mgu{\val\sub{\vartwo}{\varthree} \unif \valtwo\sub{\vartwo}{\varthree}}$ exists and
    $\subst'' = \subst' \scomp \substtwo = (\varthree \mapsto \vartwo) \scomp \subst \scomp \substtwo$ for some renaming $\substtwo$.
    \[
    {\small
    \xymatrix@C=.2cm@R=.5cm{
      \prog_1 \alt \wctxof{\val \unif \valtwo} \alt \prog_2
      \ar[dd]
      & \structeq & 
      \ar[d]
      \prog_1 \alt \wctx\sub{\vartwo}{\varthree}\ctxof{\val\sub{\vartwo}{\varthree} \unif \valtwo\sub{\vartwo}{\varthree}} \alt \prog_2
    \\
     & &
      \prog_1 \alt \wctx\sub{\vartwo}{\varthree}\ctxof{\unit}\SUB{\subst''} \alt \prog_2
      \ar@{=}[d]
    \\
      \prog_1 \alt \wctxof{\unit}\SUB{\subst} \alt \prog_2
     & \structeq & 
      \prog_1 \alt \wctx\ctxof{\unit}\SUB{\subst\scomp\substtwo} \alt \prog_2
    }
    }
    \]
    The equivalence at the bottom may be deduced by repeatedly applying the \rulename{$\structeq$-var}
    rule to perform the renaming $\substtwo$.
  \item \rulename{fail}:
    Suppose that $\mgu{\val \unif \valtwo}$ fails. Then
    $\mgu{\val\sub{\vartwo}{\varthree} \unif \valtwo\sub{\vartwo}{\varthree}}$ must also fail,
    for if $\subst$ were a unifier of $(\val\sub{\vartwo}{\varthree} \unif \valtwo\sub{\vartwo}{\varthree})$
    then $(\vartwo \mapsto \varthree)\scomp\subst$ would be a unifier of $\val \unif \valtwo$
    by \rlem{properties_of_mgus}. So we have:
    \[
    {\small
    \xymatrix@C=.2cm@R=.5cm{
      \prog_1 \alt \wctxof{\val \unif \valtwo} \alt \prog_2
      \ar[d]
      & \structeq & 
      \ar[d]
      \prog_1 \alt \wctx\sub{\vartwo}{\varthree}\ctxof{\val\sub{\vartwo}{\varthree} \unif \valtwo\sub{\vartwo}{\varthree}} \alt \prog_2
    \\
      \prog_1 \alt \prog_2
     & \structeq & 
      \prog_1 \alt \prog_2
    }
    }
    \]
  \end{enumerate}
\item \rulename{$\structeq$-loc}:
  If the \rulename{$\structeq$-loc} rule and the rewriting rule
  are applied on different threads, it is straightforward.
  Otherwise we proceed by case analysis on the reduction rule applied:
  \begin{enumerate}
  \item \rulename{alloc}:
    Let us write
    $\wctx' := \wctx\sub{\loc_1}{\loc_2}$ and
    $\progtwo' := \progtwo\sub{\loc_1}{\loc_2}$.
    Then:
    \[
    \xymatrix@C=.2cm@R=.5cm{
      \prog_1 \alt \wctxof{\lam{\var}{\progtwo}} \alt \prog_2
      \ar[d]
     & \structeq & 
      \prog_1 \alt \wctx'\ctxof{\lam{\var}{\progtwo'}} \alt \prog_2
      \ar[d]
    \\
      \prog_1 \alt \wctxof{\laml{\loc}{\var}{\progtwo}} \alt \prog_2
     & \structeq & 
      \prog_1 \alt \wctx'\ctxof{\laml{\loc'}{\var}{\progtwo'}} \alt \prog_2
    }
    \]
    The equivalence on the bottom may be deduced by applying the \rulename{$\structeq$-loc}
    rule to rename $\loc_1$ to $\loc_2$,
    and possibly the \rulename{$\structeq$-loc} again to rename $\loc$ to $\loc'$.
    Note that there is no possibility of conflict because $\loc$ and $\loc'$ are fresh.
  \item \rulename{beta}:
    Let us write
    $\wctx' := \wctx\sub{\loc_1}{\loc_2}$,
    $\progtwo' := \progtwo\sub{\loc_1}{\loc_2}$,
    and
    $\val' := \val\sub{\loc_1}{\loc_2}$.
    Then we have:
    \[
    \xymatrix@C=.2cm@R=.5cm{
      \prog_1 \alt \wctxof{(\laml{\loc}{\var}{\progtwo})\,\val} \alt \prog_2
      \ar[d]
     & \structeq & 
      \prog_1 \alt \wctx'\ctxof{(\laml{\loc\sub{\loc_1}{\loc_2}}{\var}{\progtwo'})\,\val'} \alt \prog_2
      \ar[d]
    \\
      \prog_1 \alt \wctxof{\progtwo\sub{\var}{\val}} \alt \prog_2
     & \structeq & 
      \prog_1 \alt \wctx'\ctxof{\progtwo'\sub{\var}{\val'}} \alt \prog_2
    }
    \]
    The equivalence on the bottom may be deduced by repeatedly applying the \rulename{$\structeq$-loc} rule.
  \item \rulename{guard}:
    This case is straightforward.
  \item \rulename{fresh}:
    This case is straightforward.
  \item \rulename{unif}:
    Consider a thread of the form $\wctxof{\val \unif \valtwo}$,
    and suppose that $\subst = \mgu{\val \unif \valtwo}$ exists.
    Let us write $\wctx' := \wctx\sub{\loc}{\loc'}$,
    $\val' := \val\sub{\loc}{\loc'}$, and
    $\valtwo' := \valtwo\sub{\loc}{\loc'}$.
    By \rlem{properties_of_mgus}, the substitution $\subst'$
    given by $\subst'(\var) = \subst(\var)\sub{\loc}{\loc'}$
    is an idempotent most general unifier of $\set{\val' \unif \valtwo'}$
    so $\subst'' = \mgu{\val' \unif \valtwo'}$ exists and moreover, by \rlem{properties_of_mgus},
    we have $\subst'' = \subst' \scomp \substtwo$ for some renaming $\substtwo$.
    Note that $\wctx'\ctxof{\unit}\SUB{\subst'} = \wctx\ctxof{\unit}\SUB{\subst}\sub{\loc}{\loc'}$; so:
    \[
      \xymatrix@C=.2cm@R=.5cm{
        \prog_1 \alt \wctxof{\val \unif \valtwo} \alt \prog_2
        \ar[dd]
       & \structeq & 
        \prog_1 \alt \wctx'\ctxof{\val' \unif \valtwo'} \alt \prog_2
        \ar[d]
      \\
       & & 
        \prog_1 \alt \wctx'\ctxof{\unit}\SUB{\subst' \scomp \substtwo} \alt \prog_2
        \ar@{=}[d]
      \\
        \prog_1 \alt \wctxof{\unit}\SUB{\subst} \alt \prog_2
       & \structeq & 
        \prog_1 \alt \wctx\ctxof{\unit}\SUB{\subst}\sub{\loc}{\loc'}\SUB{\substtwo} \alt \prog_2
      }
    \]
    The equivalence at the bottom may be deduced applying the \rulename{$\structeq$-loc} rule
    to rename $\loc$ to $\loc'$ and then repeatedly applying the \rulename{$\structeq$-var}
    rule to perform the renaming $\substtwo$.
  \item \rulename{fail}:
    Consider a thread of the form $\wctxof{\val \unif \valtwo}$, and
    let us write $\wctx' := \wctx\sub{\loc}{\loc'}$,
    $\val' := \val\sub{\loc}{\loc'}$, and
    $\valtwo' := \valtwo\sub{\loc}{\loc'}$.
    Suppose moreover that $\mgu{\val \unif \valtwo}$ fails.
    Then $\mgu{\val' \unif \valtwo'}$
    must also fail, for if $\subst$ were a unifier of $\val' \unif \valtwo'$,
    the substitution $\subst'$ given by $\subst'(\var) = \subst(\var)\sub{\loc'}{\loc}$
    would be a unifier of $\val \unif \valtwo$.
    Hence:
    \[
      \xymatrix@C=.2cm@R=.5cm{
        \prog_1 \alt \wctxof{\val \unif \valtwo} \alt \prog_2
        \ar[d]
       & \structeq & 
        \prog_1 \alt \wctx'\ctxof{\val' \unif \valtwo'} \alt \prog_2
        \ar[d]
      \\
        \prog_1 \alt \prog_2
       & \structeq & 
        \prog_1 \alt \prog_2
      }
    \]
  \end{enumerate}
\end{enumerate}
\end{proof}

\subsection{Proof of \rprop{characterization_normal_forms} --- Characterization of Normal Forms}
\lsec{appendix_characterization_normal_forms}

\begin{lemma}[Values are irreducible]
\llem{values_irreducible}
If $\val$ is a value, then it is a normal form.
\end{lemma}
\begin{proof}
Straightforward by induction on $\val$.
\end{proof}

\begin{lemma}[Application of a stuck term]
\llem{application_of_a_stuck_term}
If $\stm$ is stuck and $\ntm$ is a normal term,
then $\stm\,\ntm$ is stuck.
\end{lemma}
\begin{proof}
Straightforward by case analysis on the derivation of the judgment
``$\stm \stuck$''.
\end{proof}

\begin{lemma}[Values and stuck terms are disjoint]
\llem{values_disjoint_stuck_terms}
A stuck term $\stm$ is not a value.
\end{lemma}
\begin{proof}
By induction on the derivation of the judgment ``$\stm \stuck$''.
First, note that if $\stm = \var\,\ntm_1\hdots\ntm_n$ is stuck,
it cannot be value because $n > 0$.
Second, note that if $\stm = \cons\,\ntm_1\hdots\ntm_n$ is stuck,
it cannot be value because by \ih there is an $i$ such that $\ntm_i$ is not a value.
In the remaining cases,
we have either $\stm = (\ntm_1\seq\ntm_2)\,\ntmtwo_1\hdots\ntmtwo_n$,
$\stm = (\ntm_1\unif\ntm_2)\,\ntmtwo_1\hdots\ntmtwo_n$
or $\stm = (\laml{\loc}{\var}{\prog})\,\ntm\,\ntmtwo_1\hdots\ntmtwo_n$,
so the term $\stm$ is clearly not a value.
\end{proof}

We turn to the proof of \rprop{characterization_normal_forms}.
We prove the two inclusions.
For the $(\subseteq)$ inclusion,
by induction on a given normal program, it suffices to show that any normal term $\ntm$
is a $\toca{}$-normal form, which can be seen by induction on the derivation
that $\ntm$ is a normal term.
Recall that values are $\toca{}$-normal forms (\rlem{values_irreducible})
so we are left to check that any stuck term is a $\toca{}$-normal form.
If $\ntm$ is stuck, it is straightforward to check, in each case
of the definition of the judgment $\ntm \stuck$,
that the resulting term has no $\toca{}$-redexes.
Using the fact that a stuck term $\stm$ cannot be a value (\rlem{values_disjoint_stuck_terms}),
the key remarks are that:
\begin{enumerate}
\item \rulename{stuck-guard}:
  $\ntm_1\seq\ntm_2$ cannot be a $\toca{guard}$-redex because $\ntm_1$ is stuck (hence not a value);
\item \rulename{stuck-unif}:
  $\ntm_1\unif\ntm_2$ cannot be a $\toca{unif}$-redex nor a $\toca{fail}$-redex
  because for some $i \in \set{1,2}$ the term $\ntm_i$ is stuck (hence not a value);
\item \rulename{stuck-lam}:
  $(\laml{\loc}{\var}{\prog})\,\ntm$ cannot be a $\toca{beta}$-redex
  because the term $\ntm$ is stuck (hence not a value).
\end{enumerate}

For the $(\supseteq)$ inclusion,
by induction on a given program, it suffices to show that any term $\tm$
in $\toca{}$-normal form is actually a normal term.
By induction on $\tm$:
\begin{enumerate}
\item {\em Variable, $\tm = \var$.}
  Then $\tm$ is a value.
\item {\em Constructor, $\tm = \cons$.}
  Then $\tm$ is a value.
\item {\em Fresh variable declaration, $\tm = \fresh{\var}{\tmtwo}$.}
  Impossible, as it is not a $\toca{}$-normal form.
\item {\em Abstraction code, $\tm = \lam{\var}{\tmtwo}$.}
  Impossible, as it is not a $\toca{}$-normal form.
\item {\em Allocated abstraction, $\tm = \laml{\loc}{\var}{\tmtwo}$.}
  Then $\tm$ is a value.
\item {\em Application, $\tm = \tmtwo\,\tmthree$.}
  Note that $\tmtwo$ and $\tmthree$ are $\toca{}$-normal forms.
  By \ih, $\tmtwo$ and $\tmthree$ are normal terms,
  that is they are either a value or a stuck term.
  We consider the following four cases, depending on the shape of $\tmtwo$:
  \begin{enumerate}
  \item If $\tmtwo = \var$, then $\var\,\tmthree$ is stuck by \rulename{stuck-var}.
  \item If $\tmtwo = \cons\val_1\hdots\val_n$, then:
        \begin{enumerate}
        \item If $\tmthree$ is a value, $\cons\,\val_1\hdots\val_n\,\tmthree$ is a value.
        \item If $\tmthree$ is stuck, $\cons\,\val_1\hdots\val_n\,\tmthree$ is stuck by \rulename{stuck-cons}.
        \end{enumerate}
  \item If $\tmtwo = \laml{\loc}{\var}{\prog}$ then:
        \begin{enumerate}
        \item If $\tmthree$ is a value, this case is impossible because $(\laml{\loc}{\var}{\prog})\,\tmthree$ has a $\toca{beta}$-redex.
        \item If $\tmthree$ is stuck, then $(\laml{\loc}{\var}{\prog})\,\tmthree$ is stuck by \rulename{stuck-lam}.
        \end{enumerate}
  \item If $\tmtwo$ is stuck, then $\tmtwo\,\tmthree$ is stuck by \rlem{application_of_a_stuck_term}.
  \end{enumerate}
\item {\em Guarded expression, $\tm = (\tmtwo\seq\tmthree)$.}
  Note that $\tmtwo$ and $\tmthree$ are $\toca{}$-normal forms.
  By \ih, $\tmtwo$ and $\tmthree$ are normal terms,
  that is they are either a value or a stuck term.
  Note that $\tmtwo$ cannot be a value, because $\tmtwo\seq\tmthree$
  would have a $\toca{guard}$-redex, so $\tmtwo$ is stuck
  and $\tmtwo\seq\tmthree$ is stuck by \rulename{stuck-guard}.
\item {\em Unification, $\tm = (\tmtwo\unif\tmthree)$.}
  Note that $\tmtwo$ and $\tmthree$ are $\toca{}$-normal forms.
  By \ih, $\tmtwo$ and $\tmthree$ are normal terms,
  that is they are either a value or a stuck term.
  Note that $\tmtwo$ and $\tmthree$ cannot both be values, because
  $\tmtwo\unif\tmthree$ would have
  either a $\toca{unif}$-redex (if $\mgu{\tmtwo \unif \tmthree}$ exists)
  or a $\toca{unif}$-redex (if $\mgu{\tmtwo \unif \tmthree}$ fails).
  So either $\tmtwo$ is stuck or $\tmthree$ is stuck,
  so we have that $\tmtwo\unif\tmthree$ is stuck by \rulename{stuck-unif}.
\end{enumerate}

\subsection{Proof of \rlem{parallel_modulo_structeq}, item~\ref{parallel_modulo_structeq_item}. --- Simultaneous reduction modulo structural equivalence}
\lsec{appendix_parallel_modulo_structeq}

\begin{lemma}[Goals in a simultaneous reduction are in the term]
\llem{simultaneous_goals_are_in_the_term}
Let $\tm \topar{\goals} \prog$.
Then $\goals$ is a subset of the set:
\[
  \set{\val \unif \valtwo \ST \exists\wctx.\ \tm = \wctxof{\val \unif \valtwo}}
\]
In particular, $\fv{\goals} \subseteq \fv{\tm}$ and $\locs{\goals} \subseteq \locs{\tm}$.
\end{lemma}
\begin{proof}
Straightforward by induction on the derivation of $\tm \topar{\goals} \prog$.
\end{proof}

\begin{lemma}[Simultaneous evaluation of an alternative]
\llem{parallel_reduction_alternative}
The following are equivalent:
\begin{enumerate}
\item $\prog \alt \progtwo \topar{} \progthree$
\item $\progthree$ can be written as $\prog' \alt \progtwo'$,
      where $\prog \topar{} \prog'$ and $\progtwo \topar{} \progtwo'$.
\end{enumerate}
\end{lemma}
\begin{proof}
Straightforward, by induction on $\prog$.
\end{proof}

\begin{lemma}[Action of renaming on simultaneous evaluation]
\llem{parallel_reduction_renaming}
\quad
\begin{enumerate}
\item If $\tm \topar{\goals} \prog$
      then $\tm\sub{\var}{\vartwo} \topar{\goals\sub{\var}{\vartwo}} \prog\sub{\var}{\vartwo}$.
\item If $\tm \topar{\goals} \prog$
      then $\tm\sub{\loc}{\loctwo} \topar{\goals\sub{\loc}{\loctwo}} \prog\sub{\loc}{\loctwo}$.
\end{enumerate}
\end{lemma}
\begin{proof}
Straightforward by induction on the derivation of $\tm \topar{\goals} \prog$.
\end{proof}

We turn to the proof of \rlem{parallel_modulo_structeq}, item~\ref{parallel_modulo_structeq_item}:
\begin{proof}
By induction on the derivation of $\prog \structeq \prog'$.
It suffices to show that the property holds when $\prog \structeq \prog'$ is
derived using one of the axioms:
\begin{enumerate}
\item \rulename{$\structeq$-swap}:
  The situation is
  \[
    \xymatrix@C=.2cm@R=.5cm{
      \prog_1 \alt \tm_1 \alt \tm_2 \alt \prog_2
      \ar@{=>}[d]
     & \structeq & 
      \prog_1 \alt \tm_2 \alt \tm_1 \alt \prog_2
      \ar@{=>}[d]
    \\
      \prog'_1 \alt \progtwo_1 \alt \progtwo_2 \alt \prog'_2
     & \structeq & 
      \prog'_1 \alt \progtwo_2 \alt \progtwo_1 \alt \prog'_2
    }
  \]
  where by \rlem{parallel_reduction_alternative} we have that
  $\prog_1 \topar{} \prog'_1$,
  $\tm_1 \topar{} \progtwo_1$,
  $\tm_2 \topar{} \progtwo_2$,
  and $\prog_2 \topar{} \prog'_2$.
  The equivalence at the bottom is justified by \rlem{basic_properties_structeq}.
\item \rulename{$\structeq$-var}:
  Consider a program of the form $\prog_1 \alt \tm \alt \prog_2$, and let $\vartwo \not\in \fv{\tm}$. 
  Moreover, suppose that $\prog_1 \alt \tm \alt \prog_2 \topar{} \progthree$.
  By \rlem{parallel_reduction_alternative} we have that
  $\progthree = \prog'_1 \alt \progtwo \alt \prog'_2$
  where $\prog_1 \topar{} \prog'_1$, $\tm \topar{} \progtwo$, and $\prog_2 \topar{} \prog'_2$.
  The simultaneous reduction step $\tm \topar{} \progtwo$ is deduced from $\tm \topar{\goals} \progtwo'$
  for some set of goals $\goals$, in such a way that:
  \[
    \progtwo = \begin{cases}
               \progtwo'\SUB{\subst} & \text{if $\subst = \mgu{\goals}$} \\
               \fail                 & \text{if $\mgu{\goals}$ fails} \\
               \end{cases}
  \]
  By \rlem{parallel_reduction_renaming}, this means that
  $\tm\sub{\var}{\vartwo} \topar{\goals\sub{\var}{\vartwo}} \progtwo'\sub{\var}{\vartwo}$.
  Note that $\fv{\goals} \subseteq \fv{\tm}$ by \rlem{simultaneous_goals_are_in_the_term},
  so in particular $\vartwo \not\in \fv{\goals}$.
  This implies by \rlem{properties_of_mgus} that
  $\subst = \mgu{\goals}$ exists if and only if $\subst' = \mgu{\goals\sub{\var}{\vartwo}}$ exists.

  If $\subst = \mgu{\goals}$ exists, moreover by \rlem{properties_of_mgus}
  we have that $\subst' = (\vartwo \mapsto \var) \scomp \subst \scomp \substtwo$
  for some renaming $\substtwo$, and the situation is:
  \[
    \xymatrix@C=.2cm@R=.5cm{
      \prog_1 \alt \tm \alt \prog_2
      \ar@{=>}[dd]
     & \structeq & 
      \prog_1 \alt \tm\sub{\var}{\vartwo} \alt \prog_2
      \ar@{=>}[d]
    \\
     && 
      \prog_1 \alt \progtwo'\sub{\var}{\vartwo}\SUB{(\vartwo \mapsto \var) \scomp \subst \scomp \substtwo} \alt \prog_2
      \ar@{=}[d]
    \\
      \prog_1 \alt \progtwo'\SUB{\subst} \alt \prog_2
     & \structeq & 
      \prog_1 \alt (\progtwo'\SUB{\subst})\SUB{\substtwo} \alt \prog_2
    }
  \]
  The equivalence at the bottom is justified using \rulename{$\structeq$-var}
  to apply the renaming $\substtwo$.
  If $\mgu{\goals}$ fails, the situation is:
  \[
    \xymatrix@C=.2cm@R=.5cm{
      \prog_1 \alt \tm \alt \prog_2
      \ar@{=>}[d]
     & \structeq & 
      \prog_1 \alt \tm\sub{\var}{\vartwo} \alt \prog_2
      \ar@{=>}[d]
    \\
      \prog_1 \alt \prog_2
     & \structeq & 
      \prog_1 \alt \prog_2
    }
  \]
\item \rulename{$\structeq$-loc}:
  Similar as the previous case. Let $\loctwo \not\in \locs{\tm}$.
  By \rlem{parallel_reduction_renaming} we may conclude that if
  $\tm \topar{\goals} \progtwo'$
  then 
  $\tm\sub{\loc}{\loctwo} \topar{\goals\sub{\loc}{\loctwo}} \progtwo'\sub{\loc}{\loctwo}$.
  Note that $\locs{\goals} \subseteq \locs{\tm}$ by \rlem{simultaneous_goals_are_in_the_term},
  so in particular $\loctwo \not\in \fv{\goals}$.
  This implies by \rlem{properties_of_mgus} that
  $\subst = \mgu{\goals}$ exists if and only if $\subst' = \mgu{\goals\sub{\loc}{\loctwo}}$ exists.

  If $\subst = \mgu{\goals}$ exists, moreover by \rlem{properties_of_mgus}
  we have that $\subst' = \subst'' \scomp \substtwo$ where $\substtwo$ is a renaming, 
  and $\subst''$ is a substitution such that $\subst''(\var) = \subst(\var)\sub{\loc}{\loctwo}$.
  Hence the situation is:
  \[
    \xymatrix@C=.2cm@R=.5cm{
      \prog_1 \alt \tm \alt \prog_2
      \ar@{=>}[dd]
     & \structeq & 
      \prog_1 \alt \tm\sub{\loc}{\loctwo} \alt \prog_2
      \ar@{=>}[d]
    \\
     &&
      \prog_1 \alt \progtwo'\sub{\loc}{\loctwo}\SUB{\subst'' \cdot \substtwo} \alt \prog_2
      \ar@{=}[d]
    \\
      \prog_1 \alt \progtwo'\SUB{\subst} \alt \prog_2
     & \structeq & 
      \prog_1 \alt \progtwo'\SUB{\subst}\sub{\loc}{\loctwo}\SUB{\substtwo} \alt \prog_2
    }
  \]
  The equivalence at the bottom is justified using \rulename{$\structeq$-loc}
  to rename $\loc$ to $\loctwo$,
  and \rulename{$\structeq$-var} to apply the renaming $\substtwo$.
  If $\mgu{\goals}$ fails, the situation is:
  \[
    \xymatrix@C=.2cm@R=.5cm{
      \prog_1 \alt \tm \alt \prog_2
      \ar@{=>}[d]
     & \structeq & 
      \prog_1 \alt \tm\sub{\var}{\vartwo} \alt \prog_2
      \ar@{=>}[d]
    \\
      \prog_1 \alt \prog_2
     & \structeq & 
      \prog_1 \alt \prog_2
    }
  \]
\end{enumerate}
\end{proof}

\subsection{Proof of \rprop{tait_martin_lof_technique} --- Tait--Martin-L\"of's Technique}
\lsec{appendix_tait_martin_lof_technique}

For the proofs, we work with the following \indrulename{Thread} rule and
the following variant of the \indrulename{Alt} rule, which is obviously equivalent
to the one in the main body of the paper:
\[
  \indrule{Thread}{
    \tm \topar{\goals} \prog
    \HS
    \prog' =
      \begin{cases}
      \prog\SUB{\subst} & \text{if $\subst = \mgu{\goals}$} \\
      \fail             & \text{if $\mgu{\goals}$ fails} \\
      \end{cases}
  }{
    \tm \topar{} \prog'
  }
  \HS
  \indrule{Alt'}{
    \tm \topar{} \prog
    \HS
    \progtwo \topar{} \progtwo'
  }{
    \tm\alt\progtwo \topar{} \prog \alt \progtwo'
  }
\]

\begin{lemma}[``$\toca{}\ \subseteq\ \topar{}\structeq$'']
\llem{single_step_included_in_parallel}
If $\tm \toca{} \prog$ then $\tm \topar{}\structeq \prog$.
\end{lemma}
\begin{proof}
By case analysis on the rule used to conclude $\tm \toca{} \prog$.
\begin{enumerate}
\item \rulename{alloc}:
  $\wctxof{\lam{\var}{\prog}} \toca{} \wctxof{\laml{\loc}{\var}{\prog}}$ for some location
  $\loc \not\in \locs{\wctxof{\lam{\var}{\prog}}}$.
  Note that $\lam{\var}{\prog} \topar{\emptyset} \laml{\loctwo}{\var}{\prog}$ for an (a priori different)
  fresh location $\loctwo$ by rule \rulename{Abs$^\codesym_2$}.
  By context closure (\rlem{parallel_context_closure}), applying the \rulename{Thread} rule once,
  we have that
  $\wctxof{\lam{\var}{\prog}} \topar{} \wctxof{\laml{\loctwo}{\var}{\prog}} \structeq \wctxof{\laml{\loc}{\var}{\prog}}$
  as required.
  The last equivalence is justified renaming $\loctwo$ to $\loc$.
\item \rulename{beta}:
  $\wctxof{(\laml{\loc}{\var}{\prog})\,\val} \toca{} \wctxof{\prog\sub{\var}{\val}}$.
  Note that $(\laml{\loc}{\var}{\prog})\,\val \topar{\emptyset} \prog\sub{\var}{\val}$ by rule \rulename{App$_2$},
  so by context closure (\rlem{parallel_context_closure}), applying the \rulename{Thread} rule once, we conclude.
\item \rulename{fresh}:
  $\wctxof{\fresh{\var}{\tm}} \toca{} \wctxof{\tm\sub{\var}{\vartwo_1}}$ for some variable $\vartwo_1 \not\in \fv{\wctx}$.
  Note that $\fresh{\var}{\tm} \topar{\emptyset} \tm\sub{\var}{\vartwo_2}$ for an (a priori different) fresh variable $\vartwo_2$
  by rule \rulename{Fresh$_2$}. By context closure (\rlem{parallel_context_closure}), applying the
  \rulename{Thread} rule once,
  we have that $\wctxof{\fresh{\var}{\tm}} \topar{} \wctxof{\tm\sub{\var}{\vartwo_2}} \structeq \wctxof{\tm\sub{\var}{\vartwo_1}}$
  The last equivalence is justified renaming $\vartwo_2$ to $\vartwo_1$.
\item \rulename{guard}:
  $\wctxof{\val\seq\tm} \toca{} \wctxof{\tm}$. Note that
  $\val\seq\tm \topar{\emptyset} \tm$ by rule \rulename{Guard$_2$}.
  By context closure (\rlem{parallel_context_closure}), applying the \rulename{Thread} rule once,
  we have that $\wctxof{\val\seq\tm} \topar{} \wctxof{\tm}$ as required.
\item \rulename{unif}:
  Suppose that $\subst = \mgu{\val \unif \valtwo}$,
  and let $\wctxof{\val \unif \valtwo} \toca{} \wctxof{\unit}\SUB{\subst}$.
  Note that $\val \unif \valtwo \topar{\set{\val \unif \valtwo}} \unit$
  by rule \rulename{Unif$_2$}, so by context closure and applying the \rulename{Thread} 
  rule once we have that $\wctxof{\val \unif \valtwo} \topar{} \wctxof{\unit}\SUB{\subst}$,
  as required.
\item \rulename{fail}:
  Suppose that $\mgu{\val \unif \valtwo}$ fails,
  and let $\wctxof{\val \unif \valtwo} \toca{} \wctxof{\unit}\SUB{\subst}$.
  Note that $\val \unif \valtwo \topar{\set{\val \unif \valtwo}} \unit$
  by rule \rulename{Unif$_2$}, so by context closure and applying the \rulename{Thread} 
  rule once we have that $\wctxof{\val \unif \valtwo} \topar{} \fail$,
  as required.
\end{enumerate}
\end{proof}

\begin{lemma}[``$\topar{}\ \subseteq\ \rtoca{}\structeq$'']
\llem{parallel_included_in_many_step}
Let $\tm \topar{\goals} \prog$. Given any weak context $\wctx$ and any substitution $\substa$ we have:
\begin{enumerate}
\item If $\subst = \mgu{\goals\SUB{\substa}}$, then $\wctxof{\tm}\SUB{\substa} \mathrel{\rtoca{}\structeq} \wctxof{\prog}\SUB{\substa\cdot\subst}$.
\item If $\mgu{\goals\SUB{\substa}}$ fails, then $\wctxof{\tm}\SUB{\substa} \mathrel{\rtoca{}\structeq} \fail$.
\end{enumerate}
\end{lemma}
\begin{proof}
By induction on the derivation of $\tm \topar{\goals} \prog$:
\begin{enumerate}
\item \rulename{Var}:
  Note that $\mgu{\emptyset}$ is the identity substitution, so $\wctxof{\var}\SUB{\substa} \rtoca{} \wctxof{\var}\SUB{\substa}$
  in zero steps.
\item \rulename{Cons}:
  Immediate, similar to the \rulename{Var} case.
\item \rulename{Fresh$_1$}:
  Immediate, similar to the \rulename{Var} case.
\item \rulename{Fresh$_2$}:
  Let $\fresh{\var}{\tm} \topar{\goals} \prog$ be derived from
  $\tm \topar{\goals} \prog$, where $\var$ is a fresh variable.
  Moreover, let $\var' \not\in \fv{\wctx\SUB{\substa}}$.
  Then we have that:
  \[
    \wctxof{\fresh{\var}{\tm}}\SUB{\substa}
    =
    \wctx\SUB{\substa}\ctxof{\fresh{\var'}{\tm\sub{\var}{\var'}\SUB{\substa}}}
    \toca{fresh}
    \wctx\SUB{\substa}\ctxof{\tm\sub{\var}{\var'}\SUB{\substa}}
    =
    \wctxof{\tm\sub{\var}{\var'}}\SUB{\substa}
    \structeq
    \wctxof{\tm}\SUB{\substa} 
  \]
  There are two subcases, depending on whether $\mgu{\goals\SUB{\substa}}$ exists:
  \begin{enumerate}
  \item If $\subst = \mgu{\goals\SUB{\substa}}$,
        then by \ih, $\wctxof{\tm}\SUB{\substa} \rtoca{}\structeq \wctxof{\prog}\SUB{\substa\scomp\subst}$,
        so since $\structeq$ is a strong bisimulation (\rlem{reduction_modulo_structeq}),
        $\wctxof{\fresh{\var}{\tm}}\SUB{\substa} \rtoca{}\structeq \wctxof{\prog}\SUB{\substa\scomp\subst}$
        as required.
  \item If $\mgu{\goals\SUB{\substa}}$ fails,
        then by \ih, $\wctxof{\tm}\SUB{\substa} \rtoca{}\structeq \fail$,
        so since $\structeq$ is a strong bisimulation (\rlem{reduction_modulo_structeq}),
        $\wctxof{\fresh{\var}{\tm}}\SUB{\substa} \rtoca{}\structeq \fail$
        as required.
  \end{enumerate}
\item \rulename{Abs$^\codesym_1$}:
  Immediate, similar to the \rulename{Var} case.
\item \rulename{Abs$^\codesym_2$}:
  Let $\lam{\var}{\prog} \topar{\emptyset} \laml{\loc}{\var}{\prog}$,
  where $\loc$ is a fresh location.
  Moreover, let $\loc' \not\in \locs{\wctxof{\lam{\var}{\prog}}\SUB{\substa}}$.
  Then:
  \[
    \wctxof{\lam{\var}{\prog}}\SUB{\substa}
    =
    \wctx\SUB{\substa}\ctxof{\lam{\var}{\prog\SUB{\substa}}}
    \toca{alloc}
    \wctx\SUB{\substa}\ctxof{\laml{\loc'}{\var}{\prog\SUB{\substa}}}
    \structeq
    \wctx\SUB{\substa}\ctxof{\laml{\loc}{\var}{\prog\SUB{\substa}}}
    =
    \wctxof{\laml{\loc}{\var}{\prog}}\SUB{\substa}
  \]
  so $\wctxof{\lam{\var}{\prog}}\SUB{\substa} \rtoca{}\structeq \wctxof{\laml{\loc}{\var}{\prog}}\SUB{\substa}$.
  Note that $\mgu{\emptyset}$ is the identity substitution, so we are done.
\item \rulename{Abs$^\allocsym$}:
  Immediate, similar to the \rulename{Var} case.
\item \rulename{App$_1$}:
  Let $\tm\,\tmtwo \topar{\goals \cup \goalstwo} \prog\,\progtwo$
  be derived from
  $\tm \topar{\goals} \prog$ and $\tmtwo \topar{\goalstwo} \progtwo$.
  We consider two subcases,
  depending on whether $\mgu{\goals\SUB{\substa}}$ exists:
  \begin{enumerate}
  \item
    If $\subst = \mgu{\goals\SUB{\substa}}$ exists.
    Let us write $\prog = \bigalt_{i=1}^{n} \tm_i$.
    Then applying the \ih for the term $\tm$ under the weak context $\wctxof{\ctxhole\,\tmtwo}$,
    we have that
    $\wctxof{\tm\,\tmtwo}\SUB{\substa}
       \rtoca{}\structeq \wctxof{\prog\,\tmtwo}\SUB{\substa\scomp\subst}
       = \bigalt_{i=1}^{n} \wctxof{\tm_i\,\tmtwo}\SUB{\substa\scomp\subst}$.
    We consider two further subcases,
    depending on whether $\mgu{\goalstwo\SUB{\substa\scomp\subst}}$ exists:
    \begin{enumerate}
    \item
      If $\substtwo = \mgu{\goalstwo\SUB{\substa\scomp\subst}}$ exists,
      then applying the \ih for each $1 \leq i \leq n$, for the term $\tmtwo$
      under the weak context $\wctxof{\tm_i\,\ctxhole}$,
      we have that
      $\wctxof{\tm_i\,\tmtwo}\SUB{\substa\scomp\subst} \rtoca{}\structeq \wctxof{\tm_i\,\progtwo}\SUB{\substa\scomp\subst\scomp\substtwo}$.
      Moreover, by the compositionality property (\rlem{mgu_compositional})
      we have that $\substthree = \mgu{\goals\SUB{\substa} \cup \goalstwo\SUB{\substa}}$
      exists, and it is a renaming of $\subst\scomp\substtwo$.
      In summary, we have:
      \[
        \begin{array}{rrll}
        \wctxof{\tm\,\tmtwo}\SUB{\substa}
           & \rtoca{}\structeq & \wctxof{\prog\,\tmtwo}\SUB{\substa\scomp\subst}
                                 & \text{by \ih on $\tm$} \\
           & =                 & \bigalt_{i=1}^{n} \wctxof{\tm_i\,\tmtwo}\SUB{\substa\scomp\subst} \\
           & \rtoca{}\structeq & \bigalt_{i=1}^{n} \wctxof{\tm_i\,\progtwo}\SUB{\substa\scomp\subst\scomp\substtwo}
                                 & \text{by \ih on $\tmtwo$} \\
           & =                 & \wctxof{\prog\,\progtwo}\SUB{\substa\scomp\subst\scomp\substtwo} \\
           & \structeq         & \wctxof{\prog\,\progtwo}\SUB{\substa\scomp\substthree}
        \end{array}
      \]
      so since $\structeq$ is a strong bisimulation (\rlem{reduction_modulo_structeq}),
      $\wctxof{\tm\,\tmtwo}\SUB{\substa} \rtoca{}\structeq \wctxof{\prog\,\progtwo}\SUB{\substa\scomp\substthree}$,
      as required.
    \item
      If $\mgu{\goalstwo\SUB{\substa\scomp\subst}}$ fails,
      then applying the \ih for each $1 \leq i \leq n$, for the term $\tmtwo$
      under the weak context $\wctxof{\tm_i\,\tmtwo}$,
      we have that
      $\wctxof{\tm_i\,\tmtwo}\SUB{\substa\scomp\subst} \rtoca{}\structeq \fail$.
      Moreover, by the compositionality property (\rlem{mgu_compositional})
      we have that $\mgu{\goals\SUB{\substa} \cup \goalstwo\SUB{\substa}}$
      also fails, so we have:
      \[
        \begin{array}{rrll}
        \wctxof{\tm\,\tmtwo}\SUB{\substa}
          & \rtoca{}\structeq & \wctxof{\prog\,\tmtwo}\SUB{\substa\scomp\subst}
                                & \text{by \ih on $\tm$} \\
          & =                 & \bigalt_{i=1}^{n} \wctxof{\tm_i\,\tmtwo}\SUB{\substa\scomp\subst} \\
          & \rtoca{}\structeq & \bigalt_{i=1}^{n} \fail & \text{by \ih on $\tmtwo$} \\
          & =                 & \fail
        \end{array}
      \]
      so since $\structeq$ is a strong bisimulation (\rlem{reduction_modulo_structeq}),
      $\wctxof{\tm\,\tmtwo}\SUB{\substa} \rtoca{}\structeq \fail$,
      as required.
    \end{enumerate}
  \item
    If $\mgu{\goals\SUB{\substa}}$ fails,
    then applying the \ih for the term $\tm$ under the weak context $\wctxof{\ctxhole\,\tmtwo}$
    we have that $\wctxof{\tm\,\tmtwo}\SUB{\substa} \rtoca{}\structeq \fail$.
    Moreover, by the compositionality property (\rlem{mgu_compositional})
    we have that $\mgu{\goals\SUB{\substa} \cup \goalstwo\SUB{\substa}}$
    also fails, so we are done.
  \end{enumerate}
\item \rulename{App$_2$}:
  Let $(\lam{\var}{\prog})\,\val \topar{\emptyset} \prog\sub{\var}{\val}$.
  Then since $\mgu{\emptyset}$ is the identity substitution we have:
  \[
    \wctxof{(\lam{\var}{\prog})\,\val}\SUB{\substa}
    =
    \wctx\SUB{\substa}\ctxof{(\lam{\var}{\prog\SUB{\substa}})\,\val\SUB{\substa}}
    \toca{beta}
    \wctx\SUB{\substa}\ctxof{\prog\SUB{\substa}\sub{\var}{\val\SUB{\substa}}}
    =
    \wctx\ctxof{\prog\sub{\var}{\val}}\SUB{\substa}
  \]
  This concludes this case. The fact that $\val\SUB{\substa}$ is indeed a value
  (required to be able to apply the \rulename{beta} rule),
  and the last equality are justified by \rlem{properties_of_substitution}.
\item \rulename{Guard$_1$}:
  Similar to the \rulename{App$_1$} case.
\item \rulename{Guard$_2$}:
  Let $\val\seq\tm \topar{\goals} \prog$ be derived from $\tm \topar{\goals} \prog$.
  Let us write $\prog = \bigalt_{i=1}^{n} \tm_i$.
  We consider two cases, depending on whether $\mgu{\goals\SUB{\substa}}$ exists:
  \begin{enumerate}
  \item
    If $\subst = \mgu{\goals\SUB{\substa}}$ exists,
    then applying the \ih on the term $\tm$ under the context $\wctxof{\val\seq\ctxhole}$ we have
    that $\wctxof{\val\seq\tm}\SUB{\substa} \rtoca{}\structeq \wctxof{\val\seq\prog}\SUB{\substa\scomp\subst}$.
    Moreover, by \rlem{properties_of_substitution}, $\val\SUB{\substa\scomp\subst}$ is a value
    so we may apply the \rulename{guard} rule:
    \[
      \begin{array}{rrll}
      \wctxof{\val\seq\tm}\SUB{\substa}
        & \rtoca{}\structeq & \wctxof{\val\seq\prog}\SUB{\substa\scomp\subst}
                              & \text{by \ih on $\tm$} \\
        & =                 & \wctx\SUB{\substa\scomp\subst}\ctxof{\val\SUB{\substa\scomp\subst}\seq\prog\SUB{\substa\scomp\subst}} \\
        & =                 & \bigalt_{i=1}^{n} \wctx\SUB{\substa\scomp\subst}\ctxof{\val\SUB{\substa\scomp\subst}\seq\tm_i\SUB{\substa\scomp\subst}} \\
        & \rtoca{guard}     & \bigalt_{i=1}^{n} \wctx\SUB{\substa\scomp\subst}\ctxof{\tm_i\SUB{\substa\scomp\subst}} \\
        & =                 & \bigalt_{i=1}^{n} \wctxof{\tm_i}\SUB{\substa\scomp\subst} \\
        & =                 & \wctxof{\prog}\SUB{\substa\scomp\subst} \\
      \end{array}
    \]
    so since $\structeq$ is a strong bisimulation (\rlem{reduction_modulo_structeq}),
    we have that $\wctxof{\val\seq\tm}\SUB{\substa} \rtoca{}\structeq \wctxof{\prog}\SUB{\substa\scomp\subst}$
    as required.
  \item
    If $\mgu{\goals\SUB{\substa}}$ fails,
    then applying the \ih on the term $\tm$ under the context $\wctxof{\val\seq\ctxhole}$ we have
    that $\wctxof{\val\seq\tm}\SUB{\substa} \rtoca{}\structeq \fail$,
    as required.
  \end{enumerate}
\item \rulename{Unif$_1$}:
  Similar to the \rulename{App$_1$} case.
\item \rulename{Unif$_2$}:
  Let $\val \unif \valtwo \topar{\set{\val \unif \valtwo}} \unit$.
  We consider two cases, depending on whether $\mgu{\set{\val\SUB{\substa} \unif \valtwo\SUB{\substa}}}$ exists:
  \begin{enumerate}
  \item
    If $\subst = \mgu{\set{\val\SUB{\substa} \unif \valtwo\SUB{\substa}}}$ exists,
    note that by \rlem{properties_of_substitution},
    $\val\SUB{\substa\scomp\subst}$ and $\valtwo\SUB{\substa\scomp\subst}$ are values
    and we may apply the \rulename{unif} rule:
    \[
      \begin{array}{rrll}
      \wctxof{\val \unif \valtwo}\SUB{\substa}
        & =           & \wctx\SUB{\substa}\ctxof{\val\SUB{\substa} \unif \valtwo\SUB{\substa}}
      \\
        & \toca{unif} & \wctx\SUB{\substa}\ctxof{\unit}\SUB{\subst}
      \\
        & =           & \wctxof{\unit}\SUB{\substa\scomp\subst}
      \end{array}
    \]
    so $\wctxof{\val \unif \valtwo}\SUB{\substa} \rtoca{}\structeq \wctxof{\unit}\SUB{\substa\scomp\subst}$ as required.
  \item
    If $\mgu{\set{\val\SUB{\substa} \unif \valtwo\SUB{\substa}}}$ fails,
    note that by \rlem{properties_of_substitution},
    $\val\SUB{\substa\scomp\subst}$ and $\valtwo\SUB{\substa\scomp\subst}$ are values
    and we may apply the \rulename{fail} rule:
    \[
      \begin{array}{rrll}
      \wctxof{\val \unif \valtwo}\SUB{\substa}
        & =           & \wctx\SUB{\substa}\ctxof{\val\SUB{\substa} \unif \valtwo\SUB{\substa}}
      \\
        & \toca{fail} & \fail
      \end{array}
    \]
    so $\wctxof{\val \unif \valtwo}\SUB{\substa} \rtoca{}\structeq \fail$ as required.
  \end{enumerate}
\end{enumerate}
\end{proof}

\begin{lemma}[Values are irreducible]
\llem{simultaneous_reduction_of_a_value}
Let $\val \topar{\goals} \prog$ with $\val$ a value.
Then $\goals = \emptyset$ and $\prog = \val$. 
\end{lemma}
\begin{proof}
Straightforward by induction on $\val$.
Note that the only rules that may be applied are
\rulename{Var}, \rulename{Cons},
\rulename{Abs$^\allocsym$}, and \rulename{App$_1$}.
\end{proof}

\begin{lemma}[Diamond property]
\llem{diamond_property}
Let $\tm \topar{\goals_1} \bigalt_{i=1}^{n} \tm_i$ and $\tm \topar{\goals_2} \bigalt_{j=1}^{m} \tm^\star_j$.
Then there exist two sets of goals $\goals'_1$ and $\goals'_2$,
and programs $\prog_1,\hdots,\prog_n$ and $\prog^\star_1,\hdots,\prog^\star_m$ such that:
\begin{enumerate}
\item $\tm_i \topar{\goals'_2}\structeq \prog_i$ for all $1 \leq i \leq n$;
\item $\tm^\star_j \topar{\goals'_1}\structeq \prog^\star_j$ for all $1 \leq j \leq m$;
\item $\alt_{i=1}^{n} \prog_i \permeq \alt_{j=1}^{m} \prog^\star_j$
      where ``$\permeq$'' denotes the least equivalence generated by the \rulename{$\structeq$-swap} axiom,
      \ie structural equivalence allowing only permutation of threads;
\item $\goals_1 \cup \goals'_2 = \goals_2 \cup \goals'_1$.
\end{enumerate}
\end{lemma}
\begin{proof}
By induction on $\tm$:
\begin{enumerate}
\item {\em Variable, $\tm = \var$.}
  The only rule that applies is \rulename{Var}, \ie $\var \topar{\emptyset} \var$,
  so this case is trivial.
  More precisely, we have that $n = m = 1$ and $\tm_1 = \tm^\star_1 = \var$,
  with $\goals_1 = \goals_2 = \emptyset$,
  so taking $\goals'_1 = \goals'_2 = \emptyset$ and $\prog_1 = \prog^\star_1 = \var$
  it is straightforward to check that all the properties hold.
\item {\em Constructor, $\tm = \cons$.}
  Immediate, similar to the variable case.
\item {\em Fresh variable declaration, $\tm = \fresh{\var}{\tmtwo}$.}
  There are four cases, depending on whether each of the simultaneous steps
  is deduced by \rulename{Fresh$_1$} or \rulename{Fresh$_2$}:
  \begin{enumerate}
  \item \rulename{Fresh$_1$}/\rulename{Fresh$_1$}:
    Immediate, similar to the variable case.
  \item \rulename{Fresh$_1$}/\rulename{Fresh$_2$}:
    Let $\fresh{\var}{\tmtwo} \topar{\emptyset} \fresh{\var}{\tmtwo}$ be derived by rule \rulename{Fresh$_1$}
    (so that $n = 1$, $\tm_1 = \fresh{\var}{\tmtwo}$, and $\goals_1 = \emptyset$),
    and let $\fresh{\var}{\tmtwo} \topar{\goals_2} \bigalt_{j=1}^{m} \tm^\star_j$
    be derived by rule \rulename{Fresh$_2$} from
    $\tmtwo \topar{\goals_2} \tm^\star_j$.
    Then taking $\goals'_1 := \emptyset$, $\goals'_2 := \goals_2$,
    $\prog_1 := \bigalt_{j=1}^{m} \tm^\star_j$
    and $\prog^\star_j := \tm^\star_j$ for each $1 \leq j \leq m$,
    using reflexivity for terms (\rlem{properties_of_simultaneous_reduction})
    we have:
    \[
      \indrule{Fresh$_2$}{}{\tm_1 = \fresh{\var}{\tmtwo} \topar{\goals_2} \bigalt_{j=1}^{m} \tm^\star_j}
      \HS
      \indrule{}{\text{(\rlem{properties_of_simultaneous_reduction})}}{\tm^\star_j \topar{\emptyset} \tm^\star_j}
    \]
  \item \rulename{Fresh$_2$}/\rulename{Fresh$_1$}:
    Symmetric to the previous case (\rulename{Fresh$_1$}/\rulename{Fresh$_2$}).
  \item \rulename{Fresh$_2$}/\rulename{Fresh$_2$}:
    Let $\fresh{\var}{\tmtwo} \topar{\goals_1} \bigalt_{i=1}^{n} \tm_i$
    be derived by rule \rulename{Fresh$_2$} from $\tmtwo \topar{\goals_1} \bigalt_{i=1}^{n} \tm_i$,
    and let $\fresh{\var}{\tmtwo} \topar{\goals_1} \bigalt_{j=1}^{m} \tm_j$
    be derived by rule \rulename{Fresh$_2$} from $\tmtwo \topar{\goals_2} \bigalt_{j=1}^{n} \tm^\star_j$.
    Then by \ih on $\tmtwo$ there exist sets of goals $\goals'_1, \goals'_2$ and programs
    $\prog_1,\hdots,\prog_n,\prog^\star_1,\hdots,\prog^\star_m$ such that:
    \[
      \tm_i \topar{\goals'_2}\structeq \prog_i
      \HS
      \tm^\star_i \topar{\goals'_1}\structeq \prog_j
      \HS
      \bigalt_{i=1}^{n} \prog_i \permeq \bigalt_{j=1}^{m} \prog_j
      \HS
      \goals_1 \cup \goals'_2 = \goals_2 \cup \goals'_1
    \]
    which concludes this subcase.
  \end{enumerate}
\item {\em Abstraction code, $\tm = \lam{\var}{\prog}$.}
  There are four cases, depending on whether each of the simultaneous steps
  is deduced by \rulename{Abs$^\codesym_1$} or \rulename{Abs$^\codesym_2$}:
  \begin{enumerate}
  \item \rulename{Abs$^\codesym_1$}/\rulename{Abs$^\codesym_1$}:
    Immediate, similar to the variable case. 
  \item \rulename{Abs$^\codesym_1$}/\rulename{Abs$^\codesym_2$}:
    Let $\lam{\var}{\prog} \topar{\emptyset} \lam{\var}{\prog}$
    be derived from rule \rulename{Abs$^\codesym_1$}, and
    let $\lam{\var}{\prog} \topar{\emptyset} \laml{\loc}{\var}{\prog}$
    be derived from rule \rulename{Abs$^\codesym_2$}, where $\loc$ is a fresh location.
    Note that $n = m = 1$ and $\goals_1 = \goals_2 = \emptyset$.
    Taking $\goals'_1 = \goals'_2 = \emptyset$,
    for some fresh location $\loc'$, we have that:
    \[
      \indrule{Abs$^\codesym_2$}{}{\lam{\var}{\prog} \topar{\emptyset} \laml{\loc'}{\var}{\prog} \structeq \laml{\loc}{\var}{\prog}}
      \HS
      \indrule{Abs$^\allocsym$}{}{\laml{\loc}{\var}{\prog} \topar{\emptyset} \laml{\loc}{\var}{\prog}}
    \]
    which concludes this subcase. 
  \item \rulename{Abs$^\codesym_2$}/\rulename{Abs$^\codesym_1$}:
    Symmetric to the previous case (\rulename{Abs$^\codesym_1$}/\rulename{Abs$^\codesym_2$}).
  \item \rulename{Abs$^\codesym_2$}/\rulename{Abs$^\codesym_2$}:
    Let $\lam{\var}{\prog} \topar{\emptyset} \laml{\loc_1}{\var}{\prog}$
    and $\lam{\var}{\prog} \topar{\emptyset} \laml{\loc_2}{\var}{\prog}$
    be derived from rule \rulename{Abs$^\codesym_2$}, where $\loc_1$ and $\loc_2$ are fresh locations.
    Note that $n = m = 1$ and $\goals_1 = \goals_2 = \emptyset$.
    Taking $\goals'_1 = \goals'_2 = \emptyset$ we have that:
    \[
      \indrule{Abs$^\allocsym$}{}{\laml{\loc_1}{\var}{\prog} \topar{\emptyset} \laml{\loc_1}{\var}{\prog} \structeq \laml{\loc_2}{\var}{\prog}}
      \HS
      \indrule{Abs$^\allocsym$}{}{\laml{\loc_2}{\var}{\prog} \topar{\emptyset} \laml{\loc_2}{\var}{\prog}}
    \]
  \end{enumerate}
\item {\em Allocated abstraction, $\tm = \lam{\var}{\tmtwo}$.}
  Immediate, similar to the variable case.
\item {\em Application, $\tm = \tmtwo\,\tmthree$.}
  There are four cases, depending on whether each of the simultaneous steps
  is deduced by \rulename{App$^\codesym_1$} or \rulename{App$^\codesym_2$}:
  \begin{enumerate}
  \item \rulename{App$_1$}/\rulename{App$_1$}:
    This subcase is heavy to write---we give a detailed proof---but actually it
    follows directly by resorting to the inductive hypothesis.
    Let $\tmtwo\,\tmthree \topar{\goals_1 \cup \goalstwo_1} \bigalt_{i=1}^{n} \bigalt_{i'=1}^{n'} \tmtwo_i\,\tmthree_{i'}$
    be derived by rule \rulename{App$_1$}
    from $\tmtwo \topar{\goals_1} \bigalt_{i=1}^{n} \tmtwo_i$
    and $\tmthree \topar{\goalstwo_1} \bigalt_{i'=1}^{n'} \tmthree_{i'}$.
    Similarly,
    let $\tmtwo\,\tmthree \topar{\goals_2 \cup \goalstwo_2} \bigalt_{j=1}^{m} \bigalt_{j'=1}^{m'} \tmtwo^\star_j\,\tmthree^\star_{j'}$
    be derived by rule \rulename{App$_1$}
    from $\tmtwo \topar{\goals_2} \bigalt_{j=1}^{m} \tmtwo^\star_j$
    and $\tmthree \topar{\goalstwo_2} \bigalt_{j'=1}^{m'} \tmthree^\star_{j'}$.

    By \ih on $\tmtwo$, we have that there are sets of goals $\goals'_1, \goals'_2$
    and programs $\prog_1,\hdots,\prog_n,\prog^\star_1,\hdots,\prog^\star_m$
    such that for each $1 \leq i \leq n$ and each $1 \leq j \leq m$:
    \[
      \tmtwo_i \topar{\goals'_2}\structeq \prog_i
      \HS
      \tmtwo^\star_j \topar{\goals'_1}\structeq \prog^\star_j
      \HS
      \bigalt_{i=1}^{n} \prog_i \permeq \bigalt_{j=1}^{m} \prog^\star_j
      \HS
      \goals_1 \cup \goals'_2 = \goals_2 \cup \goals'_1
    \]
    Similarly, by \ih on $\tmthree$, we have that there are sets of goals $\goalstwo'_1, \goalstwo'_2$
    and programs $\progtwo_1,\hdots,\progtwo_{n'},\progtwo^\star_1,\hdots,\progtwo^\star_{m'}$
    such that for each $1 \leq i' \leq n'$ and each $1 \leq j' \leq m'$:
    \[
      \tmthree_{i'} \topar{\goalstwo'_2}\structeq \progtwo_{i'}
      \HS
      \tmthree^\star_{j'} \topar{\goalstwo'_1}\structeq \progtwo^\star_{j'}
      \HS
      \bigalt_{i'=1}^{n'} \progtwo_{i'} \permeq \bigalt_{j'=1}^{m'} \progtwo^\star_{j'}
      \HS
      \goalstwo_1 \cup \goalstwo'_2 = \goalstwo_2 \cup \goalstwo'_1
    \]
    This implies that, for each $1 \leq i \leq n$, $1 \leq j \leq m$, $1 \leq i' \leq n'$, and $1 \leq j' \leq m'$:
    \[
      \indrule{App$_1$}{}{
        \tmtwo_i\,\tmthree_{i'} \topar{\goals'_2 \cup \goalstwo'_2}\structeq \prog_i\,\progtwo_{i'}
      }
      \indrule{App$_1$}{}{
        \tmtwo^\star_j\,\tmthree^\star_{j'} \topar{\goals'_1 \cup \goalstwo'_1}\structeq \prog^\star_j\,\progtwo^\star_{j'}
      }
    \]
    Moreover, note that $\bigalt_{i=1}^{n} \bigalt_{i'=1}^{n'} \prog_i\,\progtwo_{i'} \permeq \bigalt_{j=1}^{m} \bigalt_{j'=1}^{m'} \prog^\star_j\,\progtwo^\star_{j'}$,
    and that $\goals_1 \cup \goalstwo_1 \cup \goals'_2 \cup \goalstwo'_2 = \goals_2 \cup \goalstwo_2 \cup \goals'_1 \cup \goalstwo'_1$.
    This concludes this subcase.
  \item \rulename{App$_1$}/\rulename{App$_2$}:
    Note that $\tmtwo = \laml{\loc}{\var}{\bigalt_{i=1}^{n} \tmfour_i}$ and $\tmthree = \val$,
    which are both values.
    Using the fact that a value only reduces to itself with an empty set of goals (\rlem{simultaneous_reduction_of_a_value}),
    let $(\laml{\loc}{\var}{\bigalt_{i=1}^{n} \tmfour_i})\,\val \topar{\emptyset} (\laml{\loc}{\var}{\bigalt_{i=1}^{n} \tmfour_i})\,\val$
    be derived by \rulename{App$_1$}
    from $\laml{\loc}{\var}{\bigalt_{i=1}^{n} \tmfour_i} \topar{\emptyset} \laml{\loc}{\var}{\bigalt_{i=1}^{n} \tmfour_i}$
    and $\val \topar{\emptyset} \val$.
    Moreover, let $(\laml{\loc}{\var}{\bigalt_{i=1}^{n} \tmfour_i})\,\val \topar{\emptyset} \bigalt_{i=1}^{n} \tmfour_i\sub{\var}{\val}$
    be derived by \rulename{App$_2$}.
    It is then easy to conclude this subcase noting that, for each $1 \leq i \leq n$,
    using reflexivity for terms (\rlem{properties_of_simultaneous_reduction}), we have:
    \[
      \indrule{App$_2$}{}{
        \left(\laml{\loc}{\var}{\bigalt_{i=1}^{n} \tmfour_i}\right)\,\val \topar{\emptyset} \bigalt_{i=1}^{n} \tmfour_i\sub{\var}{\val}
      }
      \indrule{}{\text{(\rlem{properties_of_simultaneous_reduction})}}{
        \tmfour_i\sub{\var}{\val} \topar{\emptyset} \tmfour_i\sub{\var}{\val}
      }
    \]
  \item \rulename{App$_2$}/\rulename{App$_1$}:
    Symmetric to the previous case (\rulename{App$_1$}/\rulename{App$_2$}).
  \item \rulename{App$_2$}/\rulename{App$_2$}:
    There is only one way to derive a reduction using rule \rulename{App$_2$},
    namely $(\lam{\var}{\bigalt_{i=1}^{n} \tmtwo_i})\,\val \topar{\emptyset} \bigalt_{i=1}^{n} \tmtwo_i\sub{\var}{\val}$.
    It is then easy to conclude this subcase noting that, for each $1 \leq i \leq n$,
    using reflexivity for terms (\rlem{properties_of_simultaneous_reduction}), we have:
    \[
      \indrule{}{\text{(\rlem{properties_of_simultaneous_reduction})}}{
        \tmtwo_i\sub{\var}{\val} \topar{\emptyset} \tmtwo_i\sub{\var}{\val} 
      }
    \]
  \end{enumerate}
\item {\em Guarded expression, $\tm = (\tmtwo\seq\tmthree)$.}
  There are four cases, depending on whether each of the simultaneous steps
  is deduced by \rulename{Guard$_1$} or \rulename{Guard$_2$}:
  \begin{enumerate}
  \item \rulename{Guard$_1$}/\rulename{Guard$_1$}:
    This subcase follows directly by resorting to the inductive hypothesis,
    similar to the \rulename{App$_1$}/\rulename{App$_1$} case.
  \item \rulename{Guard$_1$}/\rulename{Guard$_2$}:
    Note that $\tmtwo$ must be a value $\tmtwo = \val$.
    Using the fact that a value only reduces to itself with an empty set of goals (\rlem{simultaneous_reduction_of_a_value}),
    let $\val\seq\tmthree \topar{\goals_1} \bigalt_{i=1}^{n} \val\seq\tmthree_i$
    be derived by \rulename{Guard$_1$} from $\tmthree \topar{\goals_1} \bigalt_{i=1}^{n} \tmthree_i$.
    Moreover, let $\tmtwo\seq\tmthree = \val\seq\tmthree \topar{\goals_2} \bigalt_{j=1}^{m} \tmthree^\star_j$
    be derived from $\tmthree \topar{\goals_2} \bigalt_{j=1}^{m} \tmthree^\star_j$.
    By \ih on $\tmthree$,
    there are sets of goals $\goals'_1,\goals'_2$ and programs $\prog_1,\hdots,\prog_n,\prog^\star_1,\hdots,\prog^\star_m$
    such that for each $1 \leq i \leq n$ and $1 \leq j \leq m$:
    \[
      \tmthree_i \topar{\goals'_2}\structeq \prog_i
      \HS
      \tmthree^\star_j \topar{\goals'_1}\structeq \prog^\star_i
      \HS
      \bigalt_{i=1}^{n} \prog_i \permeq \bigalt_{j=1}^{m} \prog^\star_j
      \HS
      \goals_1 \cup \goals'_2 = \goals_2 \cup \goals'_1
    \]
    To conclude this subcase, note that moreover:
    \[
      \indrule{Guard$_2$}{
        \tmthree_i \topar{\goals'_2}\structeq \prog_i
      }{
        \val\seq\tmthree_i \topar{\goals'_2}\structeq \prog_i
      }
    \]
  \item \rulename{Guard$_2$}/\rulename{Guard$_1$}:
    Symmetric to the previous case (\rulename{Guard$_1$}/\rulename{Guard$_2$}).
  \item \rulename{Guard$_2$}/\rulename{Guard$_2$}:
    Straightforward by \ih. More precisely, let $\val\seq\tmthree \topar{\goals_1} \bigalt_{i=1}^{n} \tmthree_i$
    be derived from $\tmthree \topar{\goals_1} \bigalt_{i=1}^{n} \tmthree_i$
    and, similarly,
    let $\val\seq\tmthree \topar{\goals_2} \bigalt_{j=1}^{m} \tmthree^\star_j$
    be derived from $\tmthree \topar{\goals_2} \bigalt_{j=1}^{m} \tmthree^\star_j$.
    By \ih on $\tmthree$, there are sets of goals $\goals'_1,\goals'_2$ and programs
    $\prog_1,\hdots,\prog_n,\prog^\star_1,\hdots,\prog^\star_m$ such that
    for each $1 \leq i \leq n$ and $1 \leq j \leq m$:
    \[
      \tmthree_i \topar{\goals'_2}\structeq \prog_i
      \HS
      \tmthree^\star_j \topar{\goals'_1}\structeq \prog^\star_i
      \HS
      \bigalt_{i=1}^{n} \prog_i \permeq \bigalt_{j=1}^{m} \prog^\star_j
      \HS
      \goals_1 \cup \goals'_2 = \goals_2 \cup \goals'_1
    \]
    which concludes this subcase.
  \end{enumerate}
\item {\em Unification, $\tm = (\tmtwo\unif\tmthree)$.}
  There are four cases, depending on whether each of the simultaneous steps
  is deduced by \rulename{Unif$_1$} or \rulename{Unif$_2$}:
  \begin{enumerate}
  \item \rulename{Unif$_1$}/\rulename{Unif$_1$}:
    This subcase follows directly by resorting to the inductive hypothesis,
    similar to the \rulename{App$_1$}/\rulename{App$_1$} case.
  \item \rulename{Unif$_1$}/\rulename{Unif$_2$}:
    Note that $\tmtwo$ and $\tmthree$ must both be values,
    \ie $\tmtwo = \val$ and $\tmthree = \valtwo$.
    Using the fact that a value only reduces to itself with an empty set of goals (\rlem{simultaneous_reduction_of_a_value}),
    let $\val \unif \valtwo \topar{\emptyset} \val \unif \valtwo$
    be derived by \rulename{Unif$_1$} from $\val \topar{\emptyset} \val$ and $\valtwo \topar{\emptyset} \valtwo$,
    and
    let $\val \unif \valtwo \topar{\set{\val \unif \valtwo}} \unit$ be derived by \rulename{Unif$_2$}.
    To conclude this subcase, note that:
    \[
      \indrule{Unif$_2$}{}{
        \val \unif \valtwo \topar{\set{\val \unif \valtwo}} \unit
      }
      \indrule{Cons}{}{
        \unit \topar{\emptyset} \unit
      }
    \]
  \item \rulename{Unif$_2$}/\rulename{Unif$_1$}:
    Symmetric to the previous case (\rulename{Unif$_1$}/\rulename{Unif$_2$}).
  \item \rulename{Unif$_2$}/\rulename{Unif$_2$}:
    There is a unique way that the reduction may be derived from rule \rulename{Unif$_2$},
    namely $\val \unif \valtwo \topar{\set{\val \unif \valtwo}} \unit$. To
    conclude this case, note that:
    \[
      \indrule{Cons}{}{
        \unit \topar{\emptyset} \unit
      }
    \]
  \end{enumerate}
\end{enumerate}
\end{proof}

\noindent We now turn to the proof of \rprop{tait_martin_lof_technique} itself:
\begin{enumerate}
\item
  Item 1. of the proposition is precisely \rlem{single_step_included_in_parallel}.
\item
  For item 2. of the proposition, let $\prog \topar{} \progtwo$,
  and proceed by induction on $\prog$.
  If $\prog = \fail$, then $\progtwo = \fail$, and indeed $\prog \rtoca{} \progtwo$
  with the empty reduction sequence.
  If $\prog = \tm \alt \prog'$, then $\progtwo = \progthree \alt \progtwo'$
  where $\tm \topar{} \progthree$ and $\prog' \topar{} \progtwo'$.
  This in turn means that $\tm \topar{\goals} \progthree'$ in such a way that:
  \[
    \progthree \eqdef \begin{cases}
                        \progthree'\SUB{\subst} & \text{if $\subst = \mgu{\goals}$} \\
                        \fail                   & \text{if $\mgu{\goals}$ fails.} \\
                      \end{cases}
  \]
  Then:
  \[
     \begin{array}{rcll}
       \tm \alt \prog'
     & \rtoca{}\structeq &
       \progthree \alt \prog'
       & \text{by \rlem{parallel_included_in_many_step}}
     \\
     & \rtoca{}\structeq &
       \progthree \alt \progtwo'
       & \text{by \ih}
     \end{array}
  \]
  Using the fact that $\structeq$ is a strong bisimulation (\rlem{reduction_modulo_structeq}),
  this implies that
  $\prog = \tm \alt \prog' \rtoca{}\structeq \progthree \alt \progtwo' = \progtwo$,
  as required.
\item
  For item 3. of the proposition
  let $\prog \topar{} \prog_1$ and $\prog \topar{} \prog_2$,
  and proceed by induction on $\prog$.
  If $\prog = \fail$ then $\prog_1 = \prog_2 = \fail$ and the diamond may be closed
  with $\fail \topar{} \fail$ on each side.
  If $\prog = \tm \alt \prog'$
  then $\prog_1 = \progtwo_1 \alt \prog'_1$
  where $\tm \topar{} \progtwo_1$ and $\prog' \topar{} \prog'_1$,
  and similarly
  $\prog_2 = \progtwo_2 \alt \prog'_2$
  where $\tm \topar{} \progtwo_2$ and $\prog' \topar{} \prog'_2$.
  By \ih there are programs $\prog'_3, \prog''_3$ such that $\prog'_1 \topar{} \prog'_3$
  and $\prog'_2 \topar{} \prog''_3 \structeq \prog'_3$.
  Moreover $\tm \topar{\goals_1} \bigalt_{i=1}^{n}{\tm_i}$
  and $\tm \topar{\goals_2} \bigalt_{j=1}^{m}{\tm^\star_j}$ in such a way that:
  \[
    \progtwo_1 = \begin{cases}
                   \bigalt_{i=1}^{n}{\tm_i\SUB{\subst_1}} & \text{if $\subst_1 = \mgu{\goals_1}$} \\
                   \fail                                  & \text{if $\mgu{\goals_1}$ fails} \\
                 \end{cases}
    \HS
    \progtwo_2 = \begin{cases}
                   \bigalt_{j=1}^{m}{\tm^\star_j\SUB{\subst_2}} & \text{if $\subst_2 = \mgu{\goals_2}$} \\
                   \fail                                        & \text{if $\mgu{\goals_2}$ fails} \\
                 \end{cases}
  \]
  By \rlem{diamond_property}, there exist sets of goals $\goals'_1, \goals'_2$
  and programs $\progthree_1,\hdots,\progthree_n,\progthree^\star_1,\hdots,\progthree^\star_m$
  such that, for each $1 \leq i \leq n$ and $1 \leq j \leq m$:
  \[
    \tm_i \topar{\goals'_2} \progthree_i
    \HS
    \tm^\star_j \topar{\goals'_1} \progthree^\star_j
    \HS
    \bigalt_{i=1}^{n} \progthree_i \permeq \bigalt_{j=1}^{m} \progthree^\star_j
    \HS
    \goals_1 \cup \goals'_2 = \goals_2 \cup \goals'_1
  \]
  We consider two subcases, depending on whether $\mgu{\goals_1}$ exists:
  \begin{enumerate}
  \item
    If $\subst_1 = \mgu{\goals_1}$ exists, then by \rlem{parallel_substitution}
    we have that $\tm_i\SUB{\subst_1} \topar{\goals'_2\SUB{\subst_1}} \progthree_i\SUB{\subst_1}$
    for each $1 \leq i \leq n$. We consider two further subcases, depending on whether
    $\mgu{\goals'_2\SUB{\subst_1}}$ exists:
    \begin{enumerate}
    \item
      If $\substtwo_1 = \mgu{\goals'_2\SUB{\subst_1}}$ exists, then
      by the compositionality property (\rlem{mgu_compositional})
      we have that $\mgu{\goals_1 \cup \goals'_2} = \mgu{\goals_2 \cup \goals'_1}$
      also exists, and it is a renaming of $\subst_1\scomp\substtwo_1$.
      Again, by the compositionality property (\rlem{mgu_compositional}),
      this in turn implies that $\subst_2 = \mgu{\goals_2}$ and $\substtwo_2 = \mgu{\goals'_1\SUB{\subst_2}}$ 
      both exist, and $\subst_2\scomp\substtwo_2$ is a renaming of $\subst_1\scomp\substtwo_1$,
      \ie $\subst_2\scomp\substtwo_2 = \subst_1\scomp\substtwo_1\scomp\substthree$ 
      for some renaming $\substthree$. So
      by \rlem{parallel_substitution}
      we have that $\tm^\star_j\SUB{\subst_2} \topar{\goals'_1\SUB{\subst_2}} \progthree^\star_i\SUB{\subst_2}$
      for each $1 \leq j \leq m$, and the situation is:
      \[
        \xymatrix{
          \tm \alt \prog'
            \ar@{=>}[rr]
            \ar@{=>}[d]
        &
        &
          \bigalt_{i=1}^{n} \tm_i\SUB{\subst_1} \alt \prog'_1
            \ar@{=>}[d]
        \\
          \bigalt_{j=1}^{m} \tm^\star_j\SUB{\subst_2} \alt \prog'_2
            \ar@{=>}[r]
        &
          \bigalt_{j=1}^{m} \progthree^\star_j\SUB{\subst_2\scomp\substtwo_2} \alt \prog''_3
          \HS\structeq\!\!\!\!
        &
          \bigalt_{i=1}^{n} \progthree_i\SUB{\subst_1\scomp\substtwo_1} \alt \prog'_3
        }
      \]
      The structural equivalence at the bottom of the diagram is justified as follows:
      \[
        \begin{array}{rcll}
          \bigalt_{j=1}^{m} \progthree^\star_j\SUB{\subst_2\scomp\substtwo_2} \alt \prog''_3
          & \permeq &
          \bigalt_{i=1}^{n} \progthree_i\SUB{\subst_2\scomp\substtwo_2} \alt \prog''_3
          &
          \text{since $\bigalt_{j=1}^{m} \progthree^\star_j \permeq \bigalt_{i=1}^{n} \progthree_i$}
        \\
          & \structeq &
          \bigalt_{i=1}^{n} \progthree_i\SUB{\subst_1\scomp\substtwo_1} \alt \prog''_3
          &
          \text{since $\subst_2\scomp\substtwo_2 = \subst_1\scomp\substtwo_2\scomp\substthree$}
        \\
          & \structeq &
          \bigalt_{i=1}^{n} \progthree_i\SUB{\subst_1\scomp\substtwo_1} \alt \prog'_3
          &
          \text{since $\prog''_3 \structeq \prog'_3$}
        \end{array}
      \]
    \item
      If $\mgu{\goals'_2\SUB{\subst_1}}$ fails, then
      by the compositionality property (\rlem{mgu_compositional})
      we have that $\mgu{\goals_1 \cup \goals'_2} = \mgu{\goals_2 \cup \goals'_1}$
      also fails.
      Again, by the compositionality property (\rlem{mgu_compositional}),
      this in turn implies that either $\subst_2 = \mgu{\goals_2}$ fails or $\mgu{\goals'_1\SUB{\subst_2}}$ fails.
      On one hand, if $\mgu{\goals_2}$ fails, the situation is:
      \[
        \xymatrix{
          \tm \alt \prog'
            \ar@{=>}[r]
            \ar@{=>}[d]
        &
          \bigalt_{i=1}^{n} \tm_i\SUB{\subst_1} \alt \prog'_1
            \ar@{=>}[d]
        \\
          \prog'_2
            \ar@{=>}[r]
        &
          \prog''_3
          \structeq
          \prog'_3
        }
      \]
      On the other hand, if $\subst_2 = \mgu{\goals_2}$ exists and $\mgu{\goals'_1\SUB{\subst_2}}$, the situation is:
      \[
        \xymatrix{
          \tm \alt \prog'
            \ar@{=>}[r]
            \ar@{=>}[d]
        &
          \bigalt_{i=1}^{n} \tm_i\SUB{\subst_1} \alt \prog'_1
            \ar@{=>}[d]
        \\
          \bigalt_{j=1}^{m} \tm^\star_j\SUB{\subst_2} \alt \prog'_2
            \ar@{=>}[r]
        &
          \prog''_3
          \structeq
          \prog'_3
        }
      \]
    \end{enumerate}
  \item
    If $\subst_1 = \mgu{\goals_1}$ fails,
    then by the compositionality property (\rlem{mgu_compositional})
    we have that $\mgu{\goals_1 \cup \goals'_2} = \mgu{\goals_2 \cup \goals'_1}$
    also fails. Again by the compositionality property (\rlem{mgu_compositional})
    this implies that either $\subst_2 = \mgu{\goals_2}$ fails
    or $\substtwo_2 = \mgu{\goals'_1\SUB{\subst_2}}$ fails.
    On one hand, if $\mgu{\goals_2}$ fails, the situation is:
    \[
      \xymatrix{
        \tm \alt \prog'
          \ar@{=>}[r]
          \ar@{=>}[d]
      &
        \prog'_1
          \ar@{=>}[d]
      \\
        \prog'_2
          \ar@{=>}[r]
      &
        \prog''_3
        \structeq
        \prog'_3
      }
    \]
    On the other hand, if $\subst_2 = \mgu{\goals_2}$ exists and $\mgu{\goals'_1\SUB{\subst_2}}$, the situation is:
    \[
      \xymatrix{
        \tm \alt \prog'
          \ar@{=>}[r]
          \ar@{=>}[d]
      &
        \prog'_1
          \ar@{=>}[d]
      \\
        \bigalt_{j=1}^{m} \tm^\star_j\SUB{\subst_2} \alt \prog'_2
          \ar@{=>}[r]
      &
        \prog''_3
        \structeq
        \prog'_3
      }
    \]
  \end{enumerate}
\end{enumerate}

\subsection{Proof of \rprop{subject_reduction} --- Subject Reduction}
\lsec{appendix_subject_reduction}

\begin{definition}[Typing unification problems]
We define the judgment $\tctx \vdash \goals$ for
each unification problem $\goals$ as follows:
\[
  \indrule{}{
    \tctx \vdash \val_i \unif \valtwo_i : \constyp{\unit}
    \HS
    \text{for all $i = 1..n$}
  }{
    \tctx \vdash \set{\val_1 \unif \valtwo_1, \hdots, \val_n \unif \valtwo_n}
  }
\]
\end{definition}

\begin{lemma}[Subject reduction for the unification algorithm]
\llem{mgu_subject_reduction}
Let $\tctx \vdash \goals$ and suppose that $\goals \tounifa{} \goalstwo$
is a step that does not fail. Then $\tctx \vdash \goalstwo$.
\end{lemma}
\begin{proof}
Routine by case analysis on the transition $\goals \tounifa{} \goalstwo$,
using~\rlem{type_system_basic_properties}.
\end{proof}

We turn to the proof of \rprop{subject_reduction} itself.
The proof proceeds by case analysis, depending on the rule applied to conclude
that $\prog \toca{} \progtwo$.
Most cases are straightforward using using~\rlem{type_system_basic_properties}.
The only interesting case is when applying the \rulename{unif} rule.
Then we have that:
\[
  \prog_1 \alt \wctxof{\val\unif\valtwo} \alt \prog_2
  \toca{unif}
  \prog_1 \alt \wctxof{\unit}\SUB{\subst} \alt \prog_2    
\]
where $\subst = \mgu{\set{\val\unif\valtwo}}$.
Moreover, by hypothesis the program is typable, \ie
\[
  \tctx \vdash \prog_1 \alt \wctxof{\val\unif\valtwo} \alt \prog_2 : \typ
\]
By \rlem{program_composition} the following holds for some type $\typtwo$:
\[
  \tctx \vdash \prog_1 : \typ
  \HS
  \tctx, \ctxhole : \typtwo \vdash \wctx : \typ
  \HS
  \tctx \vdash \val\unif\valtwo : \typtwo
  \HS
  \tctx \vdash \prog_2 : \typ
\]
The third judgment can only be derived using the \rulename{t-unif} rule,
so necessarily $\typtwo = \constyp{\unit}$,
and in particular $\tctx \vdash \wctxof{\unit} : \typ$ by
contextual substitution (\rlem{contextual_substitution}).
Note that $\tctx \vdash \set{\val \unif \valtwo}$.
Moreover the most general unifier exists by hypothesis, so
the unification algorithm terminates,
\ie there is a finite sequence of $n \geq 0$ steps:
\[
  \set{\val \unif \valtwo} =
  \goals_0
  \tounifa{}
  \goals_1
  \tounifa{}
  \hdots
  \tounifa{}
  \goals_n = \set{\var_1 \unif \val'_1, \hdots, \var_n \unif \val'_n}
\]
such that for all $i,j$ we have that $\var_i \neq \var_j$
and $\var_i \not\in \fv{\val'_j}$.
Moreover
$\sigma = \mgu{\set{\val \unif \valtwo}} =
 \set{\var_1 \mapsto \val'_1, \hdots, \var_n \mapsto \val'_n}$.
Recall that the unification algorithm preserves typing
(\rlem{mgu_subject_reduction}) so for each $i=1..n$
there is a type $\typthree_i$ such that
$\tctx \vdash \var_i : \typthree_i$
and
$\tctx \vdash \val'_i : \typthree_i$
hold.
This means that
$\tctx$ is of the form
$\tctxtwo,\var_1:\typthree_1,\hdots,\var_n:\typthree_n$.
By repeatedly applying the substitution property (\rlem{substitution}),
we conclude that
$\tctxtwo \vdash
  \wctxof{\unit}\sub{\var_1}{\val'_1}\hdots\sub{\var_n}{\val'_n} : \typ$,
that is
$\tctxtwo \vdash \wctxof{\unit}\SUB{\subst} : \typ$.
Finally, applying \rlem{weakening}
we obtain that the following judgment holds, as required:
\[
  \tctx \vdash \prog_1 \alt \wctxof{\unit}\SUB{\subst} \alt \prog_2 : \typ
\]

\subsection{Proof of \rprop{properties_of_the_denotation} --- Properties of the denotational semantics}
\lsec{appendix_properties_of_the_denotation}

Let us introduce some auxiliary notation.
We write $\fctx$, $\fctx'$, etc. for sequences of
variables ($\fctx = \var_1^{\typ_1},\hdots,\var_n^{\typ_n}$)
without repetition.
If $\vec{\typ} = (\typ_1,\hdots,\typ_n)$ is a sequence of types,
we write $\semtyp{\vec{\typ}}$ for
$\semtyp{\typ_1} \times \hdots \times \semtyp{\typ_n}$.
If $\vec{\var} = (\var_1,\hdots,\var_n)$
is a sequence of variable names,
we write $\vec{\var}^{\vec{\typ}}$
for the sequence $(\var_1^{\typ_1},\hdots,\var_n^{\typ_n})$.
Moreover,
if $\vec{\obj} = (\obj_1,\hdots,\obj_n) \in \semtyp{\vec{\typ}}$
then we write $\asg\asgextend{\vec{\var}}{\vec{\obj}}$
for $\asg\asgextend{\var_1}{\obj_1}\hdots\asgextend{\var_n}{\obj_n}$.
Sometimes we treat sequences of variables as sets, when the
order is not relevant.
If $X$ is a term or a program we define $\semf{X}{\asg}$ as follows,
by induction on $\fctx$:
\[
  \begin{array}{rcl}
    \semf[\emptyctx]{X}{\asg}
  & \eqdef &
    \seme{X}{\asg}
  \\
    \semf[\var^\typ, \fctx]{X}{\asg}
  & \eqdef &
    \set{\objtwo \ST \obj \in \semtyp{\typ},
                     \objtwo \in \semf{X}{\asg\asgextend{\var}{\obj}}}
  \end{array}
\]

\noindent
The following lemma generalizes the {\bf Irrelevance} property of \rlem{irrelevance}.
An easy corollary of this lemma is that
$\seme{\prog}{} = \semf[\fv{\prog}]{\prog}{\asg}$,
whatever be the environment $\asg$.

\begin{lemma}[Irrelevance --- proof of \rlem{irrelevance}, point 1]
\llem{appendix:irrelevance}
Let $\vdash X : \typ$ be a typable term or program.
\begin{enumerate}
\item
  If $\asg, \asg'$ are environments
  that agree on $\fv{X} \setminus \fctx$,
  \ie for any variable $\var^\typtwo \in \fv{X} \setminus \fctx$
  one has that $\asg(\var^\typtwo) = \asg'(\var^\typtwo)$,
  then $\semf{X}{\asg} = \semf{X}{\asg'}$.
\item
  Let $\fctx, \fctx'$ be sequences of variables
  such that $\fv{X} \setminus \fctx = \fv{X} \setminus \fctx'$.
  Then $\semf{X}{\asg} = \semf[\fctx']{X}{\asg}$.
\end{enumerate}
\end{lemma}
\begin{proof}
\begin{enumerate}
\item
  By induction on $\fctx$.
  \begin{enumerate}
  \item
    {\bf Empty, \ie $\fctx = \emptyset$.}
    By induction on $X$, \ie the term or program:
    \begin{enumerate}
    \item {\bf Variable, $X = \var^\typ$.}
      Immediate, as
      $\seme{\var^\typ}{\asg}
      = \asg(\var^\typ)
      = \asg'(\var^\typ)
      = \seme{\var^\typ}{\asg'}$.
    \item {\bf Constructor, $X = \cons$.}
      Immediate, as
      $\seme{\cons}{\asg}
      = \set{\cinterp{\cons}}
      = \seme{\cons}{\asg'}$.
    \item {\bf Abstraction, $X = \lam{\var^\typ}{\prog}$.}
      Note that
      $\seme{\lam{\var^\typ}{\prog}}{\asg} = \set{\objfun}$
      where
      $\objfun(\obj) = \seme{\prog}{\asg\asgextend{\var^\typ}{\obj}}$.
      Symmetrically,
      $\seme{\lam{\var^\typ}{\prog}}{\asg'} = \set{\objfuntwo}$
      where
      $\objfuntwo(\obj) = \seme{\prog}{\asg'\asgextend{\var^\typ}{\obj}}$.
      Note that, for any fixed $\obj \in \semtyp{\typ}$,
      we have that $\asg\asgextend{\var^\typ}{\obj}$
      and $\asg'\asgextend{\var^\typ}{\obj}$
      agree on $\fv{\lam{\var}{\prog}}$ and also on $\var$
      so they agree on $\fv{\prog}$.
      This allows us to apply the \ih to conclude that
      $\seme{\prog}{\asg\asgextend{\var^\typ}{\obj}}
      = \seme{\prog}{\asg'\asgextend{\var^\typ}{\obj}}$,
      so $\objfun = \objfuntwo$ as required.
    \item {\bf Allocated abstraction, $X = \laml{\loc}{\var^\typ}{\prog}$.}
      Similar to the previous case.
    \item {\bf Application, $X = \tm\,\tmtwo$.}
      Straightforward by \ih, as
      $\seme{\tm\,\tmtwo}{\asg}
      = \set{\objtwo \ST \objfun \in \seme{\tm}{\asg},
                         \obj \in \seme{\tmtwo}{\asg},
                         \objtwo \in \objfun(\obj)
        }
      = \set{\objtwo \ST \objfun \in \seme{\tm}{\asg'},
                         \obj \in \seme{\tmtwo}{\asg'},
                         \objtwo \in \objfun(\obj)
        }
      = \seme{\tm\,\tmtwo}{\asg'}$.
    \item {\bf Unification, $X = (\tm \unif \tmtwo)$.}
      Straightforward by \ih as
      $\seme{\tm \unif \tmtwo}{\asg}
      = \set{\cinterp{\unit} \ST
          \obj \in \seme{\tm}{\asg},
          \objtwo \in \seme{\tmtwo}{\asg},
          \obj = \objtwo
        }
      = \set{\cinterp{\unit} \ST
          \obj \in \seme{\tm}{\asg'},
          \objtwo \in \seme{\tmtwo}{\asg'},
          \obj = \objtwo
        }
      = \seme{\tm \unif \tmtwo}{\asg'}$.
    \item {\bf Guarded expression, $X = \tm \seq \tmtwo$.}
      Straightforward by \ih as
      $\seme{\tm\seq\tmtwo}{\asg}
      = \set{\objtwo \ST \obj \in \seme{\tm}{\asg},
                         \objtwo \in \seme{\tmtwo}{\asg}}
      = \set{\objtwo \ST \obj \in \seme{\tm}{\asg'},
                         \objtwo \in \seme{\tmtwo}{\asg'}}
      = \seme{\tm\seq\tmtwo}{\asg'}$.
    \item {\bf Fresh, $X = \fresh{\var^\typ}{\tm}$.}
      Note that
      $\seme{\fresh{\var^\typ}{\tm}}{\asg}
      = \set{\objtwo \ST
          \obj \in \semtyp{\typ},
          \objtwo \in \seme{\tm}{\asg\asgextend{\var^\typ}{\obj}}
        }$.
      Symetrically,
      $\seme{\fresh{\var^\typ}{\tm}}{\asg'}
      = \set{\objtwo \ST
          \obj \in \semtyp{\typ},
          \objtwo \in \seme{\tm}{\asg'\asgextend{\var^\typ}{\obj}}
        }$.
      Note that, for any fixed $\obj \in \semtyp{\typ}$
      we have that $\asg\asgextend{\var^\typ}{\obj}$
      and $\asg'\asgextend{\var^\typ}{\obj}$
      agree on $\fv{\fresh{\var^\typ}{\tm}}$ and also on $\var$,
      so they agree on $\fv{\tm}$.
      This allows us to apply the \ih to conclude that
      $\seme{\tm}{\asg\asgextend{\var^\typ}{\obj}} =
       \seme{\tm}{\asg'\asgextend{\var^\typ}{\obj}}$,
      so $\seme{\fresh{\var^\typ}{\tm}}{\asg}
         = \seme{\fresh{\var^\typ}{\tm}}{\asg'}$,
      as required.
    \item {\bf Fail, $X = \fail^\typ$.}
      Immediate, as
      $\seme{\fail^\typ}{\asg}
      = \emptyset
      = \seme{\fail^\typ}{\asg'}$.
    \item {\bf Alternative, $X = \tm \alt \prog$.}
      Straightforward by \ih as
      $\seme{\tm \alt \prog}{\asg}
      = \seme{\tm}{\asg} \cup \seme{\prog}{\asg}
      = \seme{\tm}{\asg'} \cup \seme{\prog}{\asg'}
      = \seme{\tm \alt \prog}{\asg}$.
    \end{enumerate}
  \item
    {\bf Non-empty, \ie $\fctx = \var^{\typ}, \fctxtwo$.}
    Then note that
    $\asg\asgextend{\var}{\obj}$ and
    $\asg\asgextend{\var}{\obj}$
    agree on $\fv{X} \setminus \fctxtwo$
    for any $\obj \in \semtyp{\typ}$.
    Then:
    \[
    \begin{array}{rcll}
      \semf[\var^\typ,\fctxtwo]{X}{\asg}
    & = &
      \set{ \objtwo \ST
        \obj \in \semtyp{\typ},
        \semf[\fctxtwo]{X}{\asg\asgextend{\var}{\obj}}
      }
    \\
    & = &
      \set{ \objtwo \ST
        \obj \in \semtyp{\typ},
        \semf[\fctxtwo]{X}{\asg'\asgextend{\var}{\obj}}
      }
      & \text{by \ih}
    \\
    & = &
      \semf[\var^\typ,\fctxtwo]{X}{\asg'}
    \end{array}
    \]
  \end{enumerate}
\item
  Note that, seen as sets, $\fv{X} \cap \fctx = \fv{X} \cap \fctx'$
  so the sequence $\fctx$ may be converted into the sequence $\fctx'$
  by repeatedly removing {\em spurious} variables (not in $\fv{X}$),
  adding spurious variables, and swapping variables.
  Indeed, we first note that the two following properties hold:
  \begin{itemize}
  \item {\bf Add/remove spurious variable.}
    $\semf{X}{\asg} = \semf[\var^\typ,\fctx]{X}{\asg}$
    if $\var^\typ \not\in \fv{X}$. \\
    It suffices to show that
    $\semf{X}{\asg} =
     \set{\objtwo \ST
       \obj \in \semtyp{\typ},
       \objtwo \in \semf{X}{\asg\asgextend{\var}{\obj}}
     }$,
     which is immediate since by item 1. of this lemma,
     $\semf{X}{\asg} =
      \semf{X}{\asg\asgextend{\var}{\obj}}$
     for all $\obj \in \semtyp{\typ}$.
     Note that here we crucially use the
     fact that $\semtyp{\typ}$ is a non-empty
     set.
  \item {\bf Swap.}
    $\semf[\fctx_1,\var^\typ,\fctx_2]{X}{\asg} =
     \semf[\var^\typ,\fctx_1,\fctx_2]{X}{\asg}$. \\
    Proceed by induction on $\fctx_1$.
    If $\fctx_1$ is empty, it is immediate.
    Otherwise, let $\fctx_1 = \vartwo^\typtwo,\fctx'_1$.
    Then:
    \[
      \begin{array}{rcll}
        \semf[\vartwo^\typtwo,\fctx'_1,\var^\typ,\fctx_2]{X}{\asg}
      & = &
        \set{
          \objthree \ST
          \objtwo \in \semtyp{\typtwo},
          \objthree \in \semf[\fctx'_1,\var^\typ,\fctx_2]{X}{\asg\asgextend{\vartwo}{\objtwo}}
        }
      \\
      & = &
        \set{
          \objthree \ST
          \objtwo \in \semtyp{\typtwo},
          \objthree \in \semf[\var^\typ,\fctx'_1,\fctx_2]{X}{\asg\asgextend{\vartwo}{\objtwo}}
        }
        & \text{by \ih}
      \\
      & = &
        \set{
          \objthree \ST
          \objtwo \in \semtyp{\typtwo},
          \obj \in \semtyp{\typ},
          \objthree \in \semf[\fctx'_1,\fctx_2]{X}{\asg\asgextend{\vartwo}{\objtwo}\asgextend{\var}{\obj}}
        }
      \\
      & = &
        \set{
          \objthree \ST
          \obj \in \semtyp{\typ},
          \objtwo \in \semtyp{\typtwo},
          \objthree \in \semf[\fctx'_1,\fctx_2]{X}{\asg\asgextend{\var}{\obj}\asgextend{\vartwo}{\objtwo}}
        }
        & (\star)
      \\
      & = &
        \semf[\var^\typ,\vartwo^\typtwo,\fctx'_1,\fctx_2]{X}{\asg}
      \end{array}
    \]
    To justify the $(\star)$ step, note that
    $\asg\asgextend{\vartwo}{\objtwo}\asgextend{\var}{\obj} =
     \asg\asgextend{\var}{\obj}\asgextend{\vartwo}{\objtwo}$ holds
    by definition.
  \end{itemize}
  Now we proceed by induction on $\fctx$:
  \begin{enumerate}
  \item {\bf Empty, \ie $\fctx = \emptyset$.}
    Then $\fv{X} = \fv{X} \setminus \fctx'$
    so $\fctx' \cap \fv{X} = \emptyset$.
    By iteratively adding spurious
    variables we have that
    $\semf{X}{\asg} = \seme{X}{\asg} = \semf[\fctx']{X}{\asg}$
    as required.
  \item {\bf Non-empty, \ie $\fctx = \var^\typ,\fctxtwo$.}
    We consider two subcases, depending on whether the variable $\var^\typ$
    is spurious (\ie $\var^\typ \not\in \fv{X}$) or not:
    \begin{enumerate}
    \item
      If $\var^\typ \not\in \fv{X}$ then note
      that $\fv{X} \setminus \fctxtwo
           = \fv{X} \setminus \fctx
           = \fv{X} \setminus \fctx'$, so
      removing the spurious variable and appyling the \ih
      we have that
      $\semf[\var^\typ,\fctxtwo]{X}{\asg}
      = \semf[\fctxtwo]{X}{\asg}
      = \semf[\fctx']{X}{\asg}$.
    \item
      If $\var^\typ \in \fv{X}$ then
      since $\fv{\tm} \setminus \fctx = \fv{\tm} \setminus \fctx'$
      we have that $\var^\typ \in \fctx'$.
      Hence $\fctx'$ must be of the form
      $\fctx' = \fctx'_1, \var^\typ, \fctx'_2$.
      Then by applying the \ih and swapping we have that
      $\semf[\var^\typ,\fctxtwo]{X}{\asg}
       = \semf[\var^\typ,\fctx'_1,\fctx'_2]{X}{\asg}
       = \semf[\fctx'_1,\var^\typ,\fctx'_2]{X}{\asg}$
      as required.
    \end{enumerate}
  \end{enumerate}
\end{enumerate}
\end{proof}

\noindent
The following lemma generalizes the {\bf Compositionality} property of \rlem{irrelevance}.
\begin{lemma}[Compositionality --- proof of \rlem{compositionality}, point 2]
\llem{appendix:compositionality}
\quad
\begin{enumerate}
\item
  $\semf{\prog \alt \progtwo}{\asg} =
   \semf{\prog}{\asg} \cup \semf{\progtwo}{\asg}$.
\item
  If $\wctx$ is a context whose hole is
  of type $\typ$, then
  $\seme{\wctxof{\tm}}{\asg} =
   \set{\objtwo \ST
        \obj \in \seme{\tm}{\asg},
        \objtwo \in \seme{\wctx}{\asg\asgextend{\ctxhole^\typ}{\obj}}}$.
\end{enumerate}
\end{lemma}
\begin{proof}
\quad
\begin{enumerate}
\item
  By induction on the length of $\fctx$.
  \begin{enumerate}
  	\item {\bf Empty, $\fctx = \emptyset$.}
      Then we proceed by induction on $\prog$:
      \begin{enumerate}
      \item
        If $\prog = \fail$,
        then 
        $\seme{\fail \alt \progtwo}{\asg}
        = \seme{\progtwo}{\asg}
        = \seme{\fail}{\asg} \cup \seme{\progtwo}{\asg}
        $.
      \item
        If $\prog = \tm \alt \prog'$, then:
        \[
        \begin{array}{rcll}
          \seme{(\tm \alt \prog') \alt \progtwo}{\asg}
        & = &
          \seme{\tm \alt (\prog' \alt \progtwo)}{\asg}
        \\
        & = &
          \seme{\tm}{\asg} \cup \seme{\prog' \alt \progtwo}{\asg}
        \\
        & = &
          \seme{\tm}{\asg} \cup \seme{\prog'}{\asg}
          \cup \seme{\progtwo}{\asg}
          & \text{by \ih}
        \\
        & = &
          \seme{\tm \alt \prog'}{\asg} \cup \seme{\progtwo}{\asg}
        \end{array}
        \]
      \end{enumerate}
  	\item {\bf Non-empty, $\fctx = \var^\typ,\fctx'$.}
  	  Then:
  	  \[
  	    \begin{array}{rcll}
    	  \semf[\var^{\typ}, \fctx']{\prog \alt \progtwo}{\asg}
        & = &
          \set{\objtwo \ST
               \obj \in \semtyp{\typ},
               \objtwo \in \semf[\fctx']{\prog \alt \progtwo}{\asg\asgextend{\var}{\obj}}}
		\\
		& = &
          \set{\objtwo \ST
               \obj \in \semtyp{\typ},
               \objtwo \in
               (\semf[\fctx']{\prog}{\asg\asgextend{\var}{\obj}}
               \cup
               \semf[\fctx']{\progtwo}{\asg\asgextend{\var}{\obj}})
          }
          & \text{by \ih}
		\\
        & = &
          \set{\objtwo \ST
               \obj \in \semtyp{\typ},
               \objtwo \in \semf[\fctx']{\prog}{\asg\asgextend{\var}{\obj}}}
          \cup
          \set{\objtwo \ST
               \obj \in \semtyp{\typ},
               \objtwo \in \semf[\fctx']{\progtwo}{\asg\asgextend{\var}{\obj}}}
		\\
		& = &
		  \semf[\var^{\typ}, \fctx']{\prog}{\asg} \cup
          \semf[\var^{\typ}, \fctx']{\progtwo}{\asg}
		\end{array}
	  \]
  \end{enumerate}
\item
  By induction on the structure of the weak context $\wctx$.
  \begin{itemize}
      \item {\bf Empty, $\wctx = \ctxhole$.}
        \[
        \begin{array}{rcll}
        \seme{\tm}{\asg}
        & = & \set{\objtwo \ST
                    \obj \in \seme{\tm}{\asg},
                    \objtwo \in \set{\obj}}
        \\
        & = & \set{\objtwo \ST
                    \obj \in \seme{\tm}{\asg},
                    \objtwo \in \seme{\ctxhole}{\asg\asgextend{\ctxhole^\typ}{\obj}}}
        \end{array}
        \]
  \item {\bf Left of an application, $\wctx = \wctx'\,\tmtwo$.}
    \[
    \begin{array}{rcll}
      \seme{\wctx'\ctxof{\tm}\,\tmtwo}{\asg}
      & = & \set{\objtwo \ST
                  \objfun \in \seme{\wctx'\ctxof{\tm}}{\asg},
                  \objthree \in \seme{\tmtwo}{\asg},
                  \objtwo \in \objfun(\objthree)}
      \\
      & = & \set{\objtwo \ST
                  \obj \in \seme{\tm}{\asg},
                  \objfun \in \seme{\wctx'}{\asg\asgextend{\ctxhole^\typ}{\obj}},
                  \objthree \in \seme{\tmtwo}{\asg},
                  \objtwo \in \objfun(\objthree)}
      & \text{(By \ih)}
      \\
      & = & \set{\objtwo \ST
                  \obj \in \seme{\tm}{\asg},
                  \objfun \in \seme{\wctx'}{\asg\asgextend{\ctxhole^\typ}{\obj}},
                  \objthree \in \seme{\tmtwo}{\asg\asgextend{\ctxhole^\typ}{\obj}},
                  \objtwo \in \objfun(\objthree)}
      & \text{(By \rlem{irrelevance})}
      \\
      & = & \set{\objtwo \ST
                  \obj \in \seme{\tm}{\asg},
                  \objtwo \in \seme{\wctx'\,\tmtwo}{\asg\asgextend{\ctxhole^\typ}{\obj}}}
    \end{array}
    \]
  \item {\bf Right of an application, $\wctx = \tmtwo\,\wctx'$.}
    \[
    \begin{array}{rcll}
      \seme{\tmtwo\,\wctx'\ctxof{\tm}}{\asg}
      & = & \set{\objtwo \ST
                  \objfun \in \seme{\tmtwo}{\asg},
                  \objthree \in \seme{\wctx'\ctxof{\tm}}{\asg},
                  \objtwo \in \objfun(\objthree)}
      \\
      & = & \set{\objtwo \ST
                  \objfun \in \seme{\tmtwo}{\asg},
                  \obj \in \seme{\tm}{\asg},
                  \objthree \in \seme{\wctx'}{\asg\asgextend{\ctxhole^\typ}{\obj}},
                  \objtwo \in \objfun(\objthree)}
      & \text{(By \ih)}
      \\
      & = & \set{\objtwo \ST
                  \obj \in \seme{\tm}{\asg},
                  \objfun \in \seme{\tmtwo}{\asg\asgextend{\ctxhole^\typ}{\obj}},
                  \objthree \in \seme{\wctx'}{\asg\asgextend{\ctxhole^\typ}{\obj}},
                  \objtwo \in \objfun(\objthree)}
      & \text{(By \rlem{irrelevance})}
      \\
      & = & \set{\objtwo \ST
                  \obj \in \seme{\tm}{\asg},
                  \objtwo \in \seme{\tmtwo\,\wctx'}{\asg\asgextend{\ctxhole^\typ}{\obj}}}
    \end{array}
    \]
  \item {\bf Left of a unification, $\wctx = \wctx' \unif \tmtwo$.}
    \[
    \begin{array}{rcll}
      \seme{\wctx'\ctxof{\tm} \unif \tmtwo}{\asg}
      & = & \set{\cinterp{\unit} \ST
                  \objthree \in \seme{\wctx'\ctxof{\tm}}{\asg},
                  \objfour \in \seme{\tmtwo}{\asg},
                  \objthree = \objfour}
      \\
      & = & \set{\cinterp{\unit} \ST
                  \obj \in \seme{\tm}{\asg},
                  \objthree \in \seme{\wctx'}{\asg\asgextend{\ctxhole^\typ}{\obj}},
                  \objfour \in \seme{\tmtwo}{\asg},
                  \objthree = \objfour}
      & \text{(By \ih)}
      \\
      & = & \set{\cinterp{\unit} \ST
                  \obj \in \seme{\tm}{\asg},
                  \objthree \in \seme{\wctx'}{\asg\asgextend{\ctxhole^\typ}{\obj}},
                  \objfour \in \seme{\tmtwo}{\asg\asgextend{\ctxhole^\typ}{\obj}},
                  \objthree = \objfour}
      & \text{(By \rlem{irrelevance})}
      \\
      & = & \set{\objtwo \ST
                  \obj \in \seme{\tm}{\asg},
                  \objtwo \in \seme{\wctx' \unif \tmtwo}{\asg\asgextend{\ctxhole^\typ}{\obj}}}
    \end{array}
    \]
  \item {\bf Right of a unification, $\wctx = \tmtwo \unif \wctx'$.}
    \[
    \begin{array}{rcll}
      \seme{\tmtwo \unif \wctx'\ctxof{\tm}}{\asg}
      & = & \set{\cinterp{\unit} \ST
                  \objthree \in \seme{\tmtwo}{\asg},
                  \objfour \in \seme{\wctx'\ctxof{\tm}}{\asg},
                  \objthree = \objfour}
      \\
      & = & \set{\cinterp{\unit} \ST
                  \obj \in \seme{\tm}{\asg},
                  \objthree \in \seme{\tmtwo}{\asg},
                  \objfour \in \seme{\wctx'}{\asg\asgextend{\ctxhole^\typ}{\obj}},
                  \objthree = \objfour}
      & \text{(By \ih)}
      \\
      & = & \set{\cinterp{\unit} \ST
                  \obj \in \seme{\tm}{\asg},
                  \objthree \in \seme{\tmtwo}{\asg\asgextend{\ctxhole^\typ}{\obj}},
                  \objfour \in \seme{\wctx'}{\asg\asgextend{\ctxhole^\typ}{\obj}},
                  \objthree = \objfour}
      & \text{(By \rlem{irrelevance})}
      \\
      & = & \set{\objtwo \ST
                  \obj \in \seme{\tm}{\asg},
                  \objtwo \in \seme{\tmtwo \unif \wctx'}{\asg\asgextend{\ctxhole^\typ}{\obj}}}
    \end{array}
    \]
  \item {\bf Left of a guarded expression, $\wctx = \wctx' \seq \tmtwo$.}
    \[
    \begin{array}{rcll}
      \seme{\wctx'\ctxof{\tm} \seq \tmtwo}{\asg}
      & = & \set{\objtwo \ST
                  \objthree \in \seme{\wctx'\ctxof{\tm}}{\asg},
                  \objtwo \in \seme{\tmtwo}{\asg}}
      \\
      & = & \set{\objtwo \ST
                  \obj \in \seme{\tm}{\asg},
                  \objthree \in \seme{\wctx'}{\asg\asgextend{\ctxhole^\typ}{\obj}},
                  \objtwo \in \seme{\tmtwo}{\asg}}
      & \text{(By \ih)}
      \\
      & = & \set{\objtwo \ST
                  \obj \in \seme{\tm}{\asg},
                  \objthree \in \seme{\wctx'}{\asg\asgextend{\ctxhole^\typ}{\obj}},
                  \objtwo \in \seme{\tmtwo}{\asg\asgextend{\ctxhole^\typ}{\obj}}}
      & \text{(By \rlem{irrelevance})}
      \\
      & = & \set{\objtwo \ST
                  \obj \in \seme{\tm}{\asg},
                  \objtwo \in \seme{\wctx' \seq \tmtwo}{\asg\asgextend{\ctxhole^\typ}{\obj}}}
    \end{array}
    \]
  \item {\bf Right of a guarded expression, $\wctx = \tmtwo \seq \wctx'$.}
    \[
    \begin{array}{rcll}
      \seme{\tmtwo \seq \wctx'\ctxof{\tm}}{\asg}
      & = & \set{\objtwo \ST
                  \objthree \in \seme{\tmtwo}{\asg},
                  \objtwo \in \seme{\wctx'\ctxof{\tm}}{\asg}}
      \\
      & = & \set{\objtwo \ST
                  \obj \in \seme{\tm}{\asg},
                  \objthree \in \seme{\tmtwo}{\asg},
                  \objtwo \in \seme{\wctx'}{\asg\asgextend{\ctxhole^\typ}{\obj}}}
      & \text{(By \ih)}
      \\
      & = & \set{\objtwo \ST
                  \obj \in \seme{\tm}{\asg},
                  \objthree \in \seme{\tmtwo}{\asg\asgextend{\ctxhole^\typ}{\obj}},
                  \objtwo \in \seme{\wctx'}{\asg\asgextend{\ctxhole^\typ}{\obj}}}
      & \text{(By \rlem{irrelevance})}
      \\
      & = & \set{\objtwo \ST
                  \obj \in \seme{\tm}{\asg},
                  \objtwo \in \seme{\tmtwo \seq \wctx'}{\asg\asgextend{\ctxhole^\typ}{\obj}}}
    \end{array}
    \]
  \end{itemize}
\end{enumerate}
\end{proof}

\begin{lemma}[Free variables]
The following hold:
\begin{enumerate}
\item $\fv{\prog \alt \progtwo} = \fv{\prog} \cup \fv{\progtwo}$
\item $\fv{\wctxof{\tm}} = \fv{\wctx} \cup \fv{\tm}$
\item $\fv{\wctxof{\prog}} = \fv{\wctx} \cup \fv{\prog}$
\item $\fv{\tm\SUB{\subst}} \subseteq (\fv{\tm} \setminus \supp{\subst}) \cup \bigcup_{\var \in \supp{\subst}} \fv{\subst(\var)}$
\item $\fv{\prog\SUB{\subst}} \subseteq (\fv{\prog} \setminus \supp{\subst}) \cup \bigcup_{\var \in \supp{\subst}} \fv{\subst(\var)}$
\end{enumerate}
\end{lemma}
\begin{proof}
Routine by induction on $\prog$, $\wctx$, or $\tm$, correspondingly.
\end{proof}

\begin{lemma}[Interpretation of values --- proof of \rlem{appendix:interpretation_of_values}, point 3]
\llem{appendix:interpretation_of_values}
If $\val$ is a value then $\seme{\val}{\asg}$ is a singleton.
\end{lemma}
\begin{proof}
By induction on $\val$.
If $\val$ is a variable or an allocated
abstraction, it is immediate,
so let $\val = \cons\,\val_1\hdots\val_n$.
In that case, by induction on $n$
we claim that
$\seme{\cons\,\val_1\hdots\val_n}{\asg}$
is a singleton of
the form $\set{\obj}$
where
moreover $\obj$ is unitary:
\begin{enumerate}
\item {\bf If $n = 0$.}
  Then $\seme{\cons}{\asg} = \set{\cinterp{\cons}}$,
  which is a singleton.
  Moreover, recall that $\cinterp{\cons}$
  is always requested to be unitary.
\item {\bf If $n > 0$.}
  Then by \ih of the innermost induction
  $\seme{\cons\,\val_1\hdots\val_{n-1}}{\asg}$
  is a singleton of the form $\set{\objfun_0}$,
  where $\objfun_0$ is unitary,
  and by \ih of the outermost induction
  $\seme{\val_n}{\asg}$
  is a singleton of the form $\set{a_0}$,
  so we have that:
  \[
    \begin{array}{rcll}
      \seme{\cons\,\val_1\hdots\val_{n-1}\,\val_n}{\asg}
      & = & \set{\objtwo \ST
                  \objfun \in \seme{\cons\,\val_1\hdots\val_{n-1}}{\asg},
                  \obj \in \seme{\val_n}{\asg},
                  \objtwo \in \objfun(\obj)}
      \\
      & = & \objfun_0(\obj_0)
    \end{array}
  \]
Since $\objfun_0$ is unitary,
$\objfun_0(\obj_0)$ is a singleton of the form
$\set{\objtwo}$,
where $\objtwo$ is unitary, as required.
\end{enumerate}
\end{proof}

\begin{lemma}[Interpretation of substitution --- proof of \rlem{interpretation_of_substitution}, point 4]
\llem{appendix:interpretation_of_substitution}
Let $\subst = \set{\var^{\typ_1}_1\mapsto\val_1,\hdots,\var^{\typ_n}_n\mapsto\val_n}$
be a substitution with support $\set{\var^{\typ_1}_1,\hdots,\var^{\typ}_n}$
and such that $\var_i \notin \fv{\val_j}$ for any two $1 \leq i,j \leq n$.
Recall that
the interpretation of a value is always a singleton~(\rlem{interpretation_of_values}),
so let $\seme{\val_i}{\asg} = \set{\obj_i}$ for each $i=1..n$.
Then:
\begin{enumerate}
\item $\seme{\tm\SUB{\subst}}{\asg} =
       \seme{\tm}{\asg\asgextend{\var_1}{\obj_1}\hdots\asgextend{\var_n}{\obj_n}}$
\item $\seme{\prog\SUB{\subst}}{\asg} =
       \seme{\prog}{\asg\asgextend{\var_1}{\obj_1}\hdots\asgextend{\var_n}{\obj_n}}$
\end{enumerate}
\end{lemma}
\begin{proof}
By simultaneous induction on the term $\tm$ (resp. program $\prog$).
\begin{enumerate}
\item {\bf Variable, $\tm = \var^\typ$.}
  There are two subcases, depending on whether
  $\var \in \set{\var_1,\hdots,\var_n}$ or not.
  \begin{enumerate}
  \item
    If $\var = \var_i$ for some $1 \leq i \leq n$, then:
    \[
      \seme{(\var_i^{\typ})\SUB{\subst}}{\asg}
      = \seme{\val_i}{\asg}
      = \set{\obj_i}
      = \seme{\var_i^\typ}{
          \asg\asgextend{\var_1}{\obj_1}\hdots\asgextend{\var_n}{\obj_n}
        }
    \]
  \item
    If $\var \notin \set{\var_1,\hdots,\var_n}$,
    then:
    \[
        \seme{(\var^\typ)\SUB{\subst}}{\asg} 
      = \asg(\var^\typ)
      = \seme{\var^\typ}{
         \asg\asgextend{\var_1}{\obj_1}\hdots\asgextend{\var_n}{\obj_n}
        }
    \]
  \end{enumerate}
\item {\bf Constructor, $\tm = \cons$.}
  Immediate, as:
  \[
      \seme{\cons\SUB{\subst}}{\asg}
    = \seme{\cons}{\asg}
    = \set{\cinterp{\cons}}
    = \seme{\cons}{
        \asg\asgextend{\var_1}{\obj_1}\hdots\asgextend{\var_n}{\obj_n}
      }
  \]
\item {\bf Abstraction code, $\tm = \lam{\var^\typ}{\prog}$.}
  Then:
  \[
    \begin{array}{rcll}
      \seme{(\lam{\var^\typ}{\prog})\SUB{\subst}}{\asg}
    & = &
      \seme{\lam{\var^\typ}{\prog\SUB{\subst}}}{\asg}
    \\
    & = &
      \set{\objfun}
      \HS\text{where
        $\objfun(\obj) = \seme{\prog\SUB{\subst}}{\asg\asgextend{\var}{\obj}}$
      }
    \\
    & = &
      \set{\objfuntwo}
      \HS\text{where
        $\objfuntwo(\obj) =
         \seme{\prog}{\asg\asgextend{\var}{\obj}\asgextend{\var_1}{\obj_1}\hdots\asgextend{\var_n}{\obj_n}}$
      }
      & \text{(By \ih)}
    \\
    & = &
      \seme{\lam{\var^\typ}{\prog}}{\asg\asgextend{\var_1}{\obj_1}\hdots\asgextend{\var_n}{\obj_n}}
    \end{array}
  \]
\item {\bf Allocated abstraction, $\tm = \laml{\loc}{\var}{\prog}$.}
  Similar to the previous case.
\item {\bf Application, $\tm = \tmtwo\,\tmthree$.}
  Then:
  \[
  \begin{array}{rcll}
    \seme{(\tmtwo\,\tmthree)\SUB{\subst}}{\asg}
  & = &
    \seme{\tmtwo\SUB{\subst}\,\tmthree\SUB{\subst}}{\asg}
  \\
  & = &
    \set{\objtwo \ST 
                \objfun \in \seme{\tmtwo\SUB{\subst}}{\asg},
                \obj \in \seme{\tmthree\SUB{\subst}}{\asg},
                \objtwo \in \objfun(\obj)}
  \\
  & = &
    \set{\objtwo \ST 
                \objfun \in \seme{\tmtwo}{\asg\asgextend{\var_1}{\obj_1}\hdots\asgextend{\var_n}{\obj_n}},
                \obj \in \seme{\tmthree}{\asg\asgextend{\var_1}{\obj_1}\hdots\asgextend{\var_n}{\obj_n}},
                \objtwo \in \objfun(\obj)}
    & \text{(By \ih)}
  \\
  & = &
    \seme{\tmtwo\,\tmthree}{\asg\asgextend{\var_1}{\obj_1}\hdots\asgextend{\var_n}{\obj_n}}
  \end{array}
  \]  
\item {\bf Unification, $\tm = (\tmtwo \unif \tmthree)$.}
  Then:
  \[
  \begin{array}{rcll}
    \seme{(\tmtwo \unif \tmthree)\SUB{\subst}}{\asg}
  & = &
    \seme{\tmtwo\SUB{\subst} \unif \tmthree\SUB{\subst}}{\asg}
  \\
  & = &
    \set{\cinterp{\unit} \ST
                \obj \in \seme{\tmtwo\SUB{\subst}}{\asg},
                \objtwo \in \seme{\tmthree\SUB{\subst}}{\asg},
                \obj = \objtwo}
  \\
  & = &
    \set{\cinterp{\unit} \ST
                \obj \in \seme{\tmtwo}{\asg\asgextend{\var_1}{\obj_1}\hdots\asgextend{\var_n}{\obj_n}},
                \objtwo \in \seme{\tmthree}{\asg\asgextend{\var_1}{\obj_1}\hdots\asgextend{\var_n}{\obj_n}},
                \obj = \objtwo}
    & \text{(By \ih)}
  \\
  & = &
    \seme{\tmtwo \unif \tmthree}{\asg\asgextend{\var_1}{\obj_1}\hdots\asgextend{\var_n}{\obj_n}}
  \end{array}
  \]  
\item {\bf Guarded expression, $\tm = \tmtwo\seq\tmthree$.}
  Then:
  \[
  \begin{array}{rcll}
    \seme{(\tmtwo \seq \tmthree)\SUB{\subst}}{\asg}
  & = &
    \seme{\tmtwo\SUB{\subst} \seq \tmthree\SUB{\subst}}{\asg}
  \\
  & = &
    \set{\obj \ST
                \objtwo \in \seme{\tmtwo\SUB{\subst}}{\asg},
                \obj \in \seme{\tmthree\SUB{\subst}}{\asg}}
  \\
  & = &
    \set{\obj \ST
                \objtwo \in \seme{\tmtwo}{\asg\asgextend{\var_1}{\obj_1}\hdots\asgextend{\var_n}{\obj_n}},
                \obj \in \seme{\tmthree}{\asg\asgextend{\var_1}{\obj_1}\hdots\asgextend{\var_n}{\obj_n}}}
    & \text{(By \ih)}
  \\
  & = & 
    \seme{\tmtwo \seq \tmthree}{\asg\asgextend{\var_1}{\obj_1}\hdots\asgextend{\var_n}{\obj_n}}
  \end{array}
  \]  
\item {\bf Fresh, $\tm = \fresh{\var^\typ}{\tmtwo}$.}
  Then:
  \[
  \begin{array}{rcll}
    \seme{(\fresh{\var^\typ}{\tmtwo})\SUB{\subst}}{\asg}
  & = &
    \seme{\fresh{\var^\typ}{\tmtwo}\SUB{\subst}}{\asg}
  \\
    & = & \set{\objtwo \ST
                \obj \in \semtyp{\typ},
                \objtwo \in \seme{\tmtwo\SUB{\subst}}{\asg\asgextend{\var}{\obj}}}
    \\
    & = & \set{\objtwo \ST
                \obj \in \semtyp{\typ},
                \objtwo \in\seme{\tmtwo}{\asg\asgextend{\var}{\obj}\asgextend{\var_1}{\obj_1}\hdots\asgextend{\var_n}{\obj_n}}}
    & \text{(By \ih)}
    \\
    & = & \set{\objtwo \ST
                \obj \in \semtyp{\typ},
                \objtwo \in\seme{\tmtwo}{\asg\asgextend{\var_1}{\obj_1}\hdots\asgextend{\var_n}{\obj_n}\asgextend{\var}{\obj}}}
    & \text{(Since $\var \notin \set{\var_1,\hdots,\var_n}$)}
    \\
    & = & \seme{\fresh{\var^\typ}{\tmtwo}}{\asg\asgextend{\var_1}{\obj_1}\hdots\asgextend{\var_n}{\obj_n}}
  \end{array}
  \]  
\item {\bf Fail, $\prog = \fail$.}
  Immediate, as:
  \[
    \seme{\fail\SUB{\subst}}{\asg}
     = \seme{\fail}{\asg}
     = \emptyset
     = \seme{\fail}{\asg\asgextend{\var_1}{\obj_1}\hdots\asgextend{\var_n}{\obj_n}}
  \]
\item {\bf Alternative, $\prog = \tm\alt\prog$.}
  Then:
  \[
  \begin{array}{rcll}
    \seme{(\tm\alt\prog)\SUB{\subst}}{\asg}
  & = &
    \seme{\tm\SUB{\subst} \alt \prog\SUB{\subst}}{\asg}
  \\
  & = & 
    \seme{\tm\SUB{\subst}}{\asg} \cup \seme{\prog\SUB{\subst}}{\asg}
  \\
  & = &
    \seme{\tm}{\asg\asgextend{\var_1}{\obj_1}\hdots\asgextend{\var_n}{\obj_n}}
    \cup
    \seme{\prog}{\asg\asgextend{\var_1}{\obj_1}\hdots\asgextend{\var_n}{\obj_n}}
    & \text{(By \ih)}
  \\
  & = &
    \seme{\tm \alt \prog}{\asg\asgextend{\var_1}{\obj_1}\hdots\asgextend{\var_n}{\obj_n}}
  \end{array}
  \]  
\end{enumerate}
\end{proof}

\subsection{Proof of \rthm{soundness} --- Soundness}
\lsec{appendix_soundness}

\begin{definition}[Goal satisfaction]
Let $\asg$ be a fixed variable assignment,
and let $\vec{\var}^{\vec{\typ}}$ be a fixed sequence of variables.
Moreover, let $\goals = \set{(\val_1 \unif \valtwo_1), \hdots, (\val_n \unif \valtwo_n)}$
be a unification problem.
Given a sequence of elements $\vec{\obj} \in \semtyp{\vec{\typ}}$
we say that $\vec{\obj}$ {\em satisfies} $\goals$ (with respect to $\asg,\vec{\var}$),
written $\vec{\obj} \vDash_{\asg,\vec{\var}} \goals$, if and only if
$\seme{\val_i}{\asg\asgextend{\vec{\var}}{\vec{\obj}}} =
 \seme{\valtwo_i}{\asg\asgextend{\vec{\var}}{\vec{\obj}}}$
for all $i=1..n$. We write $\vec{\obj} \vDash \goals$
if $\asg$ and $\vec{\var}$ are clear from the context.
\end{definition}

\begin{lemma}[Unification preserves satisfaction]
\llem{unification_preserves_satisfaction}
Let $\goals \tounifa{} \goalstwo$ be a step of the unification algorithm
that does not fail. Then for any $\asg, \vec{\var}^{\vec{\typ}}$
we have that:
\[
  \set{
    \vec{\obj} \ST
    \vec{\obj} \vDash_{\asg,\vec{\var}} \goals
  }
  =
  \set{
    \vec{\obj} \ST
    \vec{\obj} \vDash_{\asg,\vec{\var}} \goalstwo
  }
\]
\end{lemma}
\begin{proof}
  Note that the step does not fail so it cannot be the result of applying
  the \rulename{u-clash} or the \rulename{u-occurs-check} rules.
  We consider the five remaining cases:
  \begin{enumerate}
  \item \rulename{u-delete}:
    Our goal is to prove that:
    \[ \set{\vec{\obj} \ST \vec{\obj} \vDash_{\asg,\vec{\var}} \set{\vartwo \unif \vartwo} \uplus \goals'}
    = \set{\vec{\obj} \ST \vec{\obj} \vDash_{\asg,\vec{\var}} \goals'}
    \]
    This is immediate since
    $\seme{\vartwo}{\asg\asgextend{\vec{\var}}{\vec{\obj}}}
    = \seme{\vartwo}{\asg\asgextend{\vec{\var}}{\vec{\obj}}}$
    always holds.
  \item \rulename{u-orient}:
    Our goal is to prove that:
    $ \set{\vec{\obj} \ST \vec{\obj} \vDash_{\asg,\vec{\var}} \set{\val \unif \vartwo} \uplus \goals'}
    = \set{\vec{\obj} \ST \vec{\obj} \vDash_{\asg,\vec{\var}} \set{\vartwo \unif \val} \uplus \goals'} $.
    Immediate by definition.
  \item \rulename{u-match-lam}:
    Our goal is to prove that:
    \[ \set{\vec{\obj} \ST
        \vec{\obj} \vDash_{\asg,\vec{\var}} \set{\laml{\loc}{\vartwo}{\prog} \unif \laml{\loc}{\vartwo}{\prog}} \uplus \goals'}
    = \set{\vec{\obj} \ST \vec{\obj} \vDash_{\asg,\vec{\var}} \goals'}
    \]
    This is immediate since
    $\seme{\laml{\loc}{\vartwo}{\prog}}{\asg\asgextend{\vec{\var}}{\vec{\obj}}}
    = \seme{\laml{\loc}{\vartwo}{\prog}}{\asg\asgextend{\vec{\var}}{\vec{\obj}}}$
    always holds.
  \item \rulename{u-match-cons}:
    Our goal is to prove that:
    \[ \set{\vec{\obj} \ST
        \vec{\obj} \vDash_{\asg,\vec{\var}} \set{\cons\,\val_1\hdots\val_n \unif \cons\,\valtwo_1\hdots\valtwo_n}
                                             \uplus \goals'}
    = \set{\vec{\obj} \ST 
        \vec{\obj} \vDash_{\asg,\vec{\var}} \set{\val_1 \unif \valtwo_1, \hdots, \val_n \unif \valtwo_n} \uplus \goals'}
    \]
    Recall that $\seme{\cons}{\asg\asgextend{\vec{\var}}{\vec{\obj}}} = \set{\cinterp{\cons}}$ is
    $\constyp{\cons}$-unitary,
    and the interpretation of a value is always a singleton~(\rlem{interpretation_of_values}),
    so let $\seme{\val_i}{\asg\asgextend{\vec{\var}}{\vec\obj}} = \set{\objtwo_i}$
    and $\seme{\valtwo_i}{\asg\asgextend{\vec{\var}}{\vec\obj}} = \set{\objtwo'_i}$.
    It suffices to note that:
    \[
    \begin{array}{rcll}
    &&  
      \vec{\obj} \vDash_{\asg,\vec{\var}} \set{\cons\,\val_1\hdots\val_n \unif \cons\,\valtwo_1\hdots\valtwo_n}
    \\
    & \iff &
      \seme{\cons\,\val_1\hdots\val_n}{\asg\asgextend{\vec{\var}}{\vec{\obj}}}
      = \seme{\cons\,\valtwo_1\hdots\valtwo_n}{\asg\asgextend{\vec{\var}}{\vec{\obj}}}
    \\
    & \iff &
      \cinterp{\cons}(\objtwo_1)\hdots(\objtwo_n)
      =
      \cinterp{\cons}(\objtwo'_1)\hdots(\objtwo'_n)
    \\
    & \iff &
      \objtwo_i = \objtwo'_i
      \text{\ for all $i = 1..n$}
    & \text{($\star$)}
    \\
    & \iff &
      \seme{\val_i}{\asg\asgextend{\vec{\var}}{\vec{\obj}}} = \seme{\valtwo_i}{\asg\asgextend{\vec{\var}}{\vec{\obj}}},
      \text{\ for all $i = 1..n$}
    \\
    & \iff &
      \vec{\obj} \vDash_{\asg,\vec{\var}} \set{\val_1 \unif \valtwo_1, \hdots, \val_n \unif \valtwo_n}
    \end{array}
    \]
    The step $(\star)$ is justified by the fact
    that we assume that constructors are injective.
  \item \rulename{u-eliminate}:
    Our goal is to prove that:
    \[
      \set{\vec{\obj} \ST \vec{\obj} \vDash_{\asg,\vec{\var}} \set{\vartwo \unif \val} \uplus \goals'}
    = \set{\vec{\obj} \ST \vec{\obj} \vDash_{\asg,\vec{\var}} \set{\vartwo \unif \val} \uplus \goals'\sub{\vartwo}{\val}}
    \]
    if $\vartwo \in \fv{\goals'}\setminus\fv{\val}$.
    Moreover, let $\goals' = \set{(\val_1 \unif \valtwo_1),\hdots,(\val_n \unif \valtwo_n)}$.
    Recall that the interpretation of a value is always a singleton~(\rlem{interpretation_of_values}),
    so let $\seme{\val}{\asg\asgextend{\vec{\var}}{\vec{\obj}}} = \set{\objtwo}$.
    Let $\vec{\obj} \in \semtyp{\vec{\typ}}$.
    It suffices to show that whenever
    $\asg\asgextend{\vec{\var}}{\vec{\obj}}(\vartwo) = \objtwo$
    then the following equivalence holds:
    \[
      \vec{\obj} \vDash_{\asg,\vec{\var}} \goals'
      \iff
      \vec{\obj} \vDash_{\asg,\vec{\var}} \goals'\sub{\vartwo}{\val}
    \]
    Note that, for each fixed $i=1..n$:
    \[
      \begin{array}{rcll}
        \seme{\val_i}{\asg\asgextend{\vec{\var}}{\vec{\obj}}}
      & = &
        \seme{\val_i}{\asg\asgextend{\vec{\var}}{\vec{\obj}}\asgextend{\vartwo}{\objtwo}}
        & (\star)
      \\
      & = &
        \seme{\val_i\sub{\vartwo}{\val}}{\asg\asgextend{\vec{\var}}{\vec{\obj}}}
        & \text{(By~\rlem{interpretation_of_substitution})}
      \end{array}
    \]
    The step $(\star)$ is trivial because, as we have
    already noted,
    $\asg\asgextend{\vec{\var}}{\vec{\obj}}(\vartwo) = \objtwo$
    so
    $\asg\asgextend{\vec{\var}}{\vec{\obj}}$
    and
    $\asg\asgextend{\vec{\var}}{\vec{\obj}}\asgextend{\vartwo}{\objtwo}$
    are the same variable assignment.
    And, similarly,
    $\seme{\valtwo_i}{\asg\asgextend{\vec{\var}}{\vec{\obj}}}
     =
     \seme{\valtwo_i\sub{\vartwo}{\val}}{\asg\asgextend{\vec{\var}}{\vec{\obj}}}$.
    Then:
    \[
    \begin{array}{rcll}
    &&  
      \vec{\obj} \vDash_{\asg,\vec{\var}} \goals'
    \\
    & \iff &
      \seme{\val_i}{\asg\asgextend{\vec{\var}}{\vec{\obj}}} =
      \seme{\valtwo_i}{\asg\asgextend{\vec{\var}}{\vec{\obj}}}
      \text{ for all $i = 1..n$}
    \\
    & \iff &
      \seme{\val_i\sub{\vartwo}{\val}}{\asg\asgextend{\vec{\var}}{\vec{\obj}}} =
      \seme{\valtwo_i\sub{\vartwo}{\val}}{\asg\asgextend{\vec{\var}}{\vec{\obj}}}
      \text{ for all $i = 1..n$}
      & \text{(\rlem{interpretation_of_substitution})}
    \\
    & \iff &
      \vec{\obj} \vDash_{\asg,\vec{\var}} \goals'\sub{\vartwo}{\val}
    \end{array}
    \]
  \end{enumerate}
\end{proof}

The following theorem generalizes \rthm{soundness}:

\begin{theorem}[Soundness]
\lthm{appendix:soundness}
Let $\tctx \vdash \prog : \typ$
and $\prog \toca{} \progtwo$.
Let $\fctx = \fv{\prog}$
and $\fctx' = \fv{\progtwo}$.
Then for any variable assignment $\asg$:
\[
  \semf[\fctx]{\prog}{\asg} \supseteq \semf[\fctx']{\progtwo}{\asg}
\]
Moreover, the inclusion is an equality for
all reduction rules other than the
\indrulename{fail} rule.
\end{theorem}
\begin{proof}
Let $\prog \toca{} \progtwo$.
We consider six cases, depending on the rule applied to conclude
that $\prog \toca{} \progtwo$:
\begin{enumerate}
\item \rulename{alloc}:
  Note that $\fctx = \fctx'$, and
  suppose that
  $\fctx = \vec{\vartwo}^{\vec{\typtwo}}$.
  Then:
  \[
    \begin{array}{rcll}
      \semf[\fctx]{\prog_1 \alt \wctxof{\lam{\var^\typ}{\progtwo}} \alt \prog_2}{\asg}
    & = &
      \set{
        \obj
      \ST
        \vec{\objtwo} \in \semtyp{\vec{\typtwo}},
        \obj \in
        \seme{\prog_1 \alt \wctxof{\lam{\var^\typ}{\progtwo}} \alt  \prog_2}{
          \asg\asgextend{\vec{\vartwo}}{\vec{\objtwo}}
        }
      }
    \\
    & = &
      \set{
        \obj
      \ST
        \vec{\objtwo} \in \semtyp{\vec{\typtwo}},
        \obj \in
        \seme{\prog_1 \alt \wctxof{\laml{\loc}{\var^\typ}{\progtwo}} \alt  \prog_2}{
          \asg\asgextend{\vec{\vartwo}}{\vec{\objtwo}}
        }
      }
     & (\star)
    \\
    & = &
      \semf[\fctx]{\prog_1 \alt \wctxof{\laml{\loc}{\var^\typ}{\progtwo}} \alt \prog_2}{\asg}
    \end{array}
  \]
  To justify $(\star)$, note that,
  by Compositionality~(\rlem{compositionality}),
  it suffices to prove that
  $\seme{\lam{\var^\typ}{\progtwo}}{\asg\asgextend{\vec{\vartwo}}{\vec{\objtwo}}}
  = \seme{\laml{\loc}{\var^\typ}{\progtwo}}{\asg\asgextend{\vec{\vartwo}}{\vec{\objtwo}}}$
  for all $\vec{\objtwo} \in \semtyp{\typtwo}$.
  This holds by definition so we are done.
\item \rulename{beta}:
  Note that $\fctx = \fctx',\vec{\varthree}^{\vec{\typthree}}$
  where
  \[
    \vec{\varthree} =
      \begin{cases}
         \emptyset & \text{if $\var \notin \fv{\progtwo}$} \\
         \fv{\val} \setminus \fv{\prog_1 \alt \wctxof{(\lam{\var}{\progtwo})\ctxhole} \alt \prog_2}
                   & \text{if $\var \in \fv{\progtwo}$}
       \end{cases}
  \]
  Moreover, suppose that
  $\fctx' = \vec{\vartwo}^{\vec{\typtwo}}$.
  Then:
  \[
    \begin{array}{rcll}
    &&
    \semf[\vec{\vartwo}^{\vec{\typtwo}},\vec{\varthree}^{\vec{\typthree}}]{
      \prog_1 \alt \wctxof{(\laml{\loc}{\var^\typ}{\progtwo})\,\val} \alt \prog_2
    }{\asg}
    \\
    & = &
    \set{
      \obj
    \ST
      \vec{\objtwo} \in \semtyp{\vec{\typtwo}},
      \vec{\objthree} \in \semtyp{\vec{\typthree}},
      \obj \in
      \seme{\prog_1 \alt \wctxof{(\laml{\loc}{\var^\typ}{\progtwo})\,\val} \alt \prog_2}{
        \asg\asgextend{\vec{\vartwo}}{\vec{\objtwo}}\asgextend{\vec{\varthree}}{\vec{\objthree}}
      }
    }
    \\
    & = &
    \set{
      \obj
    \ST
      \vec{\objtwo} \in \semtyp{\vec{\typtwo}},
      \obj \in
      \seme{\prog_1 \alt \wctxof{\progtwo\sub{\var^\typ}{\val}} \alt \prog_2}{
        \asg\asgextend{\vec{\vartwo}}{\vec{\objtwo}}
      }
    }
    & (\star)
    \\
    & = &
    \semf[\vec{\vartwo}^{\vec{\typtwo}}]{\prog_1 \alt \wctxof{\progtwo\sub{\var^\typ}{\val}} \alt \prog_2}{\asg}
    \end{array}
  \]
  To justify $(\star)$
  we proceed as follows.
  Let us write
  $\asg'$
  for
  $\asg\asgextend{\vec{\vartwo}}{\vec{\objtwo}}$.
  Recall that the interpretation of a value
  is always a singleton~(\rlem{interpretation_of_values}),
  so let $\seme{\val}{\asg'\asgextend{\vec{\varthree}}{\vec{\objthree}}} = \set{\obj_0}$.
  By Compositionality~(\rlem{compositionality})
  it suffices to note that:
  \[
    \begin{array}{rcll}
      \seme{(\laml{\loc}{\var^\typ}{\progtwo})\,\val}{\asg'\asgextend{\vec{\varthree}}{\vec{\objthree}}}
    & = &
      \set{\objtwo \ST
        \obj \in \seme{\val}{\asg'\asgextend{\vec{\varthree}}{\vec{\objthree}}},
        \objtwo \in \seme{\progtwo}{\asg'\asgextend{\vec{\varthree}}{\vec{\objthree}}\asgextend{\var^\typ}{\obj}}
      }
    \\
    & = &
      \seme{\progtwo}{\asg'\asgextend{\vec{\varthree}}{\vec{\objthree}}\asgextend{\var^\typ}{\obj_0}}
    \\
    & = &
      \seme{\progtwo}{\asg'\asgextend{\var^\typ}{\obj_0}}
      & \text{(By Irrelevance~(\rlem{irrelevance}))}
    \\
    & = &
      \seme{\progtwo\sub{\var^\typ}{\val}}{\asg'}
      & \text{(By \rlem{interpretation_of_substitution})}
    \end{array}
  \]
\item \rulename{guard}:
  Note that $\fctx = \fctx',\vec{\varthree}^{\vec{\typthree}}$, where
  $\vec{\varthree}^{\vec{\typthree}} = \fv{\val} \setminus \fv{\prog_1 \alt \wctxof{\ctxhole \seq \tm} \alt \prog_2}$.
  Suppose that $\fctx' = \vec{\vartwo}^{\vec{\typtwo}}$.
  Then:
  \[
    \begin{array}{rcll}
      \semf[\vec{\vartwo}^{\vec{\typtwo}},\vec{\varthree}^{\vec{\typthree}}]{\prog_1 \alt \wctxof{\val \seq \tm} \alt \prog_2}{\asg}
    & = &
      \set{\obj \ST
            \vec{\objtwo} \in \seme{\vec{\typtwo}},
            \vec{\objthree} \in \seme{\vec{\typthree}},
            \obj \in \seme{\prog_1 \alt \wctxof{\val \seq \tm} \alt \prog_2}
                          {\asg\asgextend{\vec{\vartwo}}{\vec{\objtwo}}\asgextend{\vec{\varthree}}{\vec{\objthree}}}}
    \\
    & = &
      \set{\obj \ST
            \vec{\objtwo} \in \seme{\vec{\typtwo}},
            \obj \in \seme{\prog_1 \alt \wctxof{\tm} \alt \prog_2}{\asg\asgextend{\vec{\vartwo}}{\vec{\objtwo}}}}
    & (\star)
    \\
    & = &
      \semf[\vec{\vartwo}^{\vec{\typtwo}}]{\prog_1 \alt \wctxof{\tm} \alt \prog_2}{\asg}
    \end{array}
  \]
  To justify $(\star)$ we proceed as follows.
  Let us write
  $\asg'$
  for
  $\asg\asgextend{\vec{\vartwo}}{\vec{\objtwo}}$.
  Recall that the interpretation of a value
  is always a singleton~(\rlem{interpretation_of_values}),
  so let $\seme{\val}{\asg'\asgextend{\vec{\varthree}}{\vec{\objthree}}} = \set{\objtwo_0}$.
  By Compositionality~(\rlem{compositionality})
  it suffices to note that:
  \[
    \begin{array}{rcll}
      \seme{\val \seq \tm}{\asg'\asgextend{\vec{\varthree}}{\objthree}}
    & = &
      \set{\obj \ST
            \objtwo \in \seme{\val}{\asg'\asgextend{\vec{\varthree}}{\vec{\objthree}}},
            \obj \in \seme{\tm}{\asg'\asgextend{\vec{\varthree}}{\vec{\objthree}}}}
    & \text{(By Irrelevance~\rlem{irrelevance})}
    \\
    & = &
      \seme{\tm}{\asg'}
    \end{array}
  \]
\item \rulename{fresh}:
  Note that $\fctx' = \fctx,\vartwo^{\typ}$
  where $\vartwo$ is a fresh variable.
  Suppose that
  $\fctx = \vec{\varthree}^{\vec{\typtwo}}$.
  Then:
  \[
    \begin{array}{rcll}
      \semf{\prog_1 \alt \wctxof{\fresh{\var^\typ}{\tm}} \alt \prog_2}{\asg}
    & = &
      \set{
        \obj
      \ST
        \vec{\objtwo} \in \semtyp{\vec{\typtwo}},
        \obj \in
        \seme{
          \prog_1 \alt \wctxof{\fresh{\var^\typ}{\tm}} \alt \prog_2
        }{
          \asg\asgextend{\vec{\varthree}}{\vec{\objtwo}}
        }
      }
    \\
    & = &
      \set{
        \obj
      \ST
        \vec{\objtwo} \in \semtyp{\vec{\typtwo}},
        \obj \in
        \semf[\vartwo^\typ]{
          \prog_1 \alt \wctxof{\tm\sub{\var^\typ}{\vartwo^\typ}} \alt \prog_2
        }{
          \asg\asgextend{\vec{\varthree}}{\vec{\objtwo}}
        }
      }
      & (\star)
    \\
    & = &
      \semf[\fctx,\vartwo^\typ]{\prog_1 \alt \wctxof{\tm\sub{\var^\typ}{\vartwo^\typ}} \alt \prog_2}{\asg}
    \end{array}
  \]
  To justify $(\star)$ we proceed
  as follows.
  Let $\asg'$ stand for
  $\asg\asgextend{\vec{\varthree}}{\vec{\objtwo}}$.
  By Irrelevance~(\rlem{irrelevance}),
  $\seme{\prog_1}{\asg'} = \semf[\vartwo^\typ]{\prog_1}{\asg'}$.
  Similarly,
  $\seme{\prog_2}{\asg'} = \semf[\vartwo^\typ]{\prog_2}{\asg'}$.
  By Compositionality~(\rlem{compositionality}), it suffices to show that
  $\seme{\wctxof{\fresh{\var^\typ}{\tm}}}{\asg'}
  = \semf[\vartwo^\typ]{\wctxof{\tm\sub{\var^\typ}{\vartwo^\typ}}}{\asg'}$. Indeed:
  \[
    \begin{array}{rcll}
    &&
      \seme{\wctxof{\fresh{\var^\typ}{\tm}}}{\asg'}
    \\
    & = &
      \set{\objthree \ST
            \objtwo \in \seme{\fresh{\var^\typ}{\tm}}{\asg'},
            \objthree \in \seme{\wctx}{\asg'\asgextend{\ctxhole}{\objtwo}}
      }
    & \text{(By~\rlem{compositionality})}
    \\
    & = &
      \set{\objthree \ST
            \obj \in \semtyp{\typ},
            \objtwo \in \seme{\tm}{\asg'\asgextend{\var^\typ}{\obj}},
            \objthree \in \seme{\wctx}{\asg'\asgextend{\ctxhole}{\objtwo}}
      }
    \\
    & = &
      \set{\objthree \ST
            \obj \in \semtyp{\typ},
            \objtwo \in \seme{\tm\sub{\var^\typ}{\vartwo^\typ}}{\asg'\asgextend{\vartwo^\typ}{\obj}},
            \objthree \in \seme{\wctx}{\asg'\asgextend{\ctxhole}{\objtwo}}
      }
      & \text{(By~\rlem{interpretation_of_substitution} and \rlem{irrelevance})}
    \\
    & = &
      \set{\objthree \ST
            \obj \in \semtyp{\typ},
            \objtwo \in \seme{\tm\sub{\var^\typ}{\vartwo^\typ}}{\asg'\asgextend{\vartwo^\typ}{\obj}},
            \objthree \in \seme{\wctx}{\asg'\asgextend{\vartwo^\typ}{\obj}\asgextend{\ctxhole}{\objtwo}}
      }
      & \text{(By \rlem{irrelevance})}
    \\
    & = &
      \set{\objthree \ST
            \obj \in \semtyp{\typ},
            \objthree \in \seme{\wctxof{\tm\sub{\var^\typ}{\vartwo^\typ}}}{\asg'\asgextend{\vartwo^\typ}{\obj}}
      }
    & \text{(By~\rlem{compositionality})}
    \\
    & = &
      \semf[\vartwo^\typ]{\wctxof{\tm\sub{\var^\typ}{\vartwo^\typ}}}{\asg'}
    \end{array}
  \]
\item \rulename{unif}:
  Our goal is to prove that
  $\semf[\fctx]{\prog_1 \alt \wctxof{\val \unif \valtwo} \alt \prog_2}{\asg}
  =
  \semf[\fctx']{\prog_1 \alt \wctxof{\unit}\SUB{\subst} \alt \prog_2}{\asg}$,
  where $\subst = \mgu{\set{\val\unif\valtwo}}$.
  Note that $\fctx'$ is a subset of $\fctx$,
  so suppose that $\fctx = \fctx',\vec{\vartwo}^{\vec{\typtwo}}$
  and $\fctx' = \vec{\var}^{\vec{\typ}}$.
  Note also that $\subst = \mgu{\val \unif \valtwo}$ exists,
  so
  $\set{\val \unif \valtwo} \tounifa{}^*
   \set{\var_1 \unif \val_1, \hdots, \var_n \unif \val_n}$
  such that $\var_i \notin \fv{\val_j}$ for all $i,j$,
  and
  the most general unifier is
  $\subst = \set{\var_1 \mapsto \val_1, \hdots, \var_n \mapsto \val_n}$.
  Moreover, recall that the interpretation of a value
  is always a singleton~(\rlem{interpretation_of_values}),
  so for each fixed assignment $\asg'$ let us write
  $\objtwo^{\asg'}_{i}$ for the only element in $\seme{\val_i}{\asg'}$.
  Moreover, let
  $\vec{\varthree}^{\vec{\typthree}} = \vec{\var}^{\vec{\typ}},\vec{\vartwo}^{\vec{\typtwo}}$.
  By Compositionality~(\rlem{compositionality}) and Irrelevance~(\rlem{irrelevance}),
  it suffices to note that:
  \[
    \begin{array}{rcll}
    &&
    \semf[\vec{\varthree}^{\vec{\typthree}}]{\wctxof{\val \unif \valtwo}}{\asg}
    \\
    & = &
    \set{\obj \ST
          \vec{\objthree} \in \semtyp{\vec{\typthree}},
          \obj \in \seme{\wctxof{\val \unif \valtwo}}{\asg\asgextend{\vec{\varthree}}{\vec{\objthree}}}
    }
    \\
    & = &
    \set{\obj \ST
          \vec{\objthree} \in \semtyp{\vec{\typthree}},
          \objtwo \in \seme{\val \unif \valtwo}{\asg\asgextend{\vec{\varthree}}{\vec{\objthree}}},
          \obj \in \seme{\wctx}{\asg\asgextend{\vec{\varthree}}{\vec{\objthree}}\asgextend{\ctxhole}{\objtwo}}
    }
    & \text{(By~\rlem{compositionality})}
    \\
    & = &
    \set{\obj \ST
          \vec{\objthree} \in \semtyp{\vec{\typthree}},
          \objtwo \in \seme{\val \unif \valtwo}{\asg\asgextend{\vec{\varthree}}{\vec{\objthree}}},
          \obj \in \seme{\wctx}{\asg\asgextend{\vec{\varthree}}{\vec{\objthree}}\asgextend{\ctxhole}{\cinterp{\unit}}}
    }
    \\
    & = &
    \set{\obj \ST
          \vec{\objthree} \in \semtyp{\vec{\typthree}},
          \vec{\objthree} \vDash_{\asg,\vec{\varthree}} \set{\val \unif \valtwo},
          \obj \in \seme{\wctx}{\asg\asgextend{\vec{\varthree}}{\vec{\objthree}}\asgextend{\ctxhole}{\cinterp{\unit}}}
    }
    \\
    & = &
    \set{\obj \ST
          \vec{\objthree} \in \semtyp{\vec{\typthree}},
          \vec{\objthree} \vDash_{\asg,\vec{\varthree}} \set{\var_1 \unif \val_1, \hdots, \var_n \unif \val_n},
          \obj \in \seme{\wctx}{\asg\asgextend{\vec{\varthree}}{\vec{\objthree}}\asgextend{\ctxhole}{\cinterp{\unit}}}
    }
    & \text{(By \rlem{unification_preserves_satisfaction})}
    \\
    & = &
    \set{\obj \ST
          \vec{\objthree} \in \semtyp{\vec{\typthree}},
          \asg\asgextend{\vec{\varthree}}{\vec{\objthree}}(\var_i)
          =
          \objtwo^{\asg\asgextend{\vec{\varthree}}{\vec{\objthree}}}_i
          \text{ for all $i$},\ %
          \obj \in \seme{\wctxof{\unit}}{\asg\asgextend{\vec{\varthree}}{\vec{\objthree}}}
    }
    & \text{(By~\rlem{compositionality})}
    \\
    & = &
    \set{\obj \ST
          \vec{\objthree} \in \semtyp{\vec{\typthree}},
          \obj \in \seme{\wctxof{\unit}}{\asg\asgextend{\vec{\varthree}}{\vec{\objthree}}
                                             \asgextend{\var_1}{\objtwo^{\asg\asgextend{\vec{\varthree}}{\vec{\objthree}}}_1}\hdots
                                             \asgextend{\var_n}{\objtwo^{\asg\asgextend{\vec{\varthree}}{\vec{\objthree}}}_n}}
    }
    & \text{($\star$)}
    \\
    & = &
      \set{\obj \ST
            \vec{\objthree} \in \semtyp{\vec{\typthree}},
            \obj \in \seme{\wctxof{\unit}\SUB{\subst}}{\asg\asgextend{\vec{\varthree}}{\vec{\objthree}}}
      }
    & \text{(By~\rlem{interpretation_of_substitution})}
    \\
    & = &
      \semf[\vec{\varthree}^{\vec{\typthree}}]{\wctxof{\unit}\SUB{\subst}}{\asg}
    \\
    & = &
      \semf[\vec{\var}^{\vec{\typ}}]{\wctxof{\unit}\SUB{\subst}}{\asg}
      & \text{(By ~\rlem{irrelevance})}
    \end{array}
  \]
  To justify $(\star)$
  note that
  $\asg\asgextend{\vec{\varthree}}{\vec{\objthree}}(\var_i) = \set{\objtwo^{\asg\asgextend{\vec{\varthree}}{\vec{\objthree}}}_i}$
  for all $i = 1..n$.
  Therefore, we can write
  $\asg\asgextend{\vec{\varthree}}{\vec{\objthree}}$ as
  $\asg\asgextend{\vec{\varthree}}{\vec{\objthree}}\asgextend{\var_1}{\objtwo^{\asg\asgextend{\vec{\varthree}}{\vec{\objthree}}}_1}\hdots\asgextend{\var_n}{\objtwo^{\asg\asgextend{\vec{\varthree}}{\vec{\objthree}}}_n}$.
\item \rulename{fail}:
  Our goal is to prove that:
  \[
    \semf[\fctx]{\prog_1 \alt \wctxof{\val \unif \valtwo} \alt \prog_2}{\asg}
  \supseteq
    \semf[\fctx']{\prog_1 \alt \prog_2}{\asg}
  \]
  which is immediate by definition.
\end{enumerate}
\end{proof}

%
%
%
\bibliographystyle{splncs04}
\bibliography{biblio}

\end{document}